\documentclass[aps,
superscriptaddress,
twocolumn]{revtex4}

\usepackage{latexsym,amsfonts,amssymb,exscale,enumerate,comment}
\usepackage{amsmath,amsthm,amscd}
\usepackage[all,knot]{xy}

\usepackage{ braket}
\usepackage[caption=false]{subfig}

\usepackage{url}
\usepackage[bookmarks=true,%
    colorlinks=true,%
    linkcolor=blue,%
    citecolor=blue,%
    filecolor=blue,%
    menucolor=blue,%
    urlcolor=blue,%
    breaklinks=true]{hyperref}

\input xy
\usepackage[all]{xy}
\SelectTips{cm}{}

\usepackage{tikz}
\usepackage{tikz-cd} 
\usetikzlibrary{shapes,snakes}
\usetikzlibrary{decorations.markings}
\usetikzlibrary{decorations.pathreplacing}
\tikzstyle directed=[postaction={decorate,decoration={markings,
    mark=at position #1 with {\arrow{>}}}}]

\newcommand{\hackcenter}[1]{
 \xy (0,0)*{#1}; \endxy}

\tikzset{->-/.style={decoration={
  markings,
  mark=at position #1 with {\arrow{>}}},postaction={decorate}}}

\usepackage{graphicx}
\usepackage{color}
\newcommand{\brk}[1]{{\left\langle{#1}\right\rangle}}
\newcommand{\bp}[1]{{\left({#1}\right)}}

\newcommand{\md}{\operatorname{\mathsf{d}}}

\newcommand{\Gr}{\mathcal{G}}
\newcommand{\X}{\mathcal{X}}
\newcommand{\col}{{\Phi}}
\newcommand{\HHH}{\ensuremath{\mathcal{H}}}
\newcommand{\St}{\operatorname{St}}

\newcommand{\I}{\mathcal{I}}
\newcommand{\m}{\mathsf{m}}
\newcommand{\mirror}[1]{{\stackrel{\leftrightarrow}{#1}}}
\def\Y{{\mathcal{Y}}}

\theoremstyle{plain}
\newtheorem{theorem}{Theorem}

\newtheorem{corollary}[theorem]{Corollary}
\newtheorem{proposition}[theorem]{Proposition}
\newtheorem{lemma}[theorem]{Lemma}

\theoremstyle{definition}

\theoremstyle{definition}
\newtheorem{remark}[theorem]{Remark}

\numberwithin{equation}{section}
\numberwithin{theorem}{section}

\newcommand{\maps}{\colon}

\newcommand{\refequal}[1]{\xy {\ar@{=}^{#1}
(-1,0)*{};(1,0)*{}};
\endxy}

\newcommand{\Hom}{{\rm Hom}}

\renewcommand{\to}{\rightarrow}

\def\dmod{{\mathrm{-mod}}}   

\def\Id{\mathrm{Id}}

\def\T{{\mathcal{T}}}

\numberwithin{equation}{section}

\newcommand{\wb}{\overline}

\newcommand{\slt}{{\mathfrak{sl}(2)}}
\newcommand{\Uq}{{U_q\slt}}
\newcommand{\UqMed}{{\wb U_q\slt}}

\newcommand{\UsltH}{{U_q^{H}\slt}}
\newcommand{\Ubar}{{\wb U_q^{H}\slt}}

\let\tilde=\widetilde

\let\epsilon=\varepsilon

\usepackage{bbm}
\def\C{{\mathbb{C}}}

\def\R{{\mathbbm R}}
\def\Z{{\mathbbm Z}}
\def\H{{\mathcal{H}}}

\def\cal#1{\mathcal{#1}}%
\def\1{\mathbbm{1}}%
\def\nn{\notag}

\def\la{\langle}
\def\ra{\rangle}

\newcommand{\A}{{\sf A}}

\newcommand{\bb}{{\sf b}}
\newcommand{\G}{\cal{G}}

\usepackage{bbm}

\def\cal#1{\mathcal{#1}}

\newcommand\nc{\newcommand}
\nc\rnc{\renewcommand}
\nc\Kar{\operatorname{Kar}}
\nc\End{\operatorname{End}}

\newcommand{\scs}{\scriptstyle}

\nc\Sym{\operatorname{Sym}}

\allowdisplaybreaks

\begin{document}
\title {Pseudo-Hermitian Levin-Wen models from non-semisimple TQFTs}


\author{Nathan Geer}
\affiliation{Mathematics \& Statistics,
  Utah State University,
  Logan, Utah 84322, USA}%
  \email{nathan.geer@gmail.com}
\author{Aaron D. Lauda}
\affiliation{Department of Mathematics,
 University of Southern California ,
  Los Angeles, California 90089, USA}
  \email{lauda@usc.edu}%

%
\author{Bertrand Patureau-Mirand}
\affiliation{UMR 6205, LMBA, universit\'e de Bretagne-Sud,
  BP 573, 56017 Vannes, France }
\email{bertrand.patureau@univ-ubs.fr}
\author{Joshua Sussan}
\affiliation{Department of Mathematics,
  CUNY Medgar Evers,
  Brooklyn, NY 11225, USA}
  \email{jsussan@mec.cuny.edu}
  \affiliation{Mathematics Program,
 The Graduate Center, CUNY,
  New York, NY 10016, USA}
  \email{jsussan@gc.cuny.edu}

\begin{abstract}
We construct large classes of exactly solvable pseudo-Hermitian 2D spin Hamiltonians.  The ground states of these systems depend only on the spatial topology of the system.  We identify the ground state system on a surface with the value assigned to the surface by a non-semisimple TQFT generalizing the Turaev-Viro model. A non-trivial example arises from a non-semisimple subcategory of representations of quantum $sl(2)$  where the quantum parameter is specialized to a root of unity.
  \end{abstract}

\maketitle
\setcounter{tocdepth}{3}

 \section{ Introduction}

There is a well developed correspondence connecting  topologically ordered phases in 2-dimensional systems with the theory of unitary modular tensor categories (UMTCs).  Many of the key properties of these systems, such as  topological ground-state degeneracy and non-abelian braiding statistics of low-energy point-like excitations, make them candidates for topological quantum computing.  These features can be formalized through UMTCs and their associated 2+1-dimensional topological quantum field theories (TQFTs)~\cite{Kit,GTKLT,FKLW,LW,MR2608953}.

Levin and Wen constructed an important class of doubled, or time reverse symmetric, 2-dimensional phases utilizing a discretized model defined on a trivalent graph. The graph need not be planar and can be regarded as living on a genus $g$ surface.    Their `string-net' model produces exactly solvable gapped Hamiltonians where the low-energy physics of a 2-dimensional topologically ordered phase is described by the condensation of string-like objects~\cite{LW}.  These models have been further generalized by modifying various assumptions \cite{KitK,PhysRevB.89.195130,HW}.  The result is the general belief that the most general class of bosonic topological orders with gapped boundaries are described by the Drinfeld center of a fusion category.  Indeed, in \cite{PhysRevB.103.195155} the generalized string-net models producing exactly solvable Hermitian Hamiltonians are classified.

\medskip

In this work, we extend the paradigm of bosonic topological order by introducing a new class of theories not arising as the Drinfeld center of a fusion category.   A key difference in our approach is that, while the Hamiltonians in our systems have positive spectrum and normalizable wave functions, these Hamiltonians are not Hermitian.  They are however \emph{pseudo-Hermitian}, so that time evolution is given by the exponential of the Hamiltonian that is self adjoint with respect to an indefinite inner product.  These inner products are also invariant under the time-evolution generated by the Schrodinger equation.

Quantum mechanics with an indefinite inner product has been studied going back to Dirac~\cite{dirac} and Pauli~\cite{Pauli}.  Even with the indefinite norms, they observed a formalism consistent with deterministic quantum mechanics.  More recently, there has been a resurgence of interest in indefinite quantum mechanics coming from the study of $\cal{P}\cal{T}$-symmetric quantum theory~\cite{BendB}.
$\cal{P}\cal{T}$-symmetric quantum theory removes the assumption that the Hamiltonian is Hermitian and replaces it by the condition that the Hamiltonian commutes with the antilinear operator $\cal{P}\cal{T}$ of parity and time reversal, so that $[H, \cal{P}\cal{T}]=0$. Hermiticity is then replaced by the condition $H = H^{\cal{P}\cal{T}}$  (see \cite{Bender_2005} for a good introduction).
$\cal{P}\cal{T}$-symmetric quantum theory was proposed as a way to measure physical phenomena in the absence of Hermitian Hamiltonians.
Indeed, non-Hermitian $\cal{P}\cal{T}$-symmetric Hamiltonians have
already been used to describe such phenomena as the ground state of a quantum system of
hard spheres~\cite{PhysRev.115.1390}, Reggeon field theory~\cite{BROWER1978213}, and the Lee-Yang edge singularity~\cite{ZAMOLODCHIKOV1991619}.
 In each of these examples, the Hamiltonians have spectral positivity and the associated quantum theories are unitary because the Hamiltonians are $\cal{P}\cal{T}$-symmetric.

It was discovered that $\cal{P}\cal{T}$-symmetric quantum theory was not the most general criteria that would ensure that a given non-Hermitian Hamiltonian would have a real spectrum and unitary evolution~\cite{Mostafazadeh_2002a, Mostafazadeh_2002,Most-III}.  Rather, the notion of \emph{pseudo-Hermiticity} provides such a criteria.   Any diagonalizable Hamiltonian admitting a symmetry generated by an invertible antilinear operator is pseudo-Hermitian, including the $\cal{P}\cal{T}$ operator.
A linear operator $H \maps \cal{H} \to \cal{H}$ acting on a Hilbert space $\cal{H} = (\cal{H}, \la \cdot, \cdot \ra_+)$ (here we are assuming positive-definite inner product $\la \cdot, \cdot \ra_+$)  is called \emph{pseudo-Hermitian}, or $\eta$-pseudo-Hermitian, if there exists a linear, invertible, Hermitian operator $\eta \maps \cal{H} \to \cal{H}$ such that
\[
H^{\sharp} = \eta H \eta^{-1} ,
\]
where $H^{\sharp}$ above is the usual Hermitian conjugate $\la H^{\sharp} \psi,   \phi \ra_+ = \la \psi,  H\phi \ra_+$   determined by the positive-definite form $\la \cdot, \cdot \ra_+$.  We use this nonstandard notation since we will be primarily interested in a different inner product and the Hermitian adjoint with respect to that form.

If $H$ is pseudo-Hermitian, the choice of such $\eta$ is not unique.  Each choice of $\eta$ determines a possibly indefinite inner product, or \emph{pseudo-inner product}, on $\cal{H}$ given by
\begin{equation}
  \la \psi, \phi \ra  := \la \psi, \eta \phi \ra_+ ~.
\end{equation}
When a Hamiltonian $H$ is $\eta$-pseudo-Hermitian, then it becomes an actual Hermitian operator with respect to the indefinite form $\la , \ra$.  We write $H^{\dagger} = H$ to mean
\[
\la \psi, H \phi \ra = \la H\psi, \phi \ra.
\]
If $H$ is pseudo-Hermitian, exponentiating $iH$ produces an operator $U$ satisfying
\[
U^{\sharp} = \eta U^{-1}\eta^{-1}.
\]
Such an operator is called pseudo-unitary (see ~\cite{Most-unitary}).
The group of such operators is controlled by the group $U(n,m)$ with $n,m \in \Z^{+}$ determined from the signature of the inner product.
\medskip

In this article, we show that $\eta$-pseudo-Hermitian Hamiltonians extending the Levin-Wen models arise naturally from recent non-semisimple modifications of the theory of 2+1-dimensional state-sum TQFTs.   As argued in \cite{Most-Is}, the key distinction between indefinite metric quantum mechanics and pseudo-Hermitian quantum mechanics is that the choice of $\eta$ is additional data that is fixed ahead of time.  Here we show that such an $\eta$ arises naturally from topological considerations.

A key observation from ~\cite{KKR} and \cite{Kir-stringnet,BalKir,Balsam}  establishes a link between the exactly solvable Levin-Wen Hamiltonians defined on trivalent graphs on a surface $\Sigma$  and the Turaev-Viro topological quantum field theory in three spatial dimensions.  In this interpretation, the plaquette operators used to define the Hamiltonian  arise from three-dimensional tetrahedra glued onto the triangulated surface of the model.  Key properties of these plaquette operators, such as being projectors and   mutually commuting, also have natural topological interpretations as change of triangulations in 2+1-dimensions, (see Section~\ref{sec:top} for more details).    The projection onto the ground state is then given by the image of the operator assigned to $\Sigma \times [0,1]$ by the Turaev-Viro TQFT.

More recently, there have been developments in the study of topological quantum field theories, so called, non-semisimple TQFTs that live outside the usual unitary modular tensor category framework. `Non-semisimple' refers to the fact that they are built on tensor categories that do not satisfy the usual semisimple assumptions prevalent in nearly all categorical descriptions of topological phases.  These TQFTs are governed by relative $\Gr$-spherical categories
 and depend on extra data including a Hamiltonian link in the 3-manifold $M$ and a cohomology class $[\Phi]\in H_1(M,\Gr)$.
The key examples of such non-semisimple categories  have an infinite number of nonisomorphic simple objects, all having vanishing quantum dimensions.    Nevertheless, these non-semisimple TQFTs have remarkable properties, often proving more powerful than their semisimple analogs. For example, non-semisimple TQFTs lead to mapping class group representations with the notable property that the action of a Dehn twist has infinite order, and thus the representation could be faithful, (see \cite{BCGP2}). This is in contrast with the usual quantum mapping class group representations where all Dehn twists have finite order and the representations are not faithful.  Also, after projectivization, these TQFTs correspond to the Lyubashenko projective mapping class group representations given in
\cite{Lyubashenko:1994tm}, (also see \cite{DGP2,derenzi2021mapping}).  Related work of Chang~\cite{Chang} considers non-semisimple generalization of Turaev-Viro TQFTs and their lattice model realizations based on non-semisimple quantum groupoids.


In this article we define pseudo-Hermitian Levin-Wen models from relative $\Gr$-spherical categories satisfying certain Hermitian properties.  Any unitary modular tensor category provides an  example of a relative $\Gr$-spherical category satisfying our assumptions, where $\Gr$ is the trivial group.  We give an example with nontrivial $\Gr$ in Section \ref{sec:coeff} coming from the non-semisimple representation theory of quantum $\slt$ at a root of unity.

\subsection{Outline}
Section \ref{sec:system} contains the definition of a relative $\Gr$-spherical category along with extra data needed to construct a pseudo-Hermitian Hamiltonian from it.
We give a non-semisimple version of the Levin-Wen construction in Section \ref{sec:hamiltonian}.  Various properties of the Hamiltonian are proved here.
In Section \ref{sec:bistellar}, we construct a pseudo-unitary operator on our Hilbert space, which may be of independent interest.
Finally in Section \ref{sec:coeff} we provide an example of a relative $\Gr$-spherical category satisfying the extra data required to define a Levin-Wen model.  This category comes from quantum $\slt$, but is not semisimple.  This is a point of departure from other papers in the area.  The appendix contains some explicit formulas for morphisms in this category.

\subsection{Acknowledgements}
The authors are grateful to Sergei Gukov and  Zhenghan Wang for helpful comments on a preliminary version of this article.
N.G.\ is partially supported by NSF grants DMS-1664387 and DMS-2104497.
A.D.L.\ is partially supported by NSF grant DMS-1902092 and Army Research Office W911NF-20-1-0075.
J.S.\ is partially supported by the NSF grant DMS-1807161 and PSC CUNY Award 64012-00 52.

 \section{System Hilbert space} \label{sec:system}

The input for defining the system Hilbert space is a triangulated surface $\Sigma$ and certain categorical data organized in the notion of a  \emph{relative $\Gr$-spherical category}.
 The main examples of relative spherical categories are the categories of
finite-dimensional weight modules over semi-restricted quantum groups.  In contrast to the usual modular tensor categories used to study topological phases, these categories are not semisimple and have an infinite number of nonisomorphic irreducible modules, all having vanishing quantum dimensions.   Nevertheless, such categories have been shown to give rise to nonabelian braiding statistics and lead to new TQFTs.

We make some additional assumptions on the relative $\Gr$-spherical category.  We list these conditions below as `Hermitian structure'.  They are motivated by a Hermitian structure on the non-semisimple category of modules for unrolled quantum $\slt$ constructed in \cite{GLPMS} and recalled in Section \ref{sec:coeff} and the appendix.
We also assume the multiplicity spaces between generic string types are one dimensional  and that the group $\Gr$ is abelian.  It should be relatively straightforward to remove the first condition, but it is less clear how to handle non-abelian groups.
%

\subsection{Basic Input} \label{subsec:basic}

The categorical input described above is the following data.
\smallskip

\noindent  {\bf String data:}
\begin{itemize}
 \item $\Gr$ is a
 abelian group
 with identity $0$ containing  a small subset $\X\subset\Gr$, such that $\X$ is symmetric ($-\X= \X$).
 \item For each $g\in \Gr$ there is a finite set of string labels $I_g$.
\end{itemize}

Set $I=\sqcup I_g$ and  $\A=\cup_{g\in\Gr\setminus \X} I_g$.  We say a string type is  \emph{generic} if it is in $\A$.   If $i\in I_g$ we say the degree of $i$ is $g$ and write $\deg(i)=g$.   Each element  $j\in I_g$ has a \emph{conjugate} string type $j^{\ast}\in I_{-g^{}}$, which satisfies $j^{**}=j$. There is unique \emph{vacuum} string type
  $j=0$ in $I_{0}$ satisfying $0^{\ast}=0$.
\smallskip

\noindent  {\bf Branching Rules:}
 To each triple of strings $i, j, k\in \A$, we associate a branching rule $\delta_{ijk}$ that equals $1$
if the triple is allowed to meet at a vertex  and $0$ otherwise.
Here we consider the triple  $\{i,j,k\}$ up to cyclic ordering and require  $\delta_{ijk}$ is
symmetric under cyclic permutations of the three labels:
$\delta_{ijk}=\delta_{jki}=\delta_{kij}$.  To be
compatible with the conjugation structure of labels, the branching
rule satisfy $\delta_{0jj^{\ast}}=\delta_{0j^{\ast} j}=1$,
$\delta_{ijk}=\delta_{k^{\ast}j^{\ast}i^{\ast}}$ and $\delta_{ijk}=0$ if  $\deg(i) + \deg(j)+ \deg(k)\neq 0 \in \Gr$.
\smallskip

\noindent  {\bf Modified dimensions:}
 There exist functions $\md:\A\rightarrow \R^{\times}$ and  $\bb:\A\to \R$ satisfying $\md(i^{\ast}) = \md(i)$, $\bb(i^{\ast}) = \bb(i)$, and for $g,g_1, g_2\in \Gr\setminus \X$ with $g+g_1+g_2=0$ we have
 \begin{align} \label{eq:b}
   \bb(j)=\sum_{j_1\in \I_{g_1},\, j_2 \in \I_{g_2}}   \bb(j_1)\bb(j_2)\delta_{j^{\ast}j_1 j_2}
 \end{align}
 for all $j\in I_g$.
\smallskip

\noindent  {\bf Modified $6j$ symbols:}
There are symmetrized $6j$ symbols $N^{i j k}_{l m n}$  defined to be zero unless
\begin{align} \label{eq:6jadmissible}
  &\delta_{ij k^{\ast}}=\delta_{i^{\ast}m n^{\ast}}=\delta_{j^{\ast}nl^{\ast}}=\delta_{klm^{\ast}}=1.
\end{align}
These symbols are required to satisfy
\begin{align}
&N^{i j k}_{l m n} = N^{ j k^{\ast}  i^{\ast}}_{m n l}
= N^{k l m}_{n^{\ast} i j^{\ast}} ; \label{eq:symm} \\
&\sum_j \md(j) N^{j_1 j_2 j_5}_{j_3 j_6 j} N^{j_1 j j_6}_{j_4 j_0 j_7} N^{j_2 j_3 j}_{j_4 j_7 j_8}
=
N^{j_5 j_3 j_6}_{j_4 j_0 j_8} N^{j_1 j_2 j_5}_{j_8 j_0 j_7} ; \label{eq:pent} \\
&\sum_n \md(n) N^{i j p}_{l m n} N^{k j^{\ast} i}_{n m l}
= \frac{\delta_{k,p}}{\md(k)} \delta_{ij k^{\ast}} \delta_{k l m^{\ast}}
 . \label{eq-ortho}
\end{align}
We refer to these identities as \emph{tetrahedral symmetry}, \emph{pentagon identity}, and \emph{orthogonality} respectively.
\smallskip

\noindent  {\bf  Hermitian Structure:}
There exist maps $\beta: \A\to\R^{\times}$ and
$\gamma:\A^3\to\R^{\times}$ such that $\gamma$ is invariant
under cyclic permutations of its three arguments.  Furthermore,
\begin{align}
&\beta(j^{\ast})=\beta(j); \\
  &\gamma(i,j,k)\gamma(k^{\ast},j^{\ast},i^{\ast})\beta(i)\beta(j)\beta(k)=1
    \text{ if }\delta_{ijk}=1;
    \label{eq:gamma-theta} \\
& \wb{(N^{j_1 j_2 j_3 }_{j_4j_5j_6})}{=}N^{j_2^{\ast} j_1^{\ast} j_3^{\ast}}_{j_5 j_4j_6}\gamma(j_1,j_2,j_3^{\ast})\gamma(j_1^{\ast},j_5,j_6^{\ast})
 \nn\\
& \qquad \qquad \;\; \times \gamma(j_2^{\ast},j_6,j_4^{\ast})\gamma(j_3,j_4,j_5^{\ast})\prod_{i=1}^6\beta(j_i) ,\label{eq:6j}
\end{align}
where all string types are generic in the above formulas.

\subsection{Definition of the state space} \label{subsec:def-state}
Let $\Sigma$ be a compact, connected, oriented surface.
Let $\T$ be a triangulation of $\Sigma$ and $\Gamma$ be a finite
trivalent graph dual to $\T$.  Each vertex of the graph
$\Gamma$ acquires a cyclic ordering compatible with the orientation of
$\Sigma$.  A \emph{$\Gr$-coloring} of $\Gamma$ is 
a map $\col$ from the set of oriented edges
of $\Gamma$ to $\Gr$ such that
\begin{enumerate}
\item $\col(-e)=-\col(e)^{}$ for any oriented edge $e$ of $\Gamma$, where
  $-e$ is $e$ with opposite orientation, and
\item if $e_1, e_2, e_3$ are edges of a vertex $v$ of $\Gamma$ with a
  cyclic ordering compatible with the orientation of $\Sigma$ and each
  edge is oriented towards the vertex $v$, then
  $\col(e_1)+ \col(e_2)+ \col(e_3)=0$.
\end{enumerate}
The $\Gr$-colorings of $\Gamma$ form a group isomorphic via Poincar\'{e} duality to the group of $\Gr$-valued simplicial 1-cocycles on $\T$.  We denote by $[\col]\in H^1(\Sigma,\Gr)$ the associated cohomology class.
A $\Gr$-coloring of $\Gamma$ is admissible if
$\col(e)\in \Gr \setminus \X$ for any oriented edge $e$ of $\Gamma$.
A state of an admissible $\Gr$-coloring $\col$ is a map $\sigma$
assigning to every oriented edge $e$ of $\Gamma$ an element
$\sigma(e)\in I_{\col(e)}$ such that
$\sigma(-e)=\sigma(e)^{*}$. Denote by $\St(\col)$ the set of such
states.  For an admissible $\Gr$-coloring $\col$ we define the Hilbert
space $\HHH=\HHH(\Gamma,\col)$ as the span of all elements corresponding
to the states of $\col$
\[
\HHH(\Gamma,\col)=\bigoplus_{\sigma\in\St(\col)}\C\ket{\Gamma,\sigma}.
\]
We write $\sigma^{\ast}$ for the state of the $\Gr$-coloring $-\Phi$ assigning $\sigma(e)^{\ast}$ to each oriented edge $e$ of $\Gamma$.

\subsection{Inner products}
We equip $\HHH=\HHH(\Gamma,\col)$ with a complete Hilbert space structure by defining a positive definite Hermitian inner product
\begin{equation}
 \la \cdot \mid \cdot \ra_+ \maps \HHH \otimes \HHH \to \C
\end{equation}
in which the states $\ket{\Gamma,\sigma}$ form an orthonormal basis:
\begin{equation}
\la \Gamma,\sigma \mid \Gamma, \sigma' \ra_+ =  \delta_{\sigma, \sigma'}.
\end{equation}
We will see that from the TQFT perspective, this is not the most natural inner product.  The most natural inner product can be obtained from this one utilizing a Hermitian operator $\eta \maps \HHH \to \HHH$.

 Define an invertible operator $\eta$ by 
\begin{align}
  \eta \maps  \HHH(\Gamma,\col) & \; \longrightarrow \; \HHH(\Gamma,\col) \label{eq:def-eta} \\
    \ket{\Gamma, \sigma}   &\;\; \mapsto \;
   \frac{ \prod_{e\in\Gamma_1} \md(\sigma(e)) }{ \prod_{v\in \Gamma_0}\gamma(\sigma(v))\prod_{e\in \Gamma_1}\beta(\sigma(e)) }
   \ket{\Gamma, \sigma}. \nn
\end{align}
where $\sigma(v)=(j_1,j_2,j_3)\in I^3$, $j_k=\sigma(e_k)$, and
$e_1,e_2,e_3$ are the $3$ ordered edges adjacent to $v$ oriented
toward $v$.
It is clear from the orthogonality of states $\ket{\Gamma,\sigma}$ and the fact that $\gamma(\sigma(v))$, $\md(\sigma(e))$ and $\beta(\sigma(e))$ are all real that
$\eta$ is Hermitian with respect to the form $\la\cdot \mid \cdot \ra_+$
\begin{equation} \label{eq:skew-eta}
 \la \psi \mid \eta \phi \ra_+ = \la \eta \psi \mid \phi \ra_+
\end{equation}
for all $\psi, \phi \in \cal{H}$.

Define a new inner product on $\HHH=\HHH(\Gamma,\col)$  via
\begin{equation} \label{eq:Aindef-inner}
\la \cdot \mid \cdot \ra := \la \cdot \mid   \eta^{-1}(\cdot) \ra_+.
\end{equation}
This new pairing is Hermitian since
\begin{align*}
 &\la \psi \mid \phi  \ra
 = \la \psi \mid\eta^{-1} \phi \ra_+
= \overline{\la \eta^{-1} \phi \mid \psi\ra_+}
\\
&\qquad \refequal{\eqref{eq:skew-eta}}  \overline{\la\phi \mid \eta^{-1}\psi\ra_+} = \overline{\la\phi \mid \psi\ra} .
\end{align*}

However, it is indefinite since
\begin{align} \label{eq:eta-inner}
\la \Gamma, \sigma \mid \Gamma, \sigma' \ra
& = \;  \frac{ \prod_{v\in \Gamma_0}\gamma(\sigma(v))\prod_{e\in \Gamma_1}\beta(\sigma(e)) }{ \prod_{e\in\Gamma_1} \md(\sigma(e)) }\ \cdot \delta_{\sigma, \sigma'}.
\end{align}
Indeed, for the example studied in Section~\ref{sec:coeff}, the coefficient on the right-hand-side of \eqref{eq:eta-inner} takes both positive and negative values.   It follows that the $\eta$ is Hermitian with respect to the new form $\la \cdot \mid \cdot \ra$  since
\begin{align} \label{eq:eta-herm}
\la \psi \mid \eta \phi \ra &:=\la \psi \mid \eta^{-1} \eta \phi \ra_+
=\la \psi \mid \eta \eta^{-1}  \phi \ra_+
\\ \nn
&\refequal{\eqref{eq:skew-eta} }\la \eta \psi \mid \eta^{-1} \phi \ra_+
=:  \la\eta \psi \mid  \phi \ra~.
\end{align}

Observe that if a Hamiltonian $H$ is Hermitian with respect to the form $\la \cdot \mid \cdot \ra$, so that $\la \psi \mid H \phi \ra = \la H \psi \mid \phi \ra$, then we have
\begin{align*}
 \la \psi\mid  H    \phi \ra_+ &:=\la \psi\mid  \eta H    \phi \ra
 =\la  (\eta H)^{\dagger} \psi \mid \phi \ra
 \\
&\refequal{\eqref{eq:eta-herm}}\la  H\eta \psi \mid \phi \ra
:=\la  H\eta \psi \mid \eta^{-1} \phi \ra_+
\\
& \refequal{\eqref{eq:skew-eta}}\la \eta^{-1} H\eta \psi \mid \phi \ra_+~.
\end{align*}
Thus the Hermitian conjugate $H^{\sharp}$ of $H$ with respect to the form $\la \cdot \mid \cdot\ra_+$ satisfies
\begin{equation}
  H^{\sharp} = \eta^{-1} H \eta.
\end{equation}
Thus the  resulting Hamiltonian then becomes \emph{pseudo-Hermitian} with respect to the inner product $\la -,-\ra_+$.

\subsection{Topological origin of the indefinite inner product. }
In this section we show how the indefinite inner product from \eqref{eq:Aindef-inner} arises naturally from topological considerations.    This helps to illustrate the distinction between indefinite quantum mechanics and $\eta$-pseudo-Hermitian quantum mechanics articulated in \cite{Most-Is}.   The key distinction is that an explicit map $\eta$ as in \eqref{eq:def-eta} is chosen as part of the data.   From the TQFT perspective discussed in Section~\ref{sec:top}, this choice arises naturally.


Given a tuple $(\Sigma,\T,\Gamma,\col)$ of the objects described in Section~\ref{subsec:def-state}, consider the tuple $(\Sigma^{\ast},\T,\wb\Gamma,\wb\col)$ where $\Sigma^{\ast}$ is the surface $\Sigma$ with opposite orientation, $\wb\Gamma$ is the graph $\Gamma$ dual to the triangulation $\T$ whose cyclic order at each vertex had been reversed to be compatible with $\Sigma^{\ast}$ and where the choice of orientation of each edge has been reversed.
The $\Gr$-coloring $\wb\col(e)=-\col(e)^{}$ for any oriented edge $e$
of $\Gamma$.
In particular the
dual
of a state of $\col$ is a state of $\wb\col$.
For any oriented edge $e$ of $\Gamma$,
$\wb\col(e)=-\col(e)^{}$ but $\wb\col$ and $\col$ represent the same
$1$-cocycle of on the 1-skeleton of $\T$.

There is a natural pairing $\bp{\cdot,\cdot}:\HHH(\wb\Gamma,\wb\col)\otimes\HHH(\Gamma,\col)\to\C$ arising from ribbon graph evaluations associated to the relative $\Gr$-spherical category
defined on
\[
x=\!\!\!\sum_{\sigma\in\St(\col)}\!\!\!x_\sigma\ket{\wb\Gamma,\sigma^{\ast}}\text
{ and }
y=\!\!\!\sum_{\sigma\in\St(\col)}\!\!\!y_\sigma\ket{\Gamma,\sigma}
\]
by
\begin{align*}
&\bp{x,y} =   \sum_{\sigma\in\St(\col)} x_\sigma
y_\sigma \prod_{e\in\Gamma_1}\frac{1}{\md(\sigma(e))} ~.
\end{align*}
We then define
$\dagger:\HHH(\Gamma,\col)\to\HHH(\wb\Gamma,\wb\col)$ by
\begin{equation} \label{eq:bra}
\ket{\Gamma,\sigma}^\dagger=\prod_{v\in \Gamma_0}\gamma(\sigma(v))\prod_{e\in \Gamma_1}\beta(\sigma(e))\ket{\wb\Gamma,\sigma^{\ast}}
\end{equation}
where $\sigma(v)=(j_1,j_2,j_3)\in\R^3$, $j_k=\sigma(e_k)$, and
$e_1,e_2,e_3$ are the $3$ ordered edges adjacent to $v$ oriented
toward $v$.
Then the inner product from \eqref{eq:Aindef-inner} is given by
$
\braket {\psi | \phi} \;    :=\bp{\psi^\dagger,\phi}.
$

\section{ Non-semisimple Levin-Wen Hamiltonian } \label{sec:hamiltonian}

In this section we define operators acting on states associated to the dual graph $\Gamma$ of a triangulation $\T$ of the surface $\Sigma$.  Recall that  vertices of the graph $\Gamma$ correspond to triangles of $\T$, while regions (plaquettes), of the dual graph $\Gamma$ can be identified with vertices of the original triangulation.  More precisely, by a plaquette $p$, we mean a region of $\Sigma\setminus\Gamma$.
Its
boundary $\delta p$ is a union of oriented edges of $\Gamma$.  We write $\Gamma_0$ for the set vertices of the graph $\Gamma$.  Using operators associated with vertices and plaquettes, we introduce a pseudo-Hermitian Hamiltonian in Section~\ref{subsec:LevinWen}.

\subsection{Vertex operators}

Associated to each vertex $v \in \Gamma_0$ we have \emph{vertex operators} $Q_v \maps \HHH(\Gamma,\Phi) \to \HHH(\Gamma,\Phi)$ that act locally to impose the branching constraints from Section~\ref{subsec:basic}.
\[
 Q_v\left| \hackcenter{\begin{tikzpicture}[   decoration={markings, mark=at position 0.6 with {\arrow{>}};}, scale =0.8]
    \draw[thick, blue, postaction={decorate}] (0,1) -- (.5,0);
    \draw[thick, blue, postaction={decorate}] (0,-1) -- (.5,0);
     \draw[thick, blue, postaction={decorate}] (1.5,0) -- (.5,0);
    \node at (-.2,.75) {$\scs j_2$};
    \node at (-.2,-.75) {$\scs j_1$};
   \node at (1.5,.25){$\scs j_3$};
\end{tikzpicture}}
\right\rangle
\;\; = \;\;
\delta_{j_1 j_2 j_3}\left| \hackcenter{\begin{tikzpicture}[   decoration={markings, mark=at position 0.6 with {\arrow{>}};}, scale =0.8]
    \draw[thick, blue, postaction={decorate}] (0,1) -- (.5,0);
    \draw[thick, blue, postaction={decorate}] (0,-1) -- (.5,0);
     \draw[thick, blue, postaction={decorate}] (1.5,0) -- (.5,0);
    \node at (-.2,.75) {$\scs j_2$};
    \node at (-.2,-.75) {$\scs j_1$};
   \node at (1.5,.25){$\scs j_3$};
\end{tikzpicture}}
\right\rangle
\]
It is straightforward to see that these operators are mutually commuting projectors, so that
\[
Q_v^2 = Q_v, \quad Q_v Q_{v'} = Q_{v'} Q_v \qquad \text{for $v,v' \in \Gamma_0$.}
\]
Furthermore, since $Q_v$ is a delta function, it is immediate that these operators are Hermitian with respect to the inner product $\la \cdot \mid \cdot \ra$ from \eqref{eq:eta-inner}.

\subsection{Plaquette operators} \label{subsec:plaquette}
   If
$g\in \Gr$ and $p$ is a plaquette in $\Sigma\setminus\Gamma$, then we write $g.\delta p$ for  the coloring
that sends oriented edges of $\delta p$ to $g$, their opposites
to $-g^{}$, and other edges to $0$. 

For $g\in \Gr\setminus\X$ let $\col$ be an admissible $\Gr$-coloring such that $\col+g.\delta p$ is also admissible.
Let $s\in I_{g}$  and define the operator $B_p^s:\HHH(\Gamma,\col)\to\HHH(\Gamma,\col+g.\delta p)$ which acts on the boundary edges of the plaquette
$p$  and is given on a triangle plaquette by
\begin{widetext}
\begin{equation}
B_p^s
\left| \hackcenter{\begin{tikzpicture}[   decoration={markings, mark=at position 0.6 with {\arrow{>}};}, scale =0.8]
    \draw[thick, blue, postaction={decorate}] (.75,0) -- (-.75,0);
    \draw[thick, blue, postaction={decorate}] (-.75,0) -- (0,-1);
    \draw[thick, blue, postaction={decorate}]  (0,-1) --  (.75,0) ;
    \draw[thick, blue, postaction={decorate}]  (-.75,0) --  (-1.25,.5) ;
    \draw[thick, blue, postaction={decorate}]  (.75,0) --  (1.25,.5) ;
       \draw[thick, blue, postaction={decorate}]  (0,-1) --  (0,-1.5) ;
    \node at (0,.35) {$\scs j_3$};
     \node at (-.75,-.5) {$\scs j_1$};
     \node at (0,-.4) {$\scs p$};
     \node at (.75,-.5) {$\scs j_2$};
     \node at (-1.2,0) {$\scs k_3$};
      \node at (1.3,0) {$\scs k_2$};
       \node at (-.3,-1.25) {$\scs k_1$};
\end{tikzpicture}}
\right\rangle
\nonumber\\
=
\sum_{j_1',j_2',j_3'}
\md({j_1'})\md({j_2'})\md({j_3'})
N^{j_3' s j_3}_{j_1^{\ast} k_3 j_1'^{\ast}}
N^{j_1' s j_1}_{j_2^{\ast} k_1 j_2'^{\ast}}
N^{j_2' s j_2}_{j_3^{\ast} k_2 j_3'^{\ast}}
\left| \hackcenter{\begin{tikzpicture}[   decoration={markings, mark=at position 0.6 with {\arrow{>}};}, scale =0.8]
    \draw[thick, blue, postaction={decorate}] (.75,0) -- (-.75,0);
    \draw[thick, blue, postaction={decorate}] (-.75,0) -- (0,-1);
    \draw[thick, blue, postaction={decorate}]  (0,-1) --  (.75,0) ;
    \draw[thick, blue, postaction={decorate}]  (-.75,0) --  (-1.25,.5) ;
    \draw[thick, blue, postaction={decorate}]  (.75,0) --  (1.25,.5) ;
       \draw[thick, blue, postaction={decorate}]  (0,-1) --  (0,-1.5) ;
    \node at (0,.35) {$\scs j_3'$};
     \node at (-.75,-.5) {$\scs j_1'$};
     \node at (.75,-.5) {$\scs j_2'$};
     \node at (0,-.4) {$\scs p$};
     \node at (-1.2,0) {$\scs k_3$};
      \node at (1.3,0) {$\scs k_2$};
       \node at (-.3,-1.25) {$\scs k_1$};
\end{tikzpicture}}
\right\rangle.
\end{equation}
More generally, for a plaquette with $n$ sides, we define the plaquette operator by:
\begin{equation}
B_p^s
\left| \hackcenter{\begin{tikzpicture}[   decoration={markings, mark=at position 0.6 with {\arrow{>}};}, scale =0.8]
    \draw[thick, blue, postaction={decorate}] (-.75,-.25) -- (0,-1);
    \draw[thick, blue, postaction={decorate}]  (0,-1) --  (.75,-.25) ;
    \draw[thick, blue, postaction={decorate}]  (-.75,-.25) --  (-1.35,-.45) ;
    \draw[thick, blue, postaction={decorate}]  (.75,-.25) --  (1.35,-.45) ;
       \draw[thick, blue, postaction={decorate}]  (0,-1) --  (0,-1.5) ;
        \draw[thick, blue, postaction={decorate}]  (.75,-.25) --  (.75,.5) ;
         \draw[thick, blue, postaction={decorate}]  (-.75,.5) --  (-.75,-.25) ;
         \draw[thick, blue, postaction={decorate}]  (.75,.5) --  (1.25,.75) ;
           \draw[thick, blue, postaction={decorate}]  (-.75,.5) --  (-1.25,.75) ;
    \node at (0,.5) {$\dots$};
     \node at (-.65,-.85) {$\scs j_1$};
      \node at (-1.1,.25) {$\scs j_n$};
      \node at (1.1,.25) {$\scs j_3$};
     \node at (.65,-.85) {$\scs j_2'$};
     \node at (-1.3,-.65) {$\scs k_n$};
      \node at (1.3,-.65) {$\scs k_2$};
       \node at (-1.25,.95) {$\scs k_{n-1}$};
      \node at (1.25,.95) {$\scs k_3$};
       \node at (-.3,-1.45) {$\scs k_1$};
       \node at (0,0) {$\scs p$};
\end{tikzpicture}}
\right\rangle
\nonumber\\
 :=
\sum_{j_1',\ldots,j_n'}
\prod_{i=1}^n
\md({j_i'})
N^{j_i' s j_i}_{j_{i+1}^{\ast} k_i j_{i+1}'^{\ast}}
\left| \hackcenter{\begin{tikzpicture}[   decoration={markings, mark=at position 0.6 with {\arrow{>}};}, scale =0.8]
    \draw[thick, blue, postaction={decorate}] (-.75,-.25) -- (0,-1);
    \draw[thick, blue, postaction={decorate}]  (0,-1) --  (.75,-.25) ;
    \draw[thick, blue, postaction={decorate}]  (-.75,-.25) --  (-1.35,-.45) ;
    \draw[thick, blue, postaction={decorate}]  (.75,-.25) --  (1.35,-.45) ;
       \draw[thick, blue, postaction={decorate}]  (0,-1) --  (0,-1.5) ;
        \draw[thick, blue, postaction={decorate}]  (.75,-.25) --  (.75,.5) ;
         \draw[thick, blue, postaction={decorate}]  (-.75,.5) --  (-.75,-.25) ;
         \draw[thick, blue, postaction={decorate}]  (.75,.5) --  (1.25,.75) ;
           \draw[thick, blue, postaction={decorate}]  (-.75,.5) --  (-1.25,.75) ;
    \node at (0,.5) {$\dots$};
     \node at (-.65,-.85) {$\scs j_1'$};
      \node at (-1.1,.25) {$\scs j_n'$};
      \node at (1.1,.25) {$\scs j_3'$};
     \node at (.65,-.85) {$\scs j_2'$};
     \node at (-1.3,-.65) {$\scs k_n$};
      \node at (1.3,-.65) {$\scs k_2$};
       \node at (-1.25,.95) {$\scs k_{n-1}$};
      \node at (1.25,.95) {$\scs k_3$};
       \node at (-.3,-1.45) {$\scs k_1$};
       \node at (0,0) {$\scs p$};
\end{tikzpicture}}
\right\rangle .
\nonumber
\end{equation}
\end{widetext}
Note that the conditions on $6j$ symbols \eqref{eq:6jadmissible} ensure that $\deg(j_i')= \deg(j_i) - \deg(s)$ otherwise the right hand side is zero in the above formula.

Then we define $\Gr$-indexed plaquette operators
\begin{equation} \label{eq:Balpha}
  B_p^{g}=
  \sum_{s\in I_{g}}\bb(s)B_p^{s}:\HHH(\Gamma,\col)\to\HHH(\Gamma,\col+g.\delta p).
\end{equation}

\begin{proposition} \label{prop:plaqexp}
 If $g_1$, $g_2$  and $g_1+g_2$ are generic, then
  $$B_p^{g_1}B_p^{g_2}=B_p^{g_1+g_2}.$$
  Moreover, if $g_1$ and $g_2$ are generic then $B_p^{g_1}B_p^{-g_1^{}}=B_p^{g_2} B_p^{-g_2^{}}$.
\end{proposition}
\begin{widetext}
\begin{proof}
For a plaquette $p$ and elements $g_2$ and $g_1$, we compute the composition of plaquette operators:

\begin{align} \nn
B_p^{g_2} B_p^{g_1}(x) &\; =
\sum_{\substack{s \in I_{g_1} \\ t \in I_{g_2}}} \bb(s) \bb(t) \sum_{j_1',\ldots, j_n'} \sum_{j_1'',\ldots, j_n''} \prod_{i=1}^n
 \md(j_i') \md(j_i'')
 N^{j_i' s j_i}_{j_{i+1}^{\ast} k_i j_{i+1}'^{\ast}}  N^{j_i'' t j_i'}_{j_{i+1}'^{\ast} k_i j_{i+1}''^{\ast}} .
\\
& \refequal{\eqref{eq:pent}} \;\;
\nn 
\sum_{\substack{s \in I_{g_1} \\ t \in I_{g_2}}} \bb(s) \bb(t) \sum_{j_1',\ldots, j_n'} \sum_{j_1'',\ldots, j_n''} \prod_{i=1}^n
 \md(j_i') \md(j_i'') \sum_{a_i} \md(a_i)
 N_{s j_i a_i}^{j_i'' t j_i'}
 N^{j_i'' a_i j_i}_{j_{i+1}^{\ast} k_i j_{i+1}''^{\ast}}
 N^{t s a_i}_{j_{i+1}^{\ast} j_{i+1}''^{\ast} j_{i+1}'^{\ast}}
\\
& \refequal{\eqref{eq:symm}} \; \;
\label{plaqexp3}
\sum_{\substack{s \in I_{g_1} \\ t \in I_{g_2}}} \bb(s) \bb(t) \sum_{j_1',\ldots, j_n'} \sum_{j_1'',\ldots, j_n''} \prod_{i=1}^n
 \md(j_i') \md(j_i'') \sum_{a_i} \md(a_i)
 N^{j_i'' a_i j_i}_{j_{i+1}^{\ast} k_i j_{i+1}''^{\ast}}
 N^{t s a_i}_{j_{i+1}^{\ast} j_{i+1}''^{\ast} j_{i+1}'^{\ast}}
  N_{j_{i}'^{\ast} j_{i}''^{\ast} j_{i}^{\ast}}^{a_i s^{\ast} t}  .
\end{align}
Reindexing the product for the last factor in \eqref{plaqexp3} yields that $B_p^{g_2} B_p^{g_1}(x) $ is equal to
\begin{equation} \label{plaqexp4}
\sum_{\substack{s \in I_{g_1} \\ t \in I_{g_2}}} \bb(s) \bb(t) \sum_{j_1',\ldots, j_n'} \sum_{j_1'',\ldots, j_n''} \prod_{i=1}^n
 \md(j_i') \md(j_i'') \sum_{a_i} \md(a_i)
 N^{j_i'' a_i j_i}_{j_{i+1}^{\ast} k_i j_{i+1}''^{\ast}}
 N^{t s a_i}_{j_{i+1}^{\ast} j_{i+1}''^{\ast} j_{i+1}'^{\ast}}
  N_{j_{i+1}'^{\ast} j_{i+1}''^{\ast} j_{i+1}^{\ast}}^{a_{i+1} s^{\ast} t}  .
\end{equation}
Using the orthogonality property \eqref{eq-ortho} for the last two factors in \eqref{plaqexp4} yields
\begin{align} \label{plaqexp5}
B_p^{g_2} B_p^{g_1}(x) &\;=\;
\sum_{\substack{s \in I_{g_1} \\ t \in I_{g_2}}} \bb(s) \bb(t)
\sum_{j_1'',\ldots, j_n''} \prod_{i=1}^n
 \md(j_i'') \sum_{a_i} \md(a_i)
  N^{j_i'' a_i j_i}_{j_{i+1}^{\ast} k_i j_{i+1}''^{\ast}} \frac{\delta_{a_i,a_{i+1}}}{\md(a_{i+1})}
  \delta_{t s a_{i+1}^{\ast}} \delta_{a_{i+1} j_{i+1}^{\ast} j_i''} \\
  &\;=\;
\sum_{\substack{s \in I_{g_1} \\ t \in I_{g_2}}} \bb(s) \bb(t)
\sum_{j_1'',\ldots, j_n''} \prod_{i=1}^n
 \md(j_i'') \sum_{a}
  N^{j_i'' a j_i}_{j_{i+1}^{\ast} k_i j_{i+1}''^{\ast}}
  \delta_{t s a_{}^{\ast}} \delta_{a_{} j_{i+1}^{\ast} j_i''}
\nn \\
&\refequal{\eqref{eq:b}}
\sum_{u \in I_{g_1+g_2}} \bb(u) \sum_{j_1'',\ldots,j_n''} \prod_i \md(j_i'')   N^{j_i'' a j_i}_{j_{i+1}^{\ast} k_i j_{i+1}''^{\ast}}
\;\; =\;\;
B_p^{g_1 +g_2}.
\end{align}

For the second statement, first note that for generic colors $B_p^{g_1}=B_p^{g_3}B_p^{g_1-g_3}$. Then
$$B_p^{g_1}B_p^{-g_1^{}}=B_p^{g_3}B_p^{g_1 - g_3^{} }B_p^{-g_1^{}}=B_p^{g_3}B_p^{-g_3^{}}. $$
\end{proof}

\subsection{Commutativity of operators}
In this subsection we show that plaquette operators commute following a proof given in \cite{HW}.
\begin{proposition} \label{1edgecomm}
Let $p$ and $p'$ be two plaquettes which share exactly one common edge $e_1$ in a state $\ket{\Gamma,\sigma}$ as in \eqref{adjplaq1}.
Then $B_{p'}^t B_{p}^s=B_{p}^s B_{p'}^t $.
\begin{equation} \label{adjplaq1}
\ket{\Gamma, \sigma} \;\; := \;\;
\hackcenter{\begin{tikzpicture}[ decoration={markings, mark=at position 0.6 with {\arrow{>}};}, scale =0.9]
    \draw[thick, blue, postaction={decorate}] (0,-.5) -- (0,.5);
    \draw[thick, blue, postaction={decorate}] (-.75,-1) -- (0,-.5);
    \draw[thick, blue, postaction={decorate}] (.75,-1) -- (0,-.5);
     \draw[thick, blue, postaction={decorate}] (0,.5) -- (-.75,1);
     \draw[thick, blue, postaction={decorate}] (0,.5) -- (.75,1);
      \draw[thick, blue, postaction={decorate}] (-.75,1) -- (-.75,1.5);
      \draw[thick, blue, postaction={decorate}] (.75,1) -- (.75,1.5);
       \draw[thick, blue, postaction={decorate}] (-1.75,1) -- (-2,1.5);
      \draw[thick, blue, postaction={decorate}] (1.75,1) -- (2,1.5);
     \draw[thick, blue, postaction={decorate}] (-.75,-1.5) -- (-.75,-1);
     \draw[thick, blue, postaction={decorate}] (.75,-1.5) -- (.75,-1);
     \draw[thick, blue, postaction={decorate}] (.75,1) -- (1.75,1);
       \draw[thick, blue, postaction={decorate}] (1.75,1) -- (2.25,.5);
      \draw[thick, blue, postaction={decorate}] (1.75,-1) -- (.75,-1);
       \draw[thick, blue, postaction={decorate}] (-1.75,-1) -- (-.75,-1);
       \draw[thick, blue, postaction={decorate}] (-.75,1) -- (-1.75,1);
       \draw[thick, blue, postaction={decorate}] (-1.75,1) -- (-2.25,.5);
      \node at (-2.25,0){$\ddots$};
     \node at (-1.25,0){$p$};
    \node at (1.25,0){$p'$};
   \node at (.3,-.2){$\scs e_1$};
   \node at (-.15,-1.05){$\scs f_{a+1}$};
   \node at (.85,-.75){$\scs g_{b+1}$};
    \node at (1.45,-1.3){$\scs g_{b}$};
    \node at (-1.45,-1.3){$\scs f_a$};
     \node at (-.3,1.1){$\scs f_2$};
     \node at (-2.3,.9){$\scs f_4$};
      \node at (2.3,.9){$\scs g_4$};
   \node at (.3,1.1){$\scs g_2$};
   \node at (.9,-1.65){$\scs \ell_b$};
   \node at (1,1.65){$\scs \ell_2$};
   \node at (2.25,1.5){$\scs \ell_3$};
   \node at (-.9,-1.65){$\scs k_a$};
   \node at (-.95,1.65){$\scs k_2$};    \node at (-2.2,1.65){$\scs k_3$};
    \node at (-1.4,1.35){$\scs f_3$};
      \node at (1.4,1.35){$\scs g_3$};
      \node at (2.25,-.25){$\vdots $};
\end{tikzpicture}}
\end{equation}
\end{proposition}

\begin{proof}
First we compute
$B_{p'}^t B_{p}^s$:
\begin{align} 
&B_{p'}^t B_{p}^s \ket{\Gamma, \Phi} = \nn
\\
& \nn \quad \sum_{e_1''} \sum_{g_2',\ldots, g_{b+1}'}
\sum_{e_1'} \sum_{f_2', \ldots, f_{a+1}'}
\left[\md(e_1') N^{e_1' s e_1}_{f_2^{\ast} g_2 f_2'^{\ast}}
\md(f_{a+1}') N^{f_{a+1}' s f_{a+1}}_{e_1^{\ast} g_{b+1}^{\ast} e_1'^{\ast}}
 \prod_{i=1}^{a-1} \md(f_{i+1}') N^{f_{i+1}' s f_{i+1}}_{f_{i+2}^{\ast} k_{i+1} f_{i+2}'^{\ast}}\right]
 \\
& \;\; \times \left[
\md(e_1''^{\ast}) N^{e_1''^{\ast} t e_1'^{\ast}}_{g_{b+1} f_{a+1}'^{\ast} g_{b+1}'}
\md(g_{2}'^{\ast}) N^{g_2'^{\ast} t g_2^{\ast}}_{e_1' f_2' e_1''}
\prod_{i=2}^{b} \md(g_{i+1}'^{\ast}) N^{g_{i+1}'^{\ast} t g_{i+1}}_{g_i' \ell_i g_i} \right]
\nonumber
\Biggl|
\hackcenter{\begin{tikzpicture}[ decoration={markings, mark=at position 0.6 with {\arrow{>}};}, scale =0.9]
    \draw[thick, blue, postaction={decorate}] (0,-.5) -- (0,.5);
    \draw[thick, blue, postaction={decorate}] (-.75,-1) -- (0,-.5);
    \draw[thick, blue, postaction={decorate}] (.75,-1) -- (0,-.5);
     \draw[thick, blue, postaction={decorate}] (0,.5) -- (-.75,1);
     \draw[thick, blue, postaction={decorate}] (0,.5) -- (.75,1);
      \draw[thick, blue, postaction={decorate}] (-.75,1) -- (-.75,1.5);
      \draw[thick, blue, postaction={decorate}] (.75,1) -- (.75,1.5);
       \draw[thick, blue, postaction={decorate}] (-1.75,1) -- (-2,1.5);
      \draw[thick, blue, postaction={decorate}] (1.75,1) -- (2,1.5);
     \draw[thick, blue, postaction={decorate}] (-.75,-1.5) -- (-.75,-1);
     \draw[thick, blue, postaction={decorate}] (.75,-1.5) -- (.75,-1);
     \draw[thick, blue, postaction={decorate}] (.75,1) -- (1.75,1);
       \draw[thick, blue, postaction={decorate}] (1.75,1) -- (2.25,.5);
      \draw[thick, blue, postaction={decorate}] (1.75,-1) -- (.75,-1);
       \draw[thick, blue, postaction={decorate}] (-1.75,-1) -- (-.75,-1);
       \draw[thick, blue, postaction={decorate}] (-.75,1) -- (-1.75,1);
       \draw[thick, blue, postaction={decorate}] (-1.75,1) -- (-2.25,.5);
     \node at (-1.25,0){$p$};
    \node at (1.25,0){$p'$};
   \node at (.3,-.2){$\scs e_1''$};
   \node at (-.05,-1.05){$\scs f_{a+1}'$};
   \node at (.85,-.7){$\scs g_{b+1}'$};
    \node at (1.45,-1.3){$\scs g_b'$};
    \node at (-1.45,-1.3){$\scs f_a'$};
     \node at (-.3,1.1){$\scs f_2'$};
     \node at (-2.3,.9){$\scs f_4'$};
      \node at (2.3,.9){$\scs g_4'$};
   \node at (.3,1.1){$\scs g_2'$};
   \node at (.9,-1.65){$\scs \ell_b$};
   \node at (1,1.65){$\scs \ell_2$};
   \node at (2.25,1.5){$\scs \ell_3$};
   \node at (-.9,-1.65){$\scs k_a$};
   \node at (-.95,1.65){$\scs k_2$};    \node at (-2.2,1.65){$\scs k_3$};
    \node at (-1.4,1.35){$\scs f_3'$};
      \node at (1.4,1.35){$\scs g_3'$};
      \node at (2.25,-.25){$\vdots $};
\end{tikzpicture}}
\Biggr\ra  .
\end{align}
Compare with the computation of $B_{p}^s B_{p'}^t$:
\begin{align}  
&B_{p}^s B_{p'}^t \ket{\Gamma, \Phi}= \nn
\\
&\nn \quad \sum_{e_1''} \sum_{f_2',\ldots,f_{a+1}'}  \sum_{r_1'} \sum_{g_2',\ldots, g_{b+1}'}
 \left[\md(r_1'^{\ast}) N^{r_1'^{\ast} t e_1^{\ast}}_{g_{b+1} f_{a+1}^{\ast} g_{b+1}'}
\md(g_2'^{\ast}) N^{g_2'^{\ast} t g_2^{\ast}}_{e_1 f_2 r_1'}
\prod_{i=2}^b \md(g_{i+1}'^{\ast}) N^{g_{i+1}'^{\ast} t g_{i+1}^{\ast}}_{g_{i}' \ell_{i} g_i} \right] \\
& \;\; \times\left[\md(e_1'') N^{e_1'' s r_1'}_{f_2^{\ast} g_2' f_2'^{\ast}}
\md(f_{a+1}') N^{f_{a+1}' s f_{a+1}}_{r_1'^{\ast} g_{b+1}'^{\ast} e_1''^{\ast}}
\prod_{i=1}^{a-1} \md(f_{i+1}') N^{f_{i+1}' s f_{i+1}}_{f_{i+2}^{\ast} k_{i+1} f_{i+2}'^{\ast}}
\right] \nonumber
\Biggl|
\hackcenter{\begin{tikzpicture}[ decoration={markings, mark=at position 0.6 with {\arrow{>}};}, scale =0.9]
    \draw[thick, blue, postaction={decorate}] (0,-.5) -- (0,.5);
    \draw[thick, blue, postaction={decorate}] (-.75,-1) -- (0,-.5);
    \draw[thick, blue, postaction={decorate}] (.75,-1) -- (0,-.5);
     \draw[thick, blue, postaction={decorate}] (0,.5) -- (-.75,1);
     \draw[thick, blue, postaction={decorate}] (0,.5) -- (.75,1);
      \draw[thick, blue, postaction={decorate}] (-.75,1) -- (-.75,1.5);
      \draw[thick, blue, postaction={decorate}] (.75,1) -- (.75,1.5);
       \draw[thick, blue, postaction={decorate}] (-1.75,1) -- (-2,1.5);
      \draw[thick, blue, postaction={decorate}] (1.75,1) -- (2,1.5);
     \draw[thick, blue, postaction={decorate}] (-.75,-1.5) -- (-.75,-1);
     \draw[thick, blue, postaction={decorate}] (.75,-1.5) -- (.75,-1);
     \draw[thick, blue, postaction={decorate}] (.75,1) -- (1.75,1);
       \draw[thick, blue, postaction={decorate}] (1.75,1) -- (2.25,.5);
      \draw[thick, blue, postaction={decorate}] (1.75,-1) -- (.75,-1);
       \draw[thick, blue, postaction={decorate}] (-1.75,-1) -- (-.75,-1);
       \draw[thick, blue, postaction={decorate}] (-.75,1) -- (-1.75,1);
       \draw[thick, blue, postaction={decorate}] (-1.75,1) -- (-2.25,.5);
     \node at (-1.25,0){$p$};
    \node at (1.25,0){$p'$};
   \node at (.3,-.2){$\scs e_1''$};
   \node at (-.25,-1.2){$\scs f_{a+1}'$};
   \node at (.35,-1.1){$\scs g_{b+1}'$};
    \node at (1.45,-1.3){$\scs g_b'$};
    \node at (-1.45,-1.3){$\scs f_a'$};
     \node at (-.3,1.1){$\scs f_2'$};
     \node at (-2.3,.9){$\scs f_4'$};
      \node at (2.3,.9){$\scs g_4'$};
   \node at (.3,1.1){$\scs g_2'$};
   \node at (.9,-1.65){$\scs \ell_b$};
   \node at (1,1.65){$\scs \ell_2$};
   \node at (2.25,1.5){$\scs \ell_3$};
   \node at (-.9,-1.65){$\scs k_a$};
   \node at (-.95,1.65){$\scs k_2$};    \node at (-2.2,1.65){$\scs k_3$};
    \node at (-1.4,1.35){$\scs f_3'$};
      \node at (1.4,1.35){$\scs g_3'$};
      \node at (2.25,-.25){$\vdots $};
\end{tikzpicture}}
\Biggr\ra.
\end{align}

Fixing $e_1'', f_2', \ldots, f_{a+1}', g_2', \ldots, g_{b+1}'$, and canceling identical factors 
leaves us with the task of verifying:
\begin{align} \label{pIpII3}
&\sum_{e_1'}
(\md(f_{a+1}') N^{f_{a+1}' s f_{a+1}}_{e_1^{\ast} g_{b+1}^{\ast} e_1'^{\ast}})
(\md(e_1') N^{e_1' s e_1}_{f_2^{\ast} g_2 f_2'^{\ast}})
(\md(e_1''^{\ast}) N^{e_1''^{\ast} t e_1'^{\ast}}_{g_{b+1} f_{a+1}'^{\ast} g_{b+1}'})
(\md(g_{2}'^{\ast}) N^{g_2'^{\ast} t g_2^{\ast}}_{e_1' f_2' e_1''})
\\
=&\sum_{r_1'}
(\md(r_1'^{\ast}) N^{r_1'^{\ast} t e_1^{\ast}}_{g_{b+1} f_{a+1}^{\ast} g_{b+1}'})
(\md(g_2'^{\ast}) N^{g_2'^{\ast} t g_2^{\ast}}_{e_1 f_2 r_1'})
(\md(f_{a+1}') N^{f_{a+1}' s f_{a+1}}_{r_1'^{\ast} g_{b+1}'^{\ast} e_1''^{\ast}})
(\md(e_1'') N^{e_1'' s r_1'}_{f_2^{\ast} g_2' f_2'^{\ast}}) . \nonumber
\end{align}
Using tetrahedral symmetry, we need to verify
\begin{align} \label{pIpII4}
&\sum_{e_1'}
 \md(f_{a+1}') \md(e_1') \md(e_1''^{\ast}) \md(g_{2}'^{\ast})
 N^{e_1'^{\ast} e_1 s}_{f_{a+1}^{\ast} f_{a+1}'^{\ast} g_{b+1}}
 N^{e_1''^{\ast} t e_1'^{\ast}}_{g_{b+1} f_{a+1}'^{\ast} g_{b+1}'}
 N^{g_2^{\ast} e_1 f_2}_{s^{\ast} f_2' e_1'}
 N^{g_2'^{\ast} t g_2^{\ast}}_{e_1' f_2' e_1''}
 \\
=&\sum_{r_1'}
\md(r_1'^{\ast}) \md(g_2'^{\ast}) \md(f_{a+1}') \md(e_1'')
N^{r_1'^{\ast} t e_1^{\ast}}_{g_{b+1} f_{a+1}^{\ast} g_{b+1}'}
N^{s^{\ast} e_1''^{\ast} r_1'^{\ast}}_{g_{b+1}' f_{a+1'}^{\ast} f_{a+1}'^{\ast}}
N^{g_2'^{\ast} t g_2^{\ast}}_{e_1 f_2 r_1'}
N^{f_2' e_1''^{\ast} g_2'^{\ast}}_{r_1' f_2 s} .
\nonumber
\end{align}
Now we perform the pentagon identity on the last two factors of the left-hand side of \eqref{pIpII4} and on the first two factors of the right-hand side of \eqref{pIpII4}, we obtain the following equality to be verified.
\begin{align} \label{pIpII5}
&\sum_{e_1', z_2}
 \md(f_{a+1}') \md(e_1') \md(e_1''^{\ast}) \md(g_{2}'^{\ast}) \md(z_2)
 N^{e_1'^{\ast} e_1 s}_{f_{a+1}^{\ast} f_{a+1}'^{\ast} g_{b+1}}
 N^{e_1''^{\ast} t e_1'^{\ast}}_{g_{b+1} f_{a+1}'^{\ast} g_{b+1}'}
N^{g_2'^{\ast} t g_2^{\ast}}_{e_1 f_2 z_2}
N^{g_2'^{\ast} z_2 f_2}_{s^{\ast} f_2' e_1''}
N^{t e_1 z_2}_{s^{\ast} e_1'' e_1'}
 \\
& \;\; = \sum_{r_1',z_1}
\md(r_1'^{\ast}) \md(g_2'^{\ast}) \md(f_{a+1}') \md(e_1'') \md(z_1)
N^{s^{\ast} e_1''^{\ast} r_1'^{\ast}}_{t e_1^{\ast} z_1}
N^{s^{\ast} z_1 e_1^{\ast}}_{g_{b+1} f_{a+1}^{\ast} f_{a+1}'^{\ast}}
N^{e_1''^{\ast} t z_1}_{g_{b+1} f_{a+1}'^{\ast} g_{b+1}'}
N^{g_2'^{\ast} t g_2^{\ast}}_{e_1 f_2 r_1'}
N^{f_2' e_1''^{\ast} g_2'^{\ast}}_{r_1' f_2 s} .
\nonumber
\end{align}
Letting $z_1=e_1'^{\ast}$, and $z_2=r_1'$, and using symmetry, we see that \eqref{pIpII5} holds by matching up the
first, second third, fourth, and fifth factors on the left-hand side with the
second, third, fourth, fifth, and first factors on the right-hand side respectively.
\end{proof}
\end{widetext}

\begin{corollary} \label{corplaqcomm}
Let $p$ and $p'$ be two plaquettes.
Then $B_{p'}^t B_{p}^s=B_{p}^s B_{p'}^t $.
\end{corollary}

\begin{proof}
If $p$ and $p'$ are identical, or if they share no edges, the result is obvious.
The case that they share exactly one edge is proved in Proposition \ref{1edgecomm}.
\end{proof}

\subsection{Hermitian adjoints of plaquette operators}

See \cite{PhysRevB.103.195155} for a closely related calculation in the semisimple case.  In the computations below, we assume that the states are identical except in a neighborhood of the graph shown.

\begin{proposition} \label{plaqhermprop}
The Hermitian adjoint of $B_p^s$ is $B_p^{s^{\ast}}$.  In particular,
\[
\left(B_p^{g}\right)^{\dagger} = B_p^{-g} .
\]
\end{proposition}

\begin{proof}
We begin by computing matrix elements of $(B_p^s)^{\dagger}$, where we ignore the contributions to the inner product from edges not in the neighborhood of the plaquette, as these will cancel with the identical contributions in \eqref{plaqherm4}.
\begin{align}
& \Biggl\langle
\hackcenter{\begin{tikzpicture}[   decoration={markings, mark=at position 0.6 with {\arrow{>}};}, scale =0.7]
    \draw[thick, blue, postaction={decorate}] (-.75,-.25) -- (0,-1);
    \draw[thick, blue, postaction={decorate}]  (0,-1) --  (.75,-.25) ;
    \draw[thick, blue, postaction={decorate}]  (-.75,-.25) --  (-1.35,-.45) ;
    \draw[thick, blue, postaction={decorate}]  (.75,-.25) --  (1.35,-.45) ;
       \draw[thick, blue, postaction={decorate}]  (0,-1) --  (0,-1.5) ;
        \draw[thick, blue, postaction={decorate}]  (.75,-.25) --  (.75,.5) ;
         \draw[thick, blue, postaction={decorate}]  (-.75,.5) --  (-.75,-.25) ;
         \draw[thick, blue, postaction={decorate}]  (.75,.5) --  (1.25,.75) ;
           \draw[thick, blue, postaction={decorate}]  (-.75,.5) --  (-1.25,.75) ;
    \node at (0,.5) {$\dots$};
     \node at (-.65,-.85) {$\scs j_1$};
      \node at (-1.1,.25) {$\scs j_n$};
      \node at (1.1,.25) {$\scs j_3$};
     \node at (.65,-.85) {$\scs j_2$};
     \node at (-1.3,-.65) {$\scs k_n$};
      \node at (1.3,-.65) {$\scs k_2$};
       \node at (-1.25,.95) {$\scs k_{n-1}$};
      \node at (1.25,.95) {$\scs k_3$};
       \node at (-.3,-1.45) {$\scs k_1$};
       \node at (0,0) {$\scs p$};
\end{tikzpicture}}
\Biggl|
(B_p^s)^{\dagger}
\Biggl|
\hackcenter{\begin{tikzpicture}[   decoration={markings, mark=at position 0.6 with {\arrow{>}};}, scale =0.7]
    \draw[thick, blue, postaction={decorate}] (-.75,-.25) -- (0,-1);
    \draw[thick, blue, postaction={decorate}]  (0,-1) --  (.75,-.25) ;
    \draw[thick, blue, postaction={decorate}]  (-.75,-.25) --  (-1.35,-.45) ;
    \draw[thick, blue, postaction={decorate}]  (.75,-.25) --  (1.35,-.45) ;
       \draw[thick, blue, postaction={decorate}]  (0,-1) --  (0,-1.5) ;
        \draw[thick, blue, postaction={decorate}]  (.75,-.25) --  (.75,.5) ;
         \draw[thick, blue, postaction={decorate}]  (-.75,.5) --  (-.75,-.25) ;
         \draw[thick, blue, postaction={decorate}]  (.75,.5) --  (1.25,.75) ;
           \draw[thick, blue, postaction={decorate}]  (-.75,.5) --  (-1.25,.75) ;
    \node at (0,.5) {$\dots$};
     \node at (-.65,-.85) {$\scs j_1'$};
      \node at (-1.1,.25) {$\scs j_n'$};
      \node at (1.1,.25) {$\scs j_3'$};
     \node at (.65,-.85) {$\scs j_2'$};
     \node at (-1.3,-.65) {$\scs k_n$};
      \node at (1.3,-.65) {$\scs k_2$};
       \node at (-1.25,.95) {$\scs k_{n-1}$};
      \node at (1.25,.95) {$\scs k_3$};
       \node at (-.3,-1.45) {$\scs k_1$};
       \node at (0,0) {$\scs p$};
\end{tikzpicture}}
\Biggr\rangle
\nn\\ \label{plaqherm1}
&=
\overline{
\Biggl\langle
\hackcenter{\begin{tikzpicture}[   decoration={markings, mark=at position 0.6 with {\arrow{>}};}, scale =0.7]
    \draw[thick, blue, postaction={decorate}] (-.75,-.25) -- (0,-1);
    \draw[thick, blue, postaction={decorate}]  (0,-1) --  (.75,-.25) ;
    \draw[thick, blue, postaction={decorate}]  (-.75,-.25) --  (-1.35,-.45) ;
    \draw[thick, blue, postaction={decorate}]  (.75,-.25) --  (1.35,-.45) ;
       \draw[thick, blue, postaction={decorate}]  (0,-1) --  (0,-1.5) ;
        \draw[thick, blue, postaction={decorate}]  (.75,-.25) --  (.75,.5) ;
         \draw[thick, blue, postaction={decorate}]  (-.75,.5) --  (-.75,-.25) ;
         \draw[thick, blue, postaction={decorate}]  (.75,.5) --  (1.25,.75) ;
           \draw[thick, blue, postaction={decorate}]  (-.75,.5) --  (-1.25,.75) ;
    \node at (0,.5) {$\dots$};
     \node at (-.65,-.85) {$\scs j_1'$};
      \node at (-1.1,.25) {$\scs j_n'$};
      \node at (1.1,.25) {$\scs j_3'$};
     \node at (.65,-.85) {$\scs j_2'$};
     \node at (-1.3,-.65) {$\scs k_n$};
      \node at (1.3,-.65) {$\scs k_2$};
       \node at (-1.25,.95) {$\scs k_{n-1}$};
      \node at (1.25,.95) {$\scs k_3$};
       \node at (-.3,-1.45) {$\scs k_1$};
       \node at (0,0) {$\scs p$};
\end{tikzpicture}}
\Biggl|
B_p^s
\Biggl|
\hackcenter{\begin{tikzpicture}[   decoration={markings, mark=at position 0.6 with {\arrow{>}};}, scale =0.7]
    \draw[thick, blue, postaction={decorate}] (-.75,-.25) -- (0,-1);
    \draw[thick, blue, postaction={decorate}]  (0,-1) --  (.75,-.25) ;
    \draw[thick, blue, postaction={decorate}]  (-.75,-.25) --  (-1.35,-.45) ;
    \draw[thick, blue, postaction={decorate}]  (.75,-.25) --  (1.35,-.45) ;
       \draw[thick, blue, postaction={decorate}]  (0,-1) --  (0,-1.5) ;
        \draw[thick, blue, postaction={decorate}]  (.75,-.25) --  (.75,.5) ;
         \draw[thick, blue, postaction={decorate}]  (-.75,.5) --  (-.75,-.25) ;
         \draw[thick, blue, postaction={decorate}]  (.75,.5) --  (1.25,.75) ;
           \draw[thick, blue, postaction={decorate}]  (-.75,.5) --  (-1.25,.75) ;
    \node at (0,.5) {$\dots$};
     \node at (-.65,-.85) {$\scs j_1$};
      \node at (-1.1,.25) {$\scs j_n$};
      \node at (1.1,.25) {$\scs j_3$};
     \node at (.65,-.85) {$\scs j_2$};
     \node at (-1.3,-.65) {$\scs k_n$};
      \node at (1.3,-.65) {$\scs k_2$};
       \node at (-1.25,.95) {$\scs k_{n-1}$};
      \node at (1.25,.95) {$\scs k_3$};
       \node at (-.3,-1.45) {$\scs k_1$};
       \node at (0,0) {$\scs p$};
\end{tikzpicture}}
\Biggr\rangle}
\end{align}
\begin{equation} \label{plaqherm2}
\refequal{\eqref{eq:eta-inner}}
\prod_{i=1}^n \overline{\md(j_i')} \overline{N^{j_i' s j_i}_{j_{i+1}^{\ast} k_i j_{i+1}'^{\ast}}} \frac{\beta(j_i') \beta(k_i)}{\md(k_i) \md(j_i')} \gamma(j_{i+1}'^{\ast}, k_i^{\ast}, j_i')
\end{equation}
This simplifies to
\begin{equation} \label{plaqherm3}
\prod_{i=1}^n  \overline{N^{j_i' s j_i}_{j_{i+1}^{\ast} k_i j_{i+1}'^{\ast}}} \frac{\beta(j_i') \beta(k_i)}{\md(k_i)} \gamma(j_{i+1}'^{\ast}, k_i^{\ast}, j_i').
\end{equation}

On the other hand, the matrix elements of $B_p^{s^{\ast}}$ are given by
\begin{align}
&\Biggl\langle
\hackcenter{\begin{tikzpicture}[   decoration={markings, mark=at position 0.6 with {\arrow{>}};}, scale =0.7]
    \draw[thick, blue, postaction={decorate}] (-.75,-.25) -- (0,-1);
    \draw[thick, blue, postaction={decorate}]  (0,-1) --  (.75,-.25) ;
    \draw[thick, blue, postaction={decorate}]  (-.75,-.25) --  (-1.35,-.45) ;
    \draw[thick, blue, postaction={decorate}]  (.75,-.25) --  (1.35,-.45) ;
       \draw[thick, blue, postaction={decorate}]  (0,-1) --  (0,-1.5) ;
        \draw[thick, blue, postaction={decorate}]  (.75,-.25) --  (.75,.5) ;
         \draw[thick, blue, postaction={decorate}]  (-.75,.5) --  (-.75,-.25) ;
         \draw[thick, blue, postaction={decorate}]  (.75,.5) --  (1.25,.75) ;
           \draw[thick, blue, postaction={decorate}]  (-.75,.5) --  (-1.25,.75) ;
    \node at (0,.5) {$\dots$};
     \node at (-.65,-.85) {$\scs j_1$};
      \node at (-1.1,.25) {$\scs j_n$};
      \node at (1.1,.25) {$\scs j_3$};
     \node at (.65,-.85) {$\scs j_2$};
     \node at (-1.3,-.65) {$\scs k_n$};
      \node at (1.3,-.65) {$\scs k_2$};
       \node at (-1.25,.95) {$\scs k_{n-1}$};
      \node at (1.25,.95) {$\scs k_3$};
       \node at (-.3,-1.45) {$\scs k_1$};
       \node at (0,0) {$\scs p$};
\end{tikzpicture}}
\Biggl|
B_p^{s^{\ast}}
\Biggl|
\hackcenter{\begin{tikzpicture}[   decoration={markings, mark=at position 0.6 with {\arrow{>}};}, scale =0.7]
    \draw[thick, blue, postaction={decorate}] (-.75,-.25) -- (0,-1);
    \draw[thick, blue, postaction={decorate}]  (0,-1) --  (.75,-.25) ;
    \draw[thick, blue, postaction={decorate}]  (-.75,-.25) --  (-1.35,-.45) ;
    \draw[thick, blue, postaction={decorate}]  (.75,-.25) --  (1.35,-.45) ;
       \draw[thick, blue, postaction={decorate}]  (0,-1) --  (0,-1.5) ;
        \draw[thick, blue, postaction={decorate}]  (.75,-.25) --  (.75,.5) ;
         \draw[thick, blue, postaction={decorate}]  (-.75,.5) --  (-.75,-.25) ;
         \draw[thick, blue, postaction={decorate}]  (.75,.5) --  (1.25,.75) ;
           \draw[thick, blue, postaction={decorate}]  (-.75,.5) --  (-1.25,.75) ;
    \node at (0,.5) {$\dots$};
     \node at (-.65,-.85) {$\scs j_1'$};
      \node at (-1.1,.25) {$\scs j_n'$};
      \node at (1.1,.25) {$\scs j_3'$};
     \node at (.65,-.85) {$\scs j_2'$};
     \node at (-1.3,-.65) {$\scs k_n$};
      \node at (1.3,-.65) {$\scs k_2$};
       \node at (-1.25,.95) {$\scs k_{n-1}$};
      \node at (1.25,.95) {$\scs k_3$};
       \node at (-.3,-1.45) {$\scs k_1$};
       \node at (0,0) {$\scs p$};
\end{tikzpicture}}
\Biggr\rangle \nn
\nn \\ \label{plaqherm4}
& \;\; =\;\;
\prod_{i=1}^n \frac{\beta(j_i) \beta(k_i)}{\md(k_i)} \gamma(j_{i+1}^{\ast}, k_i^{\ast}, j_i)
N^{j_i s^{\ast} j_i'}_{j_{i+1}'^{\ast} k_i j_{i+1}^{\ast}}
\end{align}
Showing that \eqref{plaqherm3} and \eqref{plaqherm4} are equal  reduces to checking the following equality:
\begin{align}
&\prod_{i=1}^n  \overline{N^{j_i' s j_i}_{j_{i+1}^{\ast} k_i j_{i+1}'^{\ast}}} \beta(j_i') \gamma(j_{i+1}'^{\ast}, k_i^{\ast}, j_i')
\nn \\ \label{plaqherm6} &=
\prod_{i=1}^n \beta(j_i)  \gamma(j_{i+1}^{\ast}, k_i^{\ast}, j_i)
N^{j_i s^{\ast} j_i'}_{j_{i+1}'^{\ast} k_i j_{i+1}^{\ast}}
\end{align}
Using \eqref{eq:6j}, and a symmetry from \eqref{eq:symm} on $N^{j_i s j_i'}_{j_{i+1}'^{\ast} k_i j_{i+1}^{\ast}}$, this amounts to checking:
\begin{align}\nonumber
\prod_{i=1}^n &[\beta(j_i')
\gamma(j_{i+1}'^{\ast},k_i^{\ast},j_i')]
\gamma(j_i',s,j_i^{\ast})
\gamma(j_i'^{\ast}, k_i, j_{i+1}')
\\ \nn & \times  \gamma(s^{\ast}, j_{i+1}'^{\ast},j_{i+1}) \gamma(j_i,j_{i+1}^{\ast}, k_i^{\ast})
 \\ \nn
& \quad \times \beta(j_i') \beta(s) \beta(j_i) \beta(j_{i+1}^{\ast}) \beta(k_i) \beta(j_{i+1}'^{\ast})
N^{s^{\ast} j_i'^{\ast} j_i^{\ast}}_{k_i j_{i+1}^{\ast} j_{i+1}'^{\ast}}
\\
&\;\;= \;\;\label{plaqherm7}
\prod_{i=1}^n \beta(j_i) \gamma(j_{i+1}^{\ast}, k_i^{\ast}, j_i)
N^{s^{\ast} j_i'^{\ast} j_i^{\ast}}_{k_i j_{i+1}^{\ast} j_{i+1}'^{\ast}}
\end{align}
Canceling $\beta(j_i)$ from both sides of \eqref{plaqherm7} and reorganizing factors, \eqref{plaqherm7} is equivalent to
\begin{align}\nonumber
\prod_{i=1}^n
&[\gamma(j_{i+1}'^{\ast},k_i^{\ast},j_i') \gamma(j_i'^{\ast}, k_i, j_{i+1}') \beta(j_i') \beta(j_{i+1}^{\ast}) \beta(k_i)]
\\ \nn
& \times [\gamma(j_i',s,j_i^{\ast}) \gamma(s^{\ast}, j_{i+1}'^{\ast},j_{i+1}) \beta(j_i') \beta(s) \beta(j_{i+1}^{\ast})]  \\
& \quad \times \gamma(j_i,j_{i+1}^{\ast}, k_i^{\ast}) N^{s^{\ast} j_i'^{\ast} j_i^{\ast}}_{k_i j_{i+1}^{\ast} j_{i+1}'^{\ast}}
\nn \\
&\;\;= \;\; \label{plaqherm8}
\prod_{i=1}^n \gamma(j_{i+1}^{\ast}, k_i^{\ast}, j_i)
N^{s^{\ast} j_i'^{\ast} j_i^{\ast}}_{k_i j_{i+1}^{\ast} j_{i+1}'^{\ast}}
\end{align}
Using \eqref{eq:gamma-theta} on each of the brackets in the first line of \eqref{plaqherm8} verifies that equality.
This proves the first part of the proposition.
For the second part, recall that
$B_p^{g}=
  \sum_{s\in I_{g}}\bb(s)B_p^{s}$.
Thus
\begin{equation*}
(B_p^{g})^{\dagger}= \sum_{s\in I_{g}}\bb(s)B_p^{s^{\ast}} = B_p^{-g^{}} .
\end{equation*}
\end{proof}

\subsection{Properties of the non-semisimple Levin-Wen Hamiltonian} \label{subsec:LevinWen}

Define the \emph{plaquette operator}
\begin{equation}
B_p:= B_p^{g}B_p^{-g^{}}:\HHH(\Gamma,\col)\to\HHH(\Gamma,\col).
\end{equation}
 A priori, this operator depends on ${g}$, but the following lemma shows that it does not.

\begin{lemma}
The plaquette operator defined by $B_p=B_p^{g}B_p^{-g^{}}:\HHH(\Gamma,\col)\to\HHH(\Gamma,\col)$   does not depend on $g\in \Gr\setminus \X$.
\end{lemma}

\begin{proof}
This follows from Proposition~\ref{prop:plaqexp}.
\end{proof}


\begin{theorem} \label{prop:plaquette}
The plaquette operators have the following properties.
\begin{enumerate}
    \item The plaquette operators are projectors $B_p^2 = B_p$.
    \item The plaquette operators are mutually commuting:
\[
 B_p B_{p'} = B_{p'} B_p .
\]

    \item The operators $B_p$ are Hermitian with respect to the indefinite norm \eqref{eq:eta-inner}.
\end{enumerate}
\end{theorem}

\begin{proof}
For the first item, recall that
$B_p=B_p^{g}B_p^{-g^{}}$, so if we assume $g_1, g_2, g_1 +  g_2$  are generic, by Proposition \ref{prop:plaqexp} we have
\begin{align*}
B_p^2 &= B_p^{g_1}B_p^{-g_1^{}} B_p^{-g_2^{}}B_p^{g_2}
=B_p^{g_1} B_p^{-g_1^{} -g_2^{}} B_p^{g_2} \\
&=B_p^{-g_2^{}}  B_p^{g_2}
=B_p .
\end{align*}

The second item follows from Corollary \ref{corplaqcomm}.
The third item is a consequence of Proposition \ref{plaqhermprop}.

\end{proof}

\begin{remark}
The second claim above follows immediately from the topological interpretation of the plaquette operators in Section~\ref{sec:top}.
\end{remark}

Define a Hamiltonian
\begin{equation} \label{eq:Hamiltonian}
H =   - \sum_{p} B_p-\sum_{v} Q_v,
\end{equation}
or we can shift this Hamiltonian so that the ground state has zero energy.
\begin{equation} \label{eq:Hamiltonian0}
H =   \sum_{p} (1-B_p)+ \sum_{v} (1-Q_v).
\end{equation}

\begin{theorem}
The Hamiltonian from \eqref{eq:Hamiltonian0} has the following properties.
\begin{enumerate}
  \item By Theorem~\ref{prop:plaquette}, the Hamiltonian is Hermitian with respect to the indefinite inner product.

  \item The Hamiltonian is gapped and local.

  \item Since the Hamiltonian is comprised of mutually commuting projectors, it is clear that the spectrum of $H$ is $\Z_{\geq 0}$, with ground state  corresponding to the simultaneous +1-eigenspace of the operators $B_p$ and $Q_v$ over all vertices $v$ and plaquettes $p$.
\end{enumerate}
\end{theorem}

Observe, that the simultaneous $+1$-eigenspace of all the vertex operators $Q_v$ over all vertices defines a subspace of the Hilbert space $\HHH(\Gamma,\Phi)$ consisting of $\Gr$-colored graphs where each labelled trivalent graph is consistent with the branching rules $\delta_{ijk}$.
The simultaneous $+1$-eigenspace $B_p$ operators consists of those basis vectors $\ket{\psi}$ in the image of the projector $\Pi_p B_p$.  Indeed,  if any $B_p \ket{\psi} = 0$, this operator maps $\ket{\psi} \mapsto 0$ and fixed any vector with $B_p \ket{\psi} = \ket{\psi}$ for all plaquettes $p$.

We now show that the ground state of the Hamiltonian is described by the modified Turaev-Viro invariant coming from the
relative $\Gr$-spherical category.
\bigskip

\subsection{Topological invariance of ground state} \label{sec:top} In
this section we explain the topological origins of the plaquette
operators defined in Section~\ref{subsec:plaquette}.  The modified
Turaev-Viro TQFT from \cite{GP3} is a 2+1-dimensional TQFT defined on
a 3-manifold $M$ with the additional data of an isotopy class of an
embedded graph $Y$ in $M$ and a cohomology class
$[\col]\in H^1(M,\Gr)$.  Given a triangulation $\cal{T}_M$ of $M$, the
graph $Y$ is assumed to be Hamiltonian, meaning that it intersects
every vertex of the triangulation exactly once, and its edges are part
of the triangulation.

When the manifold $M$ has boundary, the triangulation of $M$ induces a
triangulation $\cal{T}$ on the boundary $\Sigma$ and the cohomology
class $[\col]\in H^1(M,\Gr)$ restricts
to $[\col_{|\Sigma}]\in H^1(\Sigma,\Gr)$.
In \cite{GP3}, it is assumed the graph $Y$ is embedded transversely to
the boundary surface.  Thus, the modified Turaev-Viro invariant for
the surface $\Sigma$ includes the additional data of a finite set of
marked points ${\mathsf m} \subset \Sigma$ on the surface
corresponding to where the graph intersects the surface.  Since the
graph must go through every vertex of the triangulation, in
\cite[Section 4]{GP3}, the notion of a ``oscillating path’’ on a
surface is used to define enriched cobordisms between triangulated
surfaces with oscillating paths.  The oscillating path allows vertices
of the triangulation of a surface to
be outside the set of marked points.
 In this paper we do not consider oscillating paths but instead assume that the set of vertices is exactly the set of marked points.   This seems natural from the perspective of the $B_p$ operators because as we will see, they generate the cylinder of the modified Turaev-Viro invariant, and they are only defined for a vertex at a marked point (also see Remark~\ref{rem:whatifd}).   It would be interesting to consider oscillating paths and graphs which are not transverse to the surface, but we leave this for future work.

 Next, we give a quick description of the modified Turaev-Viro
 invariant $TV$ with the assumptions and notation of this paper.  Let
 $(\Sigma,\T,\Gamma,\col)$ and $(\Sigma,\T',\Gamma',\col')$ be tuples of
 objects described in Section~\ref{subsec:def-state} with
 $[\col]=[\col']\in H^1(\Sigma,\G)$.  The vertices $\T_0 $ of the
 triangulation $\T$ and $\T'_0 $ of the triangulation $\T'$ are the
 marked points ${\mathsf m}$. Let $M$ be the cylinder
 $\Sigma\times [0,1]$.  Define a graph
 $Y=\T_0\times [0,1]={\mathsf m}\times [0,1]$ in $M$.  The cohomology
 class $[\col]=[\col'] \in H^1(\Sigma,\Gr)\simeq H^1(M,\Gr)$ induces a unique cohomology class on $M$.

 Let $(\T^M , \Y, \Phi^M)$ be an admissible Hamiltonian triangulation
 of $(M, Y)$ extending the triangulation $\T$ of $\Sigma\times\{0\}$
 and the triangulation $\T'$ of $\Sigma\times\{1\}$, see \cite[Section
 3.2]{GP3}.
 The cohomology class $[\col]\in H^1(M,\Gr)$ can always be represented
 by an admissible $\Gr$-coloring $\col^M$ of $\T^M$ which restricts to
 $\col$ on $\Sigma\times\{0\}$ and to $\col'$ on
 $\Sigma\times\{1\}$~\cite[Lemma 23]{GPT2}.  Likewise, every
 admissible $\Gr$-coloring of $\T^M$ gives rise to a representative of
 a cohomology class $[\col^M]\in H^1(M,\Gr)$.  In particular, $\Phi^M$
 is an admissible $\Gr$-coloring of $\T^M$ representing
 $[\Phi^M]\in H^1(M,\Gr)$ which is a map on the set of oriented edges
 of $\T^M$ such that for any oriented edge $e$ of $\T^M$,
 $\Phi^M(e)\in \Gr \setminus \X$, $\Phi^M(-e)=-\Phi^M(e)^{}$, and the
 sum of the values of $\Phi^M$ on the oriented edges forming the
 boundary of any face of $\T^M$ is zero. Also, $\Y$ is a set of edges
 of $\T^M$ such that the union of these edges is the graph $Y$ in $M$
 and all the vertices of $\T$ are contained in $Y$ (i.e.\ all the
 vertices are incident to an edge of $\Y$).

 A state $\phi$ of $\Phi^M$ is a map assigning to every oriented edge $e$ of $\T^M$ an element $\phi(e)\in I_{\Phi^M(e)}$ such that $\phi(-e) =-\phi(e)^{}$.
Given a state $\phi$ of $\Phi^M$, for each tetrahedron $T$ of $\T^M$ whose vertices $v_1$, $v_2$, $v_3$, $v_4$ are ordered so that the (ordered) triple of oriented edges $(\overrightarrow{v_1v_2}, \overrightarrow{v_1v_3}, \overrightarrow{v_1v_4})$ is positively oriented with respect to the orientation of $M$,  set $|T|_{\phi} := N^{ijk}_{lmn}$ where
\[
\left\{
  \begin{array}{lll}
    i = \phi(\overrightarrow{v_2v_1}) & j = \phi(\overrightarrow{v_3v_2}) & k = \phi(\overrightarrow{v_3v_1}) \\
    l = \phi(\overrightarrow{v_4v_3}) &  m = \phi(\overrightarrow{v_4v_1}) & n = \phi(\overrightarrow{v_4v_2})
  \end{array}
\right. .
\]

Let $\Phi_\T$ be a $\Gr$-coloring of $\T$ representing $[\Phi]\in H^1(\Sigma,\Gr)$.
Let $\HHH(\T,\Phi)=\bigoplus_{\sigma\in\St(\Phi)}\C\ket{\T,\sigma} $ be the span of all colorings corresponding to the states $\St(\Phi)$ of $\Phi$.  Similarly, $\HHH(\T',\Phi')=\bigoplus_{\sigma\in\St(\Phi')}\C\ket{\T',\sigma} $

Now fix a state $\sigma$ of $\Phi$ and let $\St(\Phi^M,\sigma)$ be the set of states of $\Phi^M$ which restrict to $\sigma$ on $\T=\T^M\cap (\Sigma\times\{0\})$.  The modified Turaev-Viro invariant $TV(M, Y, [\Phi^M]): \HHH(\T,\Phi)\to \HHH(\T',\Phi')$
is defined on vectors as
\begin{align} \label{eq:modifiedTV}
  &\ket{\T,\sigma}
  \\ \nn & \;\; \mapsto \!\!\!\!\sum_{\phi\in \St(\Phi^M,\sigma)}\prod_{e\in\mathcal E}\md(\phi(e))\prod_{e\in  \Y}\bb(\phi(e)) \!\!\!\prod_{T\in\T_3^M}\! |T|_{\phi} \cdot
  \ket{\T',\phi_{\big|\T'}}
\end{align}
where $\mathcal E$ is the set of unoriented edges of $\T^M$ that are
not in $\T$ nor in $\Y$ and $\T_3^M$ is the set of tetrahedra of
$\T^M$.

In \cite{GP3} it is shown that the mapping \eqref{eq:modifiedTV} is independent   of triangulation, coloring, and basis. However, it does depend on $M$, $Y=\m\times[0,1]$, the cohomology class $[\Phi^M]=[\Phi]$ and the triangulation of $\partial M$. For a fixed triple  $(\Sigma,[\Phi],\m)$, let us denote the associated mapping $C_{\T,\Phi}^{\T',\Phi'}$.
Since the composition of two cylinders is topologically equivalent to a single cylinder, \cite{GP3} implies that $C_{\T',\Phi'}^{\T'',\Phi''}C_{\T,\Phi}^{\T',\Phi'}=C_{\T,\Phi}^{\T'',\Phi''}$.  In particular the maps $C_{\T,\Phi}^{\T,\Phi}$ are projectors with canonically isomorphic image. 
Then the definition of the modified Turaev-Viro invariant on surfaces with marked points is
\begin{equation}\label{E:TVSurfaceImage}
TV(\Sigma,{\mathsf m},[\col]) := {\rm Im }  C_{\T,\Phi}^{\T,\Phi} ,
\end{equation}
for any choice of $\T,\Phi$.

A key distinction between the modified invariant \eqref{eq:modifiedTV} and the usual Turaev-Viro invariant is the use of a modified dimension $\md$ in place of the usual quantum dimension.  As noted, the modified dimension is real, but need not be positive.  Another difference is the role the Hamiltonian graph $Y$ plays, allowing the modified dimension $\md$ to be used on edges not contained in $Y$, while $\bb$ is used on those edges of $Y$.  This subtle difference is what enables non-semisimple categories to be used for state-sum topological invariants.

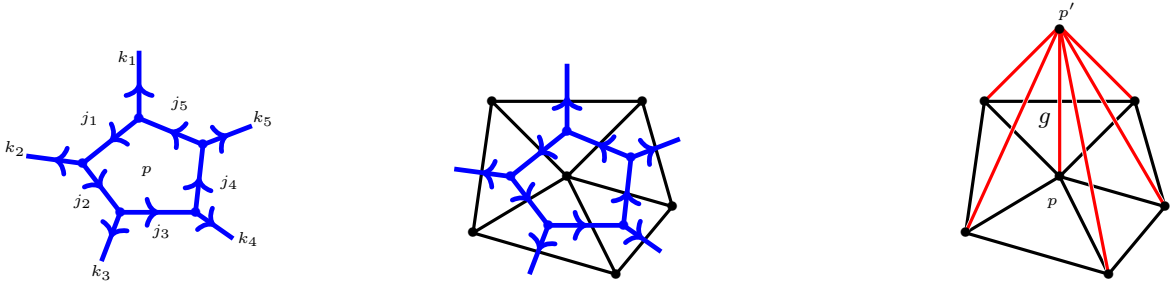
\begin{figure*}
\noindent
\subfloat[
 A plaquette $p$ in the dual graph $\Gamma$. \label{fig:a}]{
     \makebox[.28\textwidth]{ \begin{tikzpicture} [decoration={markings,
                        mark=at position 0.52 with {\arrow{>}};    }]
\node[blue,thick] at (-.75,0) {$\bullet$};
\node[blue,thick] at (.85,.25) {$\bullet$};
\node[blue,thick] at (0,.6) {$\bullet$};
\node[blue,thick] at (-.25,-.65) {$\bullet$};
\node[blue,thick] at (.75,-.65){$\bullet$};
    \draw[blue, very thick, line width=1.8pt,postaction={decorate}] (0,.6)to (-.75,0);
    \draw[blue, very thick, line width=1.8pt,postaction={decorate}]  (-.75,0)to (-.25,-.65);
    \draw[blue, very thick, line width=1.8pt,postaction={decorate}](-.25,-.65)to  (.75,-.65);
    \draw[blue, very thick, line width=1.8pt,postaction={decorate}] (.75,-.65)to (.85,.25);
    \draw[blue, very thick, line width=1.8pt,postaction={decorate}](.85,.25)to (0,.6);
    \draw[blue, very thick, line width=1.8pt,postaction={decorate}] (0,.6)to (0,1.5);
    \draw[blue, very thick, line width=1.8pt,postaction={decorate}] (.85,.25)to (1.5,.5);
    \draw[blue, very thick, line width=1.8pt,postaction={decorate}] (-.75,0) to (-1.5,.1);
    \draw[blue, very thick, line width=1.8pt,postaction={decorate}] (-.25,-.65) to (-.5,-1.3);
    \draw[blue, very thick, line width=1.8pt,postaction={decorate}] (.75,-.65) to (1.25,-1);
    \node at (.1,-.1) {$\scs p$};
    \node at (-.65,.6) {$\scs j_1$};
    \node at (-.75,-.5) {$\scs j_2$};
    \node at (.3,-.9) {$\scs j_3$};
    \node at (1.2,-.25) {$\scs j_4$};
    \node at (.55,.8) {$\scs j_5$};
        \node at (-.15,1.4) {$\scs k_1$};
        \node at (-1.65,.2) {$\scs k_2$};
        \node at (-.5,-1.45) {$\scs k_3$};
        \node at (1.45,-1) {$\scs k_4$};
        \node at (1.65,.6) {$\scs k_5$};
\node at (-2,0) {$\;$};
\node at (2,0) {$\;$};
\end{tikzpicture} }
}\hfill
\subfloat[The dual graph $\Gamma$ superimposed over the original triangulation $\cal{T}$. \label{fig:b}]{
 \makebox[.28\textwidth]{ \begin{tikzpicture} [decoration={markings,
                        mark=at position 0.52 with {\arrow{>}};    }]
\draw[very thick] (0,0) -- (-1,1);
\draw[very thick] (0,0) -- (1.4,-.4) ;
\draw[very thick] (0,0) -- (.65,-1.3) ;
\draw[very thick] (0,0) -- (-1.25,-.75);
\draw[very thick] (0,0) --  (1,1);
\draw[very thick] (-1,1) to (1,1) to (1.4,-.4) to (.65,-1.3) to (-1.25,-.75) to (-1,1);
\node at (0,0) {$\bullet$};
\node at (-1,1) {$\bullet$};
\node at (1,1) {$\bullet$};
\node at (1.4,-.4) {$\bullet$};
\node at (.65,-1.3) {$\bullet$};
\node at (-1.25,-.75) {$\bullet$};
\node[blue,thick] at (-.75,0) {$\bullet$};
\node[blue,thick] at (.85,.25) {$\bullet$};
\node[blue,thick] at (0,.6) {$\bullet$};
\node[blue,thick] at (-.25,-.65) {$\bullet$};
\node[blue,thick] at (.75,-.65){$\bullet$};
    \draw[blue, very thick, line width=1.8pt,postaction={decorate}] (0,.6)to (-.75,0);
    \draw[blue, very thick, line width=1.8pt,postaction={decorate}]  (-.75,0)to (-.25,-.65);
    \draw[blue, very thick, line width=1.8pt,postaction={decorate}](-.25,-.65)to  (.75,-.65);
    \draw[blue, very thick, line width=1.8pt,postaction={decorate}] (.75,-.65)to (.85,.25);
    \draw[blue, very thick, line width=1.8pt,postaction={decorate}](.85,.25)to (0,.6);
    \draw[blue, very thick, line width=1.8pt,postaction={decorate}] (0,.6)to (0,1.5);
    \draw[blue, very thick, line width=1.8pt,postaction={decorate}] (.85,.25)to (1.5,.5);
    \draw[blue, very thick, line width=1.8pt,postaction={decorate}] (-.75,0) to (-1.5,.1);
    \draw[blue, very thick, line width=1.8pt,postaction={decorate}] (-.25,-.65) to (-.5,-1.3);
    \draw[blue, very thick, line width=1.8pt,postaction={decorate}] (.75,-.65) to (1.25,-1);
\end{tikzpicture} }
} \hfill
\subfloat[ 3-dimensional tetrahedra obtained by adding a vertex over the vertex in $\cal{T}$ corresponding to the plaquette $p$ in $\Gamma$. \qquad \qquad \qquad \qquad \qquad \;  \label{fig:c}]{ \makebox[.35\textwidth]{
  \begin{tikzpicture}
\draw[very thick] (0,0) -- (-1,1);
\draw[very thick] (0,0) -- (1.4,-.4) ;
\draw[very thick] (0,0) -- (.65,-1.3) ;
\draw[very thick] (0,0) -- (-1.25,-.75);
\draw[very thick] (0,0) --  (1,1);
\draw[very thick] (-1,1) to (1,1) to (1.4,-.4) to (.65,-1.3) to (-1.25,-.75) to (-1,1);
\draw[line width=1.8pt, white] (0,0) -- (0,2);
\draw[line width=1.8pt, white] (.65,-1.3) -- (0,2);
\draw[line width=1.8pt, white] (-1.25,-.75) -- (0,2);
\draw[line width=1.8pt, white] (1.4,-.4) -- (0,2);
\draw[very thick, red] (0,0) -- (0,2);
\draw[very thick, red] (-1,1) -- (0,2);
\draw[very thick, red] (1.4,-.4) -- (0,2);
\draw[very thick, red] (.65,-1.3) -- (0,2);
\draw[very thick, red] (-1.25,-.75) -- (0,2);
\draw[very thick, red] (1,1) -- (0,2);
\node at (0,0) {$\bullet$};
\node at (0,1.95) {$\bullet$};
\node at (-1,1) {$\bullet$};
\node at (1,1) {$\bullet$};
\node at (1.4,-.4) {$\bullet$};
\node at (.65,-1.3) {$\bullet$};
\node at (-1.25,-.75) {$\bullet$};
\node at (-.1,-.35) {$\scs p$};
\node at (.1,2.2) {$\scs p'$};
\node at (-.2,.75) {$g$};
\node at (-3,0) {$\;$};
\node at (2,0) {$\;$};
\end{tikzpicture} \qquad   }
}
        \caption{The plaquette operator $B_p^{g}$ has a topological interpretation as the operator assigned to the tetrahedra in \eqref{fig:c} by the modified Turaev-Viro TQFT.  The initial edge labels in \eqref{fig:a} are mapped to new edge labels by the Turaev-Viro operator from \eqref{fig:c}.  The index $g\in \Gr\setminus\X$ labels the internal edge common in all the tetrahedra obtained by adding the point $p'$. }
        \label{fig:plaquette_operators}
\end{figure*}

Much like the semisimple case \cite{KKR,Kir-stringnet},
the plaquette operators $B_p^{g}$ from \eqref{eq:Balpha} used in the definition of the non-semisimple Levin-Wen Hamiltonian have a topological origin via the modified Turaev-Viro invariant.   Each plaquette in the dual graph $\Gamma$ corresponds to a vertex $p$ of the original triangulation see Figure~\eqref{fig:b}.   By adding an additional vertex $p'$ over the point $p$ and connecting it by an edge to $p$ and all vertices adjacent to $p$, we obtain a collection of tetrahedra $X_{p}$ over our original triangulation (see Figure~\eqref{fig:c}).  Labelling the edge connecting $p$ to $p'$ by $g$, we can then evaluate the tetrahedra $TV(X_{p})$ to produce an operator mapping the Hilbert space $\HHH(\Gamma,\col)$  to the Hilbert space $\HHH(\Gamma,\col+g.\delta p)$.  We have defined $B_p^{g}$ so that
\begin{equation}
B_p^{g} =  TV(X_{p}).
\end{equation}

Note that several of the key properties of the operators $B_p^{g}$ become apparent  from the topological perspective.  For example, the commutativity $[B_{p_1}^{g_1},B_{p_2}^{g_2}] =0$ follows from the independence of the modified Turaev-Viro invariant under change of triangulation.
\begin{equation}\label{E:BpBp'}
TV\left(
\hackcenter{ \begin{tikzpicture}[scale=0.7]
\draw[very thick] (0,0) -- (-1,1);
\draw[very thick] (0,0) -- (1.5,-.4) ;
\draw[very thick] (0,0) -- (.65,-1.3) ;
\draw[very thick] (0,0) -- (-1.25,-.75);
\draw[very thick] (0,0) --  (1,1);
\draw[very thick] (-1,1) to (1,1) to (1.5,-.4) to (.65,-1.3) to (-1.25,-.75) to (-1,1);
        \draw[very thick] (.65,-1.3) to (2.3,-.8) -- (1.5,-.4) ;
         \draw[very thick] (2.3,-.8) to (2.2,.6)  -- (1.5,-.4);
         \draw[very thick]   (2.2,.6) to (1,1) ;
\draw[line width=1.8pt, white] (0,0) -- (0,2);
\draw[line width=1.8pt, white] (.65,-1.3) -- (0,2);
\draw[line width=1.8pt, white] (-1.25,-.75) -- (0,2);
\draw[line width=1.8pt, white] (1.5,-.4) -- (0,2);
\draw[very thick, red] (0,0) -- (0,2);
\draw[very thick, red] (-1,1) -- (0,2);
\draw[very thick, red] (1.5,-.4) -- (0,2);
\draw[very thick, red] (.65,-1.3) -- (0,2);
\draw[very thick, red] (-1.25,-.75) -- (0,2);
\draw[very thick, red] (1,1) -- (0,2);
    \draw[line width=1.8pt, white]  (1.5,-.4) -- (1.5,2);
\draw[very thick, blue] (1.5,-.4) -- (1.5,2);
\draw[very thick, blue] (0,2) -- (1.5,2);
    \draw[line width=1.8pt, white] (.65,-1.3) -- (1.5,2);
\draw[very thick, blue] (.65,-1.3) -- (1.5,2);
    \draw[line width=1.8pt, white] (1,1) -- (1.5,2);
\draw[very thick, blue] (1,1) -- (1.5,2);
    \draw[line width=1.8pt, white] (2.3,-.8) -- (1.5,2);
\draw[very thick, blue] (2.3,-.8) -- (1.5,2);
    \draw[line width=1.8pt, white] (2.2,.6) -- (1.5,2);
\draw[very thick, blue] (2.2,.6) -- (1.5,2);
\node at (0,0) {$\bullet$};
\node at (0,1.95) {$\bullet$};
\node at (-1,1) {$\bullet$};
\node at (1,1) {$\bullet$};
\node at (1.5,-.4) {$\bullet$};
\node at (.65,-1.3) {$\bullet$};
\node at (-1.25,-.75) {$\bullet$};
         \node at (2.2,.6)  {$\bullet$};
        \node at (2.3,-.8)  {$\bullet$};
        \node at (1.5,2)  {$\bullet$};
\node at (-.1,-.35) {$\scs p_1$};
\node at (.1,2.3) {$\scs p_1'$};
\node at (1.55,2.3) {$\scs p_2'$};
\node at (1.5,-.75) {$\scs p_2$};
\node at (-.2,.75) {$g_1$};
\node at (1.7,.4) {$g_2$};
\end{tikzpicture}} \right)
 =
TV\!\left(
\hackcenter{ \begin{tikzpicture}[scale=0.7]
\draw[very thick] (0,0) -- (-1,1);
\draw[very thick] (0,0) -- (1.5,-.4) ;
\draw[very thick] (0,0) -- (.65,-1.3) ;
\draw[very thick] (0,0) -- (-1.25,-.75);
\draw[very thick] (0,0) --  (1,1);
\draw[very thick] (-1,1) to (1,1) to (1.5,-.4) to (.65,-1.3) to (-1.25,-.75) to (-1,1);
        \draw[very thick] (.65,-1.3) to (2.3,-.8) -- (1.5,-.4) ;
         \draw[very thick] (2.3,-.8) to (2.2,.6)  -- (1.5,-.4);
         \draw[very thick]   (2.2,.6) to (1,1) ;
    \draw[line width=1.8pt, white]  (1.5,-.4) -- (1.5,2);
\draw[very thick, blue] (1.5,-.4) -- (1.5,2);
    \draw[line width=1.8pt, white] (.65,-1.3) -- (1.5,2);
\draw[very thick, blue] (.65,-1.3) -- (1.5,2);
    \draw[line width=1.8pt, white] (1,1) -- (1.5,2);
\draw[very thick, blue] (1,1) -- (1.5,2);
    \draw[line width=1.8pt, white] (2.3,-.8) -- (1.5,2);
\draw[very thick, blue] (2.3,-.8) -- (1.5,2);
    \draw[line width=1.8pt, white] (2.2,.6) -- (1.5,2);
\draw[very thick, blue] (2.2,.6) -- (1.5,2);
        \draw[line width=1.8pt, white] (0,0) -- (0,2);
        \draw[line width=1.8pt, white] (-1.25,-.75) -- (0,2);
        \draw[line width=1.8pt, white] (1.5,2) -- (0,2);
\draw[very thick, red] (0,0) -- (0,2);
\draw[very thick, red] (-1,1) -- (0,2);
\draw[very thick, red] (1.5,2) -- (0,2);
\draw[very thick, red] (-1.25,-.75) -- (0,2);
\draw[very thick, red] (1,1) -- (0,2);
        \draw[line width=1.8pt, white] (0,0) -- (1.5,2);
\draw[very thick, blue] (0,0) -- (1.5,2);
    \draw[line width=1.8pt, white] (.65,-1.3) -- (0,2);
\draw[very thick, red] (.65,-1.3) -- (0,2);
\node at (0,0) {$\bullet$};
\node at (0,1.95) {$\bullet$};
\node at (-1,1) {$\bullet$};
\node at (1,1) {$\bullet$};
\node at (1.5,-.4) {$\bullet$};
\node at (.65,-1.3) {$\bullet$};
\node at (-1.25,-.75) {$\bullet$};
         \node at (2.2,.6)  {$\bullet$};
        \node at (2.3,-.8)  {$\bullet$};
        \node at (1.5,2)  {$\bullet$};
\node at (-.1,-.35) {$\scs p_1$};
\node at (.1,2.3) {$\scs p_1'$};
\node at (1.5,2.3) {$\scs p_2'$};
\node at (1.55,-.75) {$\scs p_2$};
\node at (-.2,.75) {$g_1$};
\node at (1.7,.4) {$g_2$};
\end{tikzpicture}} \right)
\end{equation}

%

\begin{remark} \label{rem:whatifd}
Those familiar with the usual Turaev-Viro construction and its connection to Levin-Wen models, might have expected the plaquette operators to be defined using the modified dimension in place of the usual quantum dimension so that $\tilde{B}_p^{g}=
  \sum_{s\in I_{g}}\md(s)B_p^{s}$, rather than \eqref{eq:Balpha}.  This is a key distinction in the non-semisimple version of the Turaev-Viro construction and is closely related to the role the embedded graph plays in the theory.  In Figure~\eqref{fig:c} the graph was assumed to be transverse to the surface at each vertex of the original triangulation and the resulting triangulation.  This means the new edge labelled $g$ now contains the embedded graph, and as such, the Turaev-Viro invariant utilized $\bb(s)$, rather than $\md(s)$ for $s\in I_g$.   The operators $\tilde{B}_p^{g}$ produce nilpotent operators that cannot be used to define projectors.
\end{remark}

\begin{theorem}
  The degenerate ground state of the Hamiltonian in \eqref{eq:Hamiltonian} defined on a surface $\Gamma$ is isomorphic to the vector space assigned to the surface by the modified TQFT from \cite{GP3} defined from the corresponding relative $\Gr$-spherical category:
  $$\operatorname{Ker} H\simeq TV(\Sigma,{\mathsf m},[\col])$$
  where $[\col]\in H^1(\Sigma,\Gr)$ and the cardinality of
  ${\mathsf m}$ is the number of vertices in the triangulation $\cal{T}$ of $\Sigma$.

 In particular, the ground state of the system is a topological invariant of the surface $\Sigma$ equipped with the set of points $\m$ and the cohomology class $[\col]\in H^1(\Sigma,\Gr)$.
\end{theorem}

\begin{proof}
Observe that the image of the operator $\prod_v Q_v$ on $\HHH(\Gamma,\Phi)$ consists of those states $\ket{\Gamma,\sigma}$ where the $\Gr$-coloring satisfies the additional constraint dictated by the branching rules $\delta_{ijk}$ at each vertex.    Denote this space by $\cal{H}'$.  The ground state of the Hamiltonian is then the   subspace of $\cal{H}'$  given by the image of $\prod B_p$.
 We identify the operator $\prod B_p$ with the operator associated to a cylinder $\Sigma \times [0,1]$ by the modified Turaev-Viro invariant.

As above,  let $M$ be the cylinder $\Sigma \times [0,1]$ with graph $Y={\mathsf m}\times [0,1]\subset M$.  The triangulation $\T$ on $\Sigma \times \{0\}$ can be extended to a  Hamiltonian  triangulation of $(M,Y)$.  First, add a new vertex $p'$ over a vertex $p$ of $\T$, making several  tetrahedra as in Figure \ref{fig:c} (here the edge between $p$ and $p'$ is in $\Y$).  Iterating this process (similar to the picture in \eqref{E:BpBp'}) for all vertices of $\T$, we obtain the desired triangulation.  Now computing the modified Turaev-Viro invariant with this triangulation, we have $TV(M,Y, [\Phi^M])=\prod B_p$, (see \eqref{eq:modifiedTV}).
 The theorem then follows from \eqref{E:TVSurfaceImage}.
\end{proof}

\section{Maps between Hilbert spaces for Bistellar operator} \label{sec:bistellar}
Here we give an example of a pseudo-unitary transformation   relating Hilbert spaces for different triangulations that share the same ground state.   These correspond to a change of triangulation according given by the 2-2 Pachner move.

\subsection{Local operators}

Let $(\Sigma,\{p\},\T,\Gamma,\col)$ be as above with $\col$ an
admissible coloring of $\Gamma$. Given an edge $e$ of the
triangulation, we can modify $\T$ and $\Gamma$ to $\T'$, $\Gamma'$ by
doing a bistellar move $I\leftrightarrow H$ on $e$.  The cocycle
condition of $\col$ implies that there is an unique coloring $\col'$
of $\Gamma'$ which coincide with $\col$ outside $\Gamma$. In
particular, it satisfies $[\col]=[\col']\in H^1(\Sigma,\Gr)$.  If
$\col'$ is admissible, we define the operator
$T:\H(\Gamma,\col)\to\H(\Gamma',\col')$ by
\begin{equation} \label{eq:Toper}
T\maps
\left| \hackcenter{\begin{tikzpicture}[   decoration={markings, mark=at position 0.6 with {\arrow{>}};}, scale =0.70]
    \draw[thick, blue, postaction={decorate}] (0,1) -- (.5,0);
    \draw[thick, blue, postaction={decorate}] (0,-1) -- (.5,0);
     \draw[thick, blue, postaction={decorate}] (.5,0) -- (1.5,0);
     \draw[thick, blue, postaction={decorate}] (2,1) -- (1.5,0);
     \draw[thick, blue, postaction={decorate}] (2,-1) -- (1.5,0);
    \node at (-.2,.75) {$\scs j_2$};
    \node at (-.2,-.75) {$\scs j_1$};
    \node at (2.2,.75) {$\scs j_3$};
    \node at (2.2,-.75) {$\scs j_4$};
   \node at (.75,.25){$\scs j_5$};
\end{tikzpicture}}
\right\rangle
\longrightarrow
\sum_{j_5'} \md({j_5'}) N^{j_1 j_2, j_5 }_{j_3 j_4^{\ast} j_5'}
\left| \hackcenter{\begin{tikzpicture}[  rotate=90, decoration={markings, mark=at position 0.6 with {\arrow{>}};}, scale =0.70]
    \draw[thick, blue, postaction={decorate}] (0,1) -- (.5,0);
    \draw[thick, blue, postaction={decorate}] (0,-1) -- (.5,0);
     \draw[thick, blue, postaction={decorate}] (1.5,0) -- (.5,0);
     \draw[thick, blue, postaction={decorate}] (2,1) -- (1.5,0);
     \draw[thick, blue, postaction={decorate}] (2,-1) -- (1.5,0);
    \node at (-.2,.75) {$\scs j_1$};
    \node at (-.2,-.75) {$\scs j_4$};
    \node at (2.2,.75) {$\scs j_2$};
    \node at (2.2,-.75) {$\scs j_3$};
   \node at (.75,.25){$\scs j_5'$};
\end{tikzpicture}} \;
\right\rangle
\end{equation}
In what follows below, we sometimes write $T_{j_5}$ for the operator above to clarify algebraic computations.

\subsection{Bistellar operators are Hermitian}
Much like the work in \cite{Hu_2012}, we can define certain local change of triangulation operators that act on our Hilbert space, though the computations in the non-semisimple case require modification.
%
%
%
Using the Hermitian inner product from \eqref{eq:eta-inner}, we have
\begin{align}
 & \left\langle \hackcenter{\begin{tikzpicture}[   decoration={markings, mark=at position 0.6 with {\arrow{>}};}, scale =0.70]
    \draw[thick, blue, postaction={decorate}] (0,1) -- (.5,0);
    \draw[thick, blue, postaction={decorate}] (0,-1) -- (.5,0);
     \draw[thick, blue, postaction={decorate}] (.5,0) -- (1.5,0);
     \draw[thick, blue, postaction={decorate}] (2,1) -- (1.5,0);
     \draw[thick, blue, postaction={decorate}] (2,-1) -- (1.5,0);
    \node at (-.2,.75) {$\scs j_2$};
    \node at (-.2,-.75) {$\scs j_1$};
    \node at (2.2,.75) {$\scs j_3$};
    \node at (2.2,-.75) {$\scs j_4$};
   \node at (.75,.25){$\scs j_5$};
\end{tikzpicture}}
\right.
\left| \hackcenter{\begin{tikzpicture}[   decoration={markings, mark=at position 0.6 with {\arrow{>}};}, scale =0.70]
    \draw[thick, blue, postaction={decorate}] (0,1) -- (.5,0);
    \draw[thick, blue, postaction={decorate}] (0,-1) -- (.5,0);
     \draw[thick, blue, postaction={decorate}] (.5,0) -- (1.5,0);
     \draw[thick, blue, postaction={decorate}] (2,1) -- (1.5,0);
     \draw[thick, blue, postaction={decorate}] (2,-1) -- (1.5,0);
    \node at (-.2,.75) {$\scs j_2$};
    \node at (-.2,-.75) {$\scs j_1$};
    \node at (2.2,.75) {$\scs j_3$};
    \node at (2.2,-.75) {$\scs j_4$};
   \node at (.75,.25){$\scs j_5$};
\end{tikzpicture}}
\right\rangle
\nn \\
\label{eq:inner-H}
&\;\; = \;\;
\gamma(j_1,j_2,j_5^{\ast}) \gamma(j_3,j_4,j_5) \beta(j_5) \frac{1}{\md(j_5)} , 
\end{align}
   Similarly, we can compute that for a state that is identical to the one above, except in a neighborhood of $j_5$ that we have
\begin{align}
& \left\langle
\hackcenter{\begin{tikzpicture}[  rotate=90, decoration={markings, mark=at position 0.6 with {\arrow{>}};}, scale =0.70]
    \draw[thick, blue, postaction={decorate}] (0,1) -- (.5,0);
    \draw[thick, blue, postaction={decorate}] (0,-1) -- (.5,0);
     \draw[thick, blue, postaction={decorate}] (1.5,0) -- (.5,0);
     \draw[thick, blue, postaction={decorate}] (2,1) -- (1.5,0);
     \draw[thick, blue, postaction={decorate}] (2,-1) -- (1.5,0);
    \node at (-.2,.75) {$\scs j_1$};
    \node at (-.2,-.75) {$\scs j_4$};
    \node at (2.2,.75) {$\scs j_2$};
    \node at (2.2,-.75) {$\scs j_3$};
   \node at (.75,.25){$\scs j_5'$};
\end{tikzpicture}}
\right|
\left.  \hackcenter{\begin{tikzpicture}[  rotate=90, decoration={markings, mark=at position 0.6 with {\arrow{>}};}, scale =0.70]
    \draw[thick, blue, postaction={decorate}] (0,1) -- (.5,0);
    \draw[thick, blue, postaction={decorate}] (0,-1) -- (.5,0);
     \draw[thick, blue, postaction={decorate}] (1.5,0) -- (.5,0);
     \draw[thick, blue, postaction={decorate}] (2,1) -- (1.5,0);
     \draw[thick, blue, postaction={decorate}] (2,-1) -- (1.5,0);
    \node at (-.2,.75) {$\scs j_1$};
    \node at (-.2,-.75) {$\scs j_4$};
    \node at (2.2,.75) {$\scs j_2$};
    \node at (2.2,-.75) {$\scs j_3$};
   \node at (.75,.25){$\scs j_5'$};
\end{tikzpicture}}
\right\rangle
\nn \\ \label{eq:I}
&= \gamma(j_2,j_3,j_5'^{\ast})\gamma(j_1,j_5'^{\ast}, j_4) \beta(j_5') \frac{1}{\md(j_5')} 
\end{align}


\begin{proposition}
 The operator $T$ is Hermitian and unitary.
\end{proposition}

\begin{proof}
Consider the matrix entries of $T^{\dagger}$ with respect to the Hermitian form:
\begin{align}
 & \left\langle \hackcenter{\begin{tikzpicture}[   decoration={markings, mark=at position 0.6 with {\arrow{>}};}, scale =0.70]
    \draw[thick, blue, postaction={decorate}] (0,1) -- (.5,0);
    \draw[thick, blue, postaction={decorate}] (0,-1) -- (.5,0);
     \draw[thick, blue, postaction={decorate}] (.5,0) -- (1.5,0);
     \draw[thick, blue, postaction={decorate}] (2,1) -- (1.5,0);
     \draw[thick, blue, postaction={decorate}] (2,-1) -- (1.5,0);
    \node at (-.2,.75) {$\scs j_2$};
    \node at (-.2,-.75) {$\scs j_1$};
    \node at (2.2,.75) {$\scs j_3$};
    \node at (2.2,-.75) {$\scs j_4$};
   \node at (.75,.25){$\scs j_5$};
\end{tikzpicture}}
\right|
T^{\dagger}
\left| \hackcenter{\begin{tikzpicture}[  rotate=90, decoration={markings, mark=at position 0.6 with {\arrow{>}};}, scale =0.70]
    \draw[thick, blue, postaction={decorate}] (0,1) -- (.5,0);
    \draw[thick, blue, postaction={decorate}] (0,-1) -- (.5,0);
     \draw[thick, blue, postaction={decorate}] (1.5,0) -- (.5,0);
     \draw[thick, blue, postaction={decorate}] (2,1) -- (1.5,0);
     \draw[thick, blue, postaction={decorate}] (2,-1) -- (1.5,0);
    \node at (-.2,.75) {$\scs j_1$};
    \node at (-.2,-.75) {$\scs j_4$};
    \node at (2.2,.75) {$\scs j_2$};
    \node at (2.2,-.75) {$\scs j_3$};
   \node at (.75,.25){$\scs j_5'$};
\end{tikzpicture}} \;
\right\rangle
\nn
\\
& \;\; := \;\;
\overline{
  \left\langle\hackcenter{\begin{tikzpicture}[  rotate=90, decoration={markings, mark=at position 0.6 with {\arrow{>}};}, scale =0.70]
    \draw[thick, blue, postaction={decorate}] (0,1) -- (.5,0);
    \draw[thick, blue, postaction={decorate}] (0,-1) -- (.5,0);
     \draw[thick, blue, postaction={decorate}] (1.5,0) -- (.5,0);
     \draw[thick, blue, postaction={decorate}] (2,1) -- (1.5,0);
     \draw[thick, blue, postaction={decorate}] (2,-1) -- (1.5,0);
    \node at (-.2,.75) {$\scs j_1$};
    \node at (-.2,-.75) {$\scs j_4$};
    \node at (2.2,.75) {$\scs j_2$};
    \node at (2.2,-.75) {$\scs j_3$};
   \node at (.75,.25){$\scs j_5'$};
\end{tikzpicture}}
\right|
T
\left| \hackcenter{\begin{tikzpicture}[   decoration={markings, mark=at position 0.6 with {\arrow{>}};}, scale =0.70]
    \draw[thick, blue, postaction={decorate}] (0,1) -- (.5,0);
    \draw[thick, blue, postaction={decorate}] (0,-1) -- (.5,0);
     \draw[thick, blue, postaction={decorate}] (.5,0) -- (1.5,0);
     \draw[thick, blue, postaction={decorate}] (2,1) -- (1.5,0);
     \draw[thick, blue, postaction={decorate}] (2,-1) -- (1.5,0);
    \node at (-.2,.75) {$\scs j_2$};
    \node at (-.2,-.75) {$\scs j_1$};
    \node at (2.2,.75) {$\scs j_3$};
    \node at (2.2,-.75) {$\scs j_4$};
   \node at (.75,.25){$\scs j_5$};
\end{tikzpicture}} \;
\right\rangle } \nn
\\
&\nn \refequal{\eqref{eq:I}} \overline{\md({j_5'}} )  \overline{N^{j_1 j_2, j_5 }_{j_3 j_4^{\ast} j_5'}} \cdot
 \gamma(j_2,j_3,j_5'^{\ast})\gamma(j_1,j_5'^{\ast}, j_4) \beta(j_5') \frac{1}{\overline{ \md(j_5') } } 
\\ \nn
&\;\; =  \overline{N^{j_1 j_2, j_5 }_{j_3 j_4^{\ast} j_5'}}
 \gamma(j_2,j_3,j_5'^{\ast})\gamma(j_1,j_5'^{\ast}, j_4) \beta(j_5') 
\end{align}

Compare with the matrix coefficients of $T$ which you get by rotating \eqref{eq:Toper} 90 degrees
\begin{align}
 &  \left\langle \hackcenter{\begin{tikzpicture}[   decoration={markings, mark=at position 0.6 with {\arrow{>}};}, scale =0.70]
    \draw[thick, blue, postaction={decorate}] (0,1) -- (.5,0);
    \draw[thick, blue, postaction={decorate}] (0,-1) -- (.5,0);
     \draw[thick, blue, postaction={decorate}] (.5,0) -- (1.5,0);
     \draw[thick, blue, postaction={decorate}] (2,1) -- (1.5,0);
     \draw[thick, blue, postaction={decorate}] (2,-1) -- (1.5,0);
    \node at (-.2,.75) {$\scs j_2$};
    \node at (-.2,-.75) {$\scs j_1$};
    \node at (2.2,.75) {$\scs j_3$};
    \node at (2.2,-.75) {$\scs j_4$};
   \node at (.75,.25){$\scs j_5$};
\end{tikzpicture}}
\right|
T
\left| \hackcenter{\begin{tikzpicture}[  rotate=90, decoration={markings, mark=at position 0.6 with {\arrow{>}};}, scale =0.70]
    \draw[thick, blue, postaction={decorate}] (0,1) -- (.5,0);
    \draw[thick, blue, postaction={decorate}] (0,-1) -- (.5,0);
     \draw[thick, blue, postaction={decorate}] (1.5,0) -- (.5,0);
     \draw[thick, blue, postaction={decorate}] (2,1) -- (1.5,0);
     \draw[thick, blue, postaction={decorate}] (2,-1) -- (1.5,0);
    \node at (-.2,.75) {$\scs j_1$};
    \node at (-.2,-.75) {$\scs j_4$};
    \node at (2.2,.75) {$\scs j_2$};
    \node at (2.2,-.75) {$\scs j_3$};
   \node at (.75,.25){$\scs j_5'$};
\end{tikzpicture}} \;
\right\rangle
\nn \\ \nn
\;\; &= \;\;
\md({j_5})  N^{j_2 j_3, j_5' }_{j_4 j_1^{\ast} j_5^{\ast}}   \cdot \gamma(j_1,j_2,j_5^{\ast}) \gamma(j_3,j_4,j_5) \beta(j_5) \frac{1}{\md(j_5)}  
\end{align}

Therefore, we have shown that the $T$ operator is Hermitian provided the following equality holds.
\begin{align}
& \overline{N^{j_1 j_2, j_5 }_{j_3 j_4^{\ast} j_5'}}
 \gamma(j_2,j_3,j_5'^{\ast})\gamma(j_1,j_5'^{\ast}, j_4) \beta(j_5') \nn \\
 \label{eq:T1-want} &\;\; =\;\;
  N^{j_2 j_3, j_5' }_{j_4 j_1^{\ast} j_5^{\ast}}     \gamma(j_1,j_2,j_5^{\ast}) \gamma(j_3,j_4,j_5) \beta(j_5).
\end{align}
Recall from \eqref{eq:6j} that
\begin{align}
&\overline{N^{j_1 j_2, j_5 }_{j_3 j_4^{\ast} j_5'}}
 =
N^{j_2^{\ast} j_1^{\ast}j_5^{\ast}}_{j_4^{\ast} j_3 j_5'}
	\gamma(j_1,j_2,j_5^{\ast})   \gamma(j_1^{\ast}, j_4^{\ast}, j_5'^{\ast}) \gamma(j_2^{\ast} ,j_5', j_3^{\ast})
\nn \\
& \label{eq:Nbar-eq}
\times \gamma(j_5,j_4,j_4 )
\beta(j_1)\beta(j_2)\beta(j_1)\beta(j_3)\beta(j_4)\beta(j_5)\beta(j_5')
\end{align}
and using the symmetries of the tetrahedron we have
\begin{equation} \label{eq:Nsymm}
N^{j_2 j_3 j_5'}_{j_4 j_1^{\ast} j_5^{\ast}}
= N^{j_3 j_5' j_2}_{j_1^{\ast} j_5^{\ast} j_4} = N^{j_2^{\ast} j_1^{\ast} j_5^{\ast}}_{j_4^{\ast} j_4 j_5'}.
\end{equation}
Hence, we have
\begin{align}
&  \overline{N^{j_1 j_2, j_5 }_{j_3 j_4^{\ast} j_5'}}
 \gamma(j_2,j_3,j_5'^{\ast})\gamma(j_1,j_5'^{\ast}, j_4) \beta(j_5')
 \;\;  \;\;
  \nn
  \\
 & \quad \refequal{\eqref{eq:Nbar-eq}}
 N^{j_2^{\ast} j_1^{\ast}j_5^{\ast}}_{j_4^{\ast} j_3 j_5'}
	\gamma(j_1,j_2,j_5^{\ast})   \gamma(j_1^{\ast}, j_4^{\ast}, j_5'^{\ast}) \gamma(j_2^{\ast} ,j_5', j_3^{\ast}) \nn \\
& \qquad  \times  \gamma(j_5,j_3,j_4 )
  \cdot
 \gamma(j_2,j_3,j_5'^{\ast})\gamma(j_1,j_5'^{\ast}, j_4) \beta(j_5')
 \nn \\
 & \qquad  \times \beta(j_1)\beta(j_2)\beta(j_1)\beta(j_3)\beta(j_4)\beta(j_5)\beta(j_5')
\nn  \\ \nn
 & \quad \refequal{\eqref{eq:Nsymm}}      N^{j_2 j_3 j_5'}_{j_4 j_1^{\ast} j_5^{\ast}}
	\gamma(j_1,j_2,j_5^{\ast})  \gamma(j_5,j_3,j_4 )  \beta(j_5) \cdot
\nn \\
& \qquad \times \left(  \gamma(j_1^{\ast}, j_4^{\ast}, j_5'^{\ast}) \gamma(j_1,j_5'^{\ast}, j_4)   \beta(j_1) \beta(j_4) \beta(j_5')\right)
	\nn \\ \nn
	&\qquad  \times
	\left(  \gamma(j_2^{\ast} ,j_5', j_3^{\ast})  \gamma(j_2,j_3,j_5'^{\ast}) \beta(j_2)\beta(j_3)  \beta(j_5')  \right)
\nn  \\ \nn
 & \quad =    N^{j_2 j_3 j_5'}_{j_4 j_1^{\ast} j_5^{\ast}}
	\gamma(j_1,j_2,j_5^{\ast})  \gamma(j_5,j_3,j_4 )
	 \beta(j_5)
\end{align}
where we have used \eqref{eq:gamma-theta} twice in the last equality.  Thus,  \eqref{eq:T1-want} holds and   $T=T^{\dagger}$.

Next we will show that $T^{\dagger} T_1= \Id$.
\begin{align}
& T^{\dagger}T
\left| \hackcenter{\begin{tikzpicture}[   decoration={markings, mark=at position 0.6 with {\arrow{>}};}, scale =0.7]
    \draw[thick, blue, postaction={decorate}] (0,1) -- (.5,0);
    \draw[thick, blue, postaction={decorate}] (0,-1) -- (.5,0);
     \draw[thick, blue, postaction={decorate}] (.5,0) -- (1.5,0);
     \draw[thick, blue, postaction={decorate}] (2,1) -- (1.5,0);
     \draw[thick, blue, postaction={decorate}] (2,-1) -- (1.5,0);
    \node at (-.2,.75) {$\scs j_2$};
    \node at (-.2,-.75) {$\scs j_1$};
    \node at (2.2,.75) {$\scs j_3$};
    \node at (2.2,-.75) {$\scs j_4$};
   \node at (.75,.25){$\scs j_5$};
\end{tikzpicture}}
\right\rangle
 =
\sum_{j_5'} \md({j_5'}) N^{j_1 j_2 j_5 }_{j_3 j_4^{\ast} j_5'}
T^{\dagger} \left| \hackcenter{\begin{tikzpicture}[  rotate=90, decoration={markings, mark=at position 0.6 with {\arrow{>}};}, scale =0.7]
    \draw[thick, blue, postaction={decorate}] (0,1) -- (.5,0);
    \draw[thick, blue, postaction={decorate}] (0,-1) -- (.5,0);
     \draw[thick, blue, postaction={decorate}] (1.5,0) -- (.5,0);
     \draw[thick, blue, postaction={decorate}] (2,1) -- (1.5,0);
     \draw[thick, blue, postaction={decorate}] (2,-1) -- (1.5,0);
    \node at (-.2,.75) {$\scs j_1$};
    \node at (-.2,-.75) {$\scs j_4$};
    \node at (2.2,.75) {$\scs j_2$};
    \node at (2.2,-.75) {$\scs j_3$};
   \node at (.75,.25){$\scs j_5'$};
\end{tikzpicture}}
\right\rangle
\nn \\ \nn
&\; = \;
\sum_{j_5'} \md({j_5'}) N^{j_1 j_2 j_5 }_{j_3 j_4^{\ast} j_5'}\;
 \sum_{j_5''}\md({j_5''})  N^{j_2 j_3 j_5' }_{j_4 j_1^{\ast} j_5''^{\ast}}
\;
\left| \hackcenter{\begin{tikzpicture}[   decoration={markings, mark=at position 0.6 with {\arrow{>}};}, scale =0.7]
    \draw[thick, blue, postaction={decorate}] (0,1) -- (.5,0);
    \draw[thick, blue, postaction={decorate}] (0,-1) -- (.5,0);
     \draw[thick, blue, postaction={decorate}] (.5,0) -- (1.5,0);
     \draw[thick, blue, postaction={decorate}] (2,1) -- (1.5,0);
     \draw[thick, blue, postaction={decorate}] (2,-1) -- (1.5,0);
    \node at (-.2,.75) {$\scs j_2$};
    \node at (-.2,-.75) {$\scs j_1$};
    \node at (2.2,.75) {$\scs j_3$};
    \node at (2.2,-.75) {$\scs j_4$};
   \node at (.75,.25){$\scs j_5''$};
\end{tikzpicture}}
\right\rangle
\end{align}
(The above works since we have established that $T^{\dagger} = T$.)
Observe that using the symmetries from \eqref{eq:symm}
we have
\begin{equation}
  N^{j_2 j_3 j_5'}_{j_4 j_1^{\ast} j_5''^{\ast}} =
  N^{j_5' j_4 j_1^{\ast}}_{j_5'' j_2 j_3^{\ast}} =
  N^{j_1^{\ast} j_5'' j_2}_{j_3 j_5' j_4^{\ast}} =
  N^{j_5'' j_2^{\ast} j_1}_{j_5' j_4^{\ast} j_3}.
\end{equation}
So unitarity follows from \eqref{eq-ortho} since
\begin{align*}
 &\sum_{j_5''}\sum_{j_5'} \md({j_5''}) \md({j_5'}) N^{j_1 j_2 j_5 }_{j_3 j_4^{\ast} j_5'}\;
N^{j_2 j_3, j_5' }_{j_4 j_1^{\ast} j_5''^{\ast}}
\\ & \quad =  \;
 \sum_{j_5''}\sum_{j_5'} \md({j_5''}) \md({j_5'}) N^{j_1 j_2 j_5 }_{j_3 j_4^{\ast} j_5'}\;
N^{j_5'' j_2^{\ast} j_1}_{j_5' j_4^{\ast} j_3} \\
& \quad
 = \;
 \sum_{j_5''}\md({j_5''}) \left(\sum_{j_5'} \md({j_5'}) N^{j_1 j_2 j_5 }_{j_3 j_4^{\ast} j_5'}\;
 N^{j_5'' j_2^{\ast} j_1}_{j_5' j_4^{\ast} j_3} \right)
 \\
& \quad  = \;
\sum_{j_5''}\md({j_5''}) \left(  \frac{\delta_{j_5,j_5''}}{\md({j_5})} \delta_{j_1 j_2 j_5''^{\ast}} \delta_{j_5 j_3 j_4}\right) = \delta_{j_5,j_5''}.
\end{align*}
\end{proof}

\subsection{Local operators commute with the Hamiltonian}

\begin{lemma}
  When all involved $\Gr$-coloring are admissible, plaquette and
  bistellar operators commute:
  $$B_p^{g}T_{j_n}=T_{j_n'}B_p^{g}.$$
\end{lemma}
\begin{proof}
By direct computation we have
\begin{align} \nn
   B_p^{s}T_{j_n} &=
 \sum_{z_{n-1}} \sum_{j_1',\ldots,j_{n-1}'} \prod_{i=1}^{n-2}
 \md(z_{n-1}) N^{k_{n-1}^{\ast} j_{n-1} j_n}_{j_1^{\ast} k_n z_{n-1}}
 \\ \nn & \qquad \times
\md(j_i') N^{j_i' s j_i}_{j_{i+1}^{\ast} k_i j_{i+1}'^{\ast}}
\md(j_{n-1}') N^{j_{n-1}' s j_{n-1}}_{j_1^{\ast} z_{n-1} j_1'^{\ast}}
\\ \nn
T_{j_n'} B_p^s &=
 \sum_{z_{n-1}} \sum_{j_1',\ldots,j_{n-1}', j_n'} \prod_{i=1}^{n-2}
 \md(z_{n-1})
\md(j_i') N^{j_i' s j_i}_{j_{i+1}^{\ast} k_i j_{i+1}'^{\ast}}
\\ \nn
& \nn \quad \times
\md(j_{n-1}')
\md(j_n') N^{j_{n-1}' s j_{n-1}}_{j_n^{\ast} k_{n-1} j_n'^{\ast}}
 N^{j_{n}' s j_{n}}_{j_1^{\ast} k_{n} j_1'^{\ast}}
 N^{k_{n-1}^{\ast} j_{n-1}' j_n'}_{j_1'^{\ast} k_n z_{n-1}}
\end{align}
 In order to verify the equality, we need to check
\begin{align*}
 & N^{k_{n-1}^{\ast} j_{n-1} j_n}_{j_1^{\ast} k_n z_{n-1}}
 N^{j_{n-1}' s j_{n-1}}_{j_1^{\ast} z_{n-1} j_1'^{\ast}}
\\ & \qquad  =
 \sum_{j_n'}
 \md(j_n')
 N^{j_{n-1}' s j_{n-1}}_{j_n^{\ast} k_{n-1} j_n'^{\ast}}
 N^{j_{n}' s j_{n}}_{j_1^{\ast} k_{n} j_1'^{\ast}}
 N^{k_{n-1}^{\ast} j_{n-1}' j_n'}_{j_1'^{\ast} k_n z_{n-1}}
\end{align*}
which follows from \eqref{eq:symm} and \eqref{eq:pent}.

\end{proof}

\section{Example: quantum $\slt$ at roots of unity} \label{sec:coeff}
In this section we provide a non-trivial example of a Hermitian relative $\Gr$-spherical category.  It is a certain category of representations of quantum $\slt$ where the quantum parameter $q$ is specialized to a root of unity.
There are two versions of the quantum group that we consider.  The first is the unrolled quantum group.  This larger algebra plays an auxiliary role here.  It was proved in \cite{GLPMS} that a category of representations for this algebra is Hermitian.  This leads to some nice properties about $6j$-symbols for the category of representations.  However, the category of representations for the unrolled quantum group does not have certain finiteness properties that are required for a relative $\Gr$-spherical category, which are necessary to construct the modified Turaev-Viro invariant.  The semi-restricted quantum group for $\slt$ does give rise to a relative $\Gr$-spherical category.
There is a forgetful functor between the categories allowing us to transport the desired properties of the $6j$-symbols for the unrolled quantum group over to the semi-restricted quantum group.

 Throughout the section, fix an odd positive integer $r=2r'+1$
and
 let $q=e^\frac{\pi\sqrt{-1}}{r}$ be a
 $2r^{th}$-root of unity.  Set
 \[
 [n]:=\frac{q^n-q^{-n}}{q-q^{-1}}, \quad \quad \{n\} := q^n - q^{-n} .
 \]

\subsection{Unrolled quantum group}

Let $\UsltH$ be the $\C$-algebra given by generators $E, F, K, K^{-1}, H$
and relations:
\begin{align}
  KK^{-1} &=K^{-1}K=1, \;\; & KEK^{-1} &=q^2E,  \;\; \nn \\
  [E,F] &=\frac{K-K^{-1}}{q-q^{-1}} &  KFK^{-1} &=q^{-2}F, \label{E:RelDCUqsl}\\
    HK &=KH, \;\;
  [H,E]=2E,    &[H,F]&=-2F. \label{Hrels}
\end{align}
The algebra $\UsltH$ is a Hopf algebra where the coproduct, counit and
antipode are defined by
\begin{align}
  \Delta(E)\!&=\! 1\otimes E + E\otimes K, \!\!\!\!
  &\varepsilon(E)\!&= 0,
  \!\!\!\!&S(E)&\!=\!-EK^{-1},  \label{Edata}
  \\
  \Delta(F)\!&=\!K^{-1} \otimes F + F\otimes 1,\!\!\!\!
  &\varepsilon(F)\!&=0, \!\!\!\!& S(F)&\!=\!-KF, \label{Fdata}
    \\
  \Delta(K)\!&=\!K\otimes K, \!\!\!\!
  &\varepsilon(K)\!&=1, \!\!\!\!
  & S(K)&\!=\!K^{-1}, \label{Kdata} \\
    \Delta(H)\!&=\!H\otimes 1 + 1 \otimes H, \!\!\!\!
  & \varepsilon(H)\!&=0, \!\!\!\!
  &S(H)&\!=\!-H. \label{Hdata}
\end{align}
Define the {\it unrolled quantum group} $\Ubar$ to be the Hopf algebra $\UsltH$ modulo the relations
$E^r=F^r=0$.

We now consider the following class of finite-dimensional highest weight modules for $\Ubar$.
For each $\alpha\in \R$, we let $V_\alpha$ be the $r$-dimensional
highest weight $\Ubar$-module of highest weight $\alpha + r-1$.  The
module $V_\alpha$ has a basis $\{v_0,\ldots,v_{r-1}\}$ whose action is
given by
\begin{gather}
H.v_d=(\alpha + r-1-2d) v_d,\quad E.v_d= \frac{\{ d \} \{d-\alpha\}}{\{1\}^2}v_{d-1} , \nn \\
\quad F.v_d=v_{d+1}. \label{E:BasisV}
\end{gather}
For all $\alpha\in \R$, the quantum dimension of $V_\alpha$ is zero, but
is possible to define a modified dimension that is nonzero.

For $\alpha$ and $\beta$ sufficiently generic, $V_{\alpha} \otimes V_{\beta}$ decomposes into a direct sum of irreducible modules.
\begin{proposition} \cite[Theorem 5.2]{CGP2} \label{decompprop}
If $\alpha, \beta \in (\R\setminus\Z) \cup r\Z$, and $\alpha+\beta \notin \Z$, then
\[
V_{\alpha} \otimes V_{\beta} \cong \bigoplus_{\gamma=-r'}^{r'} V_{\alpha+\beta+2\gamma}.
\]
\end{proposition}

\begin{remark}
The modules $V_{\alpha}$ also exist when $\alpha \in \C$.  We restrict to $\R$ here in order to ensure that the required Hermitian structures exist.
\end{remark}

The category of $\Ubar$-modules generated by the $V_{\alpha}$ for $\alpha \in \R$ possesses a set of 6$j$-symbols $N^{ijk}_{klm}$ that have explicit formulas given in \cite[Theorem 29]{GP1}.

\begin{proposition} \label{prop:mod-6j}
The following identities hold. \hfill
\begin{enumerate}
 \item \label{prop:symm1}$ N^{i j k}_{l m n} = N^{ j k^{\ast}  i^{\ast}}_{m n l}
= N^{k l m}_{n^{\ast} i j^{\ast}}$;

 \item \label{prop:pent1}
$\sum_j \md(j) N^{j_1, j_2, j_5}_{j_3,j_6, j} N^{j_1 j j_6}_{j_4 j_0 j_7} N^{j_2 j_3 j}_{j_4 j_7 j_8}
=
N^{j_5, j_3, j_6}_{j_4 j_0 j_8} N^{j_1 j_2 j_5}_{j_8 j_0 j_7}$ ;

\item \label{prop-ortho1}
$\sum_n \md(n) N^{i j p}_{l m n} N^{k j^{\ast} i}_{n m l}
= \frac{\delta_{k,p}}{\md(k)} \delta_{ij k^{\ast}} \delta_{k l m^{\ast}}
$
where
$\delta_{ijk}$ is the branching rule that equals 1 if the triple is allowed and 0 otherwise.
\end{enumerate}
\end{proposition}

\begin{proof}
 The first claim is from \cite[Section 4]{GPT2}.  The second claim is \cite[Theorem 7]{GPT2}.   The last claim is \cite[Theorem 8]{GPT2}.
\end{proof}

For $i,j,k \in (\R - \Z) \cup r\Z$, with $i+j+k=2F $, define coefficients
\begin{align}
\md(i) &=\md_0\frac{q^i-q^{-i}}{q^{ri}-q^{-ri}} \\
   \beta(i) &=  \frac{ \{1 \}^{2(r-1)} \md(i)}{r \md_0}
\\
\gamma(i,j,k) &=
\beta(-k)
\prod_{f=1}^{r'+F} \{i-f\} \{j -f\} \nn \\ & \quad \times  \prod_{g=1}^{r'+F} \frac{\{2r'-g+1\}\{-k+g\}}{\{1\}},
\end{align}
where we set $\md_0 = r$.

Using the explicit formulas for the 6$j$-symbols from  \cite[Theorem 29]{GP1} and the definitions above, one can verify the following result.

\begin{proposition} \label{prop:herm-6j}
The following identities hold. \hfill
\begin{enumerate}
 \item 
 $\gamma(i,j,k)\gamma(k^{\ast},j^{\ast},i^{\ast})\beta(i)\beta(j)\beta(k)=1$;
\item 
 $\wb{(N^{j_1 j_2 j_3 }_{j_4j_5j_6})} = N^{j_2^{\ast} j_1^{\ast} j_3^{\ast}}_{j_5 j_4j_6}\gamma(j_1,j_2,j_3^{\ast})\gamma(j_1^{\ast},j_5,j_6^{\ast})$ 
  $  \times\gamma(j_2^{\ast},j_6,j_4^{\ast})
     \gamma(j_3,j_4,j_5^{\ast})\prod_i\beta(j_i)$.
\end{enumerate}
\end{proposition}

In Appendix~\ref{app:hermitian} we give a conceptual derivation of Proposition~\ref{prop:herm-6j} by studying elementary intertwiners between $\Ubar$-modules. There is a ribbon category $\mathcal{D}^{\dagger}$ of $\Ubar$-modules that are generated by a certain set
of objects $A$ containing the $V_{i}$ for $i \in \R$ (see \cite[Equation 18]{GLPMS}).  In \cite[Theorem 4.18]{GLPMS} it is shown that the  category $\mathcal{D}^{\dagger}$ is a Hermitian ribbon category, so it possesses an indefinite inner product $\la \cdot, \cdot \ra$ allowing one to compute Hermitian adjoints $f^{\dagger}$ for any morphism in $\mathcal{D}^{\dagger}$.

 In this formulation the coefficients $\beta(i)$ and $\gamma(i,j,k)$ have natural interpretations via the Hermitian structure.  There is an isomorphism
$w_{i} \colon V_{i} \rightarrow V_{-i}^{\ast}$
determined by mapping a highest weight vector
$v_0 \in V_{i}$ by $w_{i}(v_0)=-q^{-1} y_{r-1}^{\ast}$,
where $y_{r-1}$ is a lowest weight vector of $V_{-i}$.
The coefficient $\beta(i)$ is the unique constant satisfying
\begin{equation} 
w^{\dagger} = \beta(i) w^{-1} .
\end{equation}

Likewise, it is possible to specify a specific basis vector $h_{ijk}$ in the 1-dimensional space $\Hom(1, V_i \otimes V_j \otimes V_k)$.  Then using the pivotal structure and the isomorphisms $w_{\alpha}$, these distinguished basis vectors determine basis vectors
\begin{align*}
h^{ijk} &\in  \Hom(V_i \otimes V_j \otimes V_k, 1)\cong \Hom(V_i \otimes V_j \otimes V_k, 1)
\\& \cong
\Hom(1, V_k^{\ast} \otimes V_j^{\ast} \otimes V_i^{\ast} )
\\ &  \cong
\Hom(1, V_{k^{\ast}} \otimes V_{j^{\ast}} \otimes V_{i^{\ast}} ).
\end{align*}
The coefficients $\gamma(i,j,k)$ are the unique coefficients satisfying
\begin{equation}
  h_{ijk}^{\dagger} = \gamma(i,j,k)h^{ijk}.
\end{equation}

\subsection{Semi-restricted quantum group and relative spherical data}

 Let $\Gr$ be the abelian group
 $\R/2\Z$.  Let $\X$ be the subset
 $\Z/2\Z$ in $\Gr$.  Consider the natural morphisms
 $\R \to \Gr$.  We
 denote the image of $\alpha \in \R$ in
 $ \Gr$ as  $\bar\alpha$.

Let $\tilde{\Gr}$ be the group
  $\R/2r\Z$ and $\tilde{\X}$ be the subset
 $\Z/2r\Z$ in $\tilde{\Gr}$.  Consider the natural morphisms
 $\R \to \tilde{\Gr} \to \Gr$.  We
 denote the image of $\alpha \in \R$ in
 $\tilde{\Gr}$ as $\tilde\alpha$.
 For $\bar\alpha\in \Gr $  let
\begin{align}
  I_{\bar\alpha}&=\{\tilde{\alpha} \in \tilde{\Gr} \; | \; \bar{\alpha} = {\rm im}\tilde\alpha \text{  under the morphism } \tilde{\Gr}\to\Gr    \}  \nn\\
\nn
&=\{\tilde{\alpha}, \tilde{\alpha} +\tilde{2}, \dots, \tilde{\alpha} +\tilde{2(r-1)}     \} .
\end{align}
Note that $\tilde{\Gr}$ comes with an involution $\tilde\alpha\mapsto \tilde\alpha^{\ast}=-\tilde\alpha$ with $I_{-\bar\alpha}=(I_{\bar\alpha})^{\ast}$.

Define
$$\md(\tilde\alpha)=\md_0\frac{q^\alpha-q^{-\alpha}}{q^{r\alpha}-q^{-r\alpha}}$$
for a fixed real number $\md_0$.
If $\tilde i \in  I_{\bar\alpha}$, then we define the constant quantity
\[
 \bb(\tilde i) := \frac{1}{r^2}.
\]

Let $\Uq$ be the $\C$-algebra generated by $E,F,K,K^{-1}$ subject to relations \eqref{E:RelDCUqsl}.
Let $\UqMed$ be the algebra $\Uq$ modulo the relations
$E^r=F^r=0$.  This algebra becomes a Hopf algebra via \eqref{Edata}--\eqref{Kdata}.

The category $\UqMed \dmod$ of weight $\UqMed$-modules is a relative $\Gr$-spherical category with basic data leading to  $6j$-symbols
$ N^{\tilde{i} \tilde{j} \tilde{k}}_{\tilde{l} \tilde{m} \tilde{n}} $, see Section 9 of \cite{GPT2}.  As we now explain, these $6j$-symbols are essentially the same as the values of the
$6j$-symbols derived from the category $\Ubar \dmod$.
Consider the highest weight $\UqMed$-module $V_{\tilde i}$ determined by a highest weight vector of weight $\tilde i +\tilde{r-1}\in \R /2r\Z$.  This module is simple if $\tilde i \in \tilde{\Gr}\setminus \tilde{\X}$.

There is a forgetful functor
\begin{equation} \label{forgetfunctor}
\mathcal{F} \colon \Ubar \dmod \rightarrow \UqMed \dmod .
\end{equation}
This functor maps $V_{i+2rk}$ to $V_{\tilde i}$ for all $k\in \Z$.  The branching rules $\delta_{\tilde i \tilde j  \tilde k}$ are defined to be
\begin{align*}
\delta_{\tilde i \tilde j \tilde k} &=  \dim \Hom_{\UqMed \dmod}(1, V_{\tilde i} \otimes V_{\tilde j}\otimes V_{\tilde k}
)
\\
&= \left\{
    \begin{array}{ll}
      1, & \hbox{$\tilde i   + \tilde j +\tilde k \in \{ -2r', -2r'+2, \dots, 2r'\}$,} \\
      0, & \hbox{otherwise.}
    \end{array}
  \right.
\end{align*}
Note it is possible that $\delta_{\tilde i \tilde j  \tilde k}=1$ but
$\delta_{i j k}=0$.  However, in such a case there exists lifts (via $\mathcal{F} $) of $ V_{\tilde i}, V_{\tilde j}$, and  $V_{\tilde k}$ such that the corresponding hom-space in $\Ubar \dmod$ is non-zero and $\delta_{i j k}=1$.  Moreover, there exists lifts of $i, j, k, l, m, n$ of $\tilde{i}, \tilde{j}, \tilde{k}, \tilde{l}, \tilde{m}, \tilde{n}$ such that
\[
 N^{\tilde{i} \tilde{j} \tilde{k}}_{\tilde{l} \tilde{m} \tilde{n}}=
 N^{i j k}_{l m n}
\]
see    \cite[Remark 21]{GPT2}.

Also note that the quantities $\beta(i)$ and $\gamma(i,j,k)$ computed in Sections \ref{subsecbeta} and \ref{gammasec} are independent of  the pre-images of $i$,$j$, and $k$ in $\R / 2r\Z$.
Thus, by choosing any lifts, these formulas define $\beta(\tilde{i})$ and $\gamma(\tilde{i}, \tilde{j}, \tilde{k})$.  Combining this discussion with Proposition \ref{prop:herm-6j} we have
$$
 \gamma(\tilde{i},\tilde{j},\tilde{k})\gamma(\tilde{k}^{\ast},\tilde{j}^{\ast},\tilde{i}^{\ast})\beta(\tilde{i})\beta(\tilde{j})\beta(\tilde{k})=1;
$$
\begin{align*}
 \wb{(N^{\tilde{j_1} \tilde{j_2} \tilde{j_3} }_{\tilde{j_4}\tilde{j_5}\tilde{j_6}})}
& =
N^{\tilde{j_2}^{\ast} \tilde{j_1}^{\ast} \tilde{j_3}^{\ast}}_{\tilde{j_5} \tilde{j_4}\tilde{j_6}}
\gamma(\tilde{j_1},\tilde{j_2},\tilde{j_3}^{\ast})
\gamma(\tilde{j_1}^{\ast},\tilde{j_5},\tilde{j_6}^{\ast})
\\ & \qquad \times \gamma(\tilde{j_2}^{\ast},\tilde{j_6},\tilde{j_4}^{\ast})
\gamma(\tilde{j_3},\tilde{j_4},\tilde{j_5}^{\ast})\prod_i\beta(\tilde{j_i}).
\end{align*}

The discussion of this subsection can be summarized in the following theorem.
\begin{theorem}
The category of weight modules $\UqMed \dmod$ is a relative $\Gr$-spherical category giving rise to input data
$$
\md(\tilde{i}), \bb(\tilde{i}),
N^{\tilde{i} \tilde{j} \tilde{k}}_{\tilde{l} \tilde{m} \tilde{n}},
\beta(\tilde{i}),
\gamma(\tilde{i}, \tilde{j}, \tilde{k}), \delta_{\tilde i \tilde j  \tilde k}$$
for $\tilde{i}, \tilde{j},\tilde{k}\in \A=\cup_{\bar\alpha\in\Gr\setminus \X} I_{\bar\alpha}$  satisfying the axioms of Section \ref{subsec:basic}.
\end{theorem}

\appendix
\section{Some graphical calculus for the unrolled quantum group} \label{app:hermitian}
The aim of this section is to give representation-theoretic proofs for Proposition \ref{prop:herm-6j}.
Here we recall the evaluation of a network and extend the definition
to uni-trivalent graphs.   Let $\Gamma$ be an admissible uni-trivalent
network in a disc with  its $n$ univalent vertices on the boundary
circle. One can associate to $\Gamma$ a $\mathcal D^\dagger$-ribbon
graph which can be evaluated with the Turaev functor to some invariant
tensor
$\in\Hom(1,V_{i_1}\otimes V_{i_2}\otimes \cdots\otimes
V_{i_n})$.
Note that we read our graphs from bottom to top.

\subsection{The module ${\mathsf v} $}
Let ${\mathsf v}$ denote the two-dimensional irreducible highest weight module of $\Ubar$ with basis $\{v_0,v_1\}$.
The algebra acts as follows:
\[
Ev_1=v_0, \quad, F v_0 = v_1, \quad H v_0 = v_0, \quad H v_1 = -v_1.
\]

Using definitions from \cite[Section 3.3]{GP1}, one could deduce that the following formulas determine maps of $\Ubar$-modules:
\begin{align}
\label{+def1}
\hackcenter{ \begin{tikzpicture} [scale=.70, decoration={markings,
                        mark=at position 0.45 with {\arrow{>}};    }]
\draw[very thick, ->]   (0,-.25) --(0,-1.25) ;
\draw[very thick, postaction={decorate} ] (.5,1.25) .. controls ++(0,-.5) and ++(.25,.35) .. (.0,0);
\draw[very thick, postaction={decorate}]  (-.5,1.25).. controls ++(0,-.5) and ++(-.25,.35)   .. (0,0) ;
\node at (-.5,1.5) { $\scs i+1$};
\node at (0,-1.5) {$\scs i$};
\node at (.5,1.5) {$\scs {\mathsf v}$};
\node[draw, fill=white!20 ,rounded corners ] at (0,0) {$ +$};
\end{tikzpicture}}
\!\!&\colon
\begin{array}{l}
V_i \rightarrow V_{i+1} \otimes \mathsf{v}, \\
v_d \mapsto -q^d \{i-d\} (v_d \otimes v_1) -q^{-1} \{1\} v_{d+1} \otimes v_0 ,
\end{array}
%
\\
\label{+def2}
\hackcenter{ \begin{tikzpicture} [scale=.70, decoration={markings,
                        mark=at position 0.45 with {\arrow{>}};    }]
\draw[very thick, ->]   (0,-.25) --(0,-1.25) ;
\draw[very thick, postaction={decorate} ] (.5,1.25) .. controls ++(0,-.5) and ++(.25,.35) .. (.0,0);
\draw[very thick, postaction={decorate}]  (-.5,1.25).. controls ++(0,-.5) and ++(-.25,.35)   .. (0,0) ;
\node at (.5,1.5) { $\scs i+1$};
\node at (0,-1.5) {$\scs i$};
\node at (-.5,1.5) {$\scs {\mathsf v}$};
\node[draw, fill=white!20 ,rounded corners ] at (0,0) {$ +$};
\end{tikzpicture}}
\!\!& \colon
\begin{array}{l}
V_i \rightarrow \mathsf{v} \otimes V_{i+1}, \\
  v_d \mapsto \frac{1}{q} \{i-d\} (v_1 \otimes v_d) \!-\! q^{i-d-1} \{1\} v_{0} \otimes v_{d+1},
\end{array}
%
\\
  \hackcenter{ \begin{tikzpicture} [scale=.70, decoration={markings,
                        mark=at position 0.6 with {\arrow{>}};    }]
\draw[very thick, postaction={decorate}] (0,1.25) -- (0,.25);
\draw[very thick, -> ] (0,0) .. controls ++(0,.35) and ++(0,.5) .. (.5,-1.25);
\draw[very thick, ->] (0,0) .. controls ++(0,.35) and ++(0,.5) ..  (-.5,-1.25);
\node at (-.5,-1.5) { $\scs i+1$};
\node at (0,1.5) {$\scs i $};
\node at (.5,-1.5) {$\scs \mathsf{v}$};
\node[draw, fill=white!20 ,rounded corners ] at (0,0) {$ -$};
\end{tikzpicture}}
\!\!  &\colon
\begin{array}{l}
V_{i+1} \otimes \mathsf{v} \rightarrow V_{i}, \\
  v_d \otimes v_0 \mapsto [2r'-d-1] q^{-2r'+d+1-i} v_{d-1}, \\
   v_d \otimes v_1 \mapsto -q v_d,
\end{array}
\\
  \hackcenter{ \begin{tikzpicture} [scale=.70, decoration={markings,
                        mark=at position 0.6 with {\arrow{>}};    }]
\draw[very thick, postaction={decorate}] (0,1.25) -- (0,.25);
\draw[very thick, -> ] (0,0) .. controls ++(0,.35) and ++(0,.5) .. (.5,-1.25);
\draw[very thick, ->] (0,0) .. controls ++(0,.35) and ++(0,.5) ..  (-.5,-1.25);
\node at (.5,-1.5) { $\scs i+1$};
\node at (0,1.5) {$\scs i $};
\node at (-.5,-1.5) {$\scs \mathsf{v}$};
\node[draw, fill=white!20 ,rounded corners ] at (0,0) {$ -$};
\end{tikzpicture}}
\!\!&\colon
\begin{array}{l}
\mathsf{v} \otimes V_{i+1} \rightarrow V_i, \\
v_0 \otimes v_d \mapsto  -q[d] v_{d-1}, \\
 v_1 \otimes v_d \mapsto q^{-d} v_d .
\end{array}
\end{align}

\begin{proposition} \cite[Equation 13]{GP1}
For $i \in \R\setminus\Z $, there is an isomorphism ${\mathsf v} \otimes V_{i} \cong V_{i+1} \oplus V_{i-1}$.
\end{proposition}

\begin{proposition} \cite[Section 3.2]{GP1}
There is an isomorphism $w_{\mathsf{v}} \colon \mathsf{v} \rightarrow \mathsf{v}^{\ast}$, where
\[
v_0 \mapsto -q v_1^{\ast} , \quad v_1 \mapsto v_0^{\ast} .
\]
\end{proposition}

\subsection{Isomorphisms $w_{i}$} \label{subsecbeta}
There are isomorphisms
$w_{i} \colon V_{i} \rightarrow V_{-i}^{\ast}$
determined by mapping a highest weight vector
$v_0 \in V_{i}$ by $w_{i}(v_0)=-q^{-1} y_{r-1}^{\ast}$,
where $y_{r-1}$ is a lowest weight vector of $V_{-i}$.
We depict $w_i$ in the graphical calculus as follows:
\begin{equation} \label{eq:w}
w_{i} \;\;  = \;\;
\hackcenter{ \begin{tikzpicture} [scale=.85]
\draw[very thick,  directed=.75] (0,0) to   (0,-.75);
\draw[very thick,  directed=.75] (0,0) to   (0,.75);
\node at (0,0) {$\bullet$};
\node at (-.3,-.65) {$\scs i $};
\node at (-.3,.65) {$\scs -i$};
\end{tikzpicture}} .
\end{equation}

 We want to compute $w_{i}^{\dagger}$.  In order to do this, consider the Hermitian pairing $( , )_{V_{i}}$ in \cite[Proposition 4.5]{GLPMS}  and the Hermitian pairing $(,)_{V_{-i}^{\ast}} $ in \cite[Proposition 4.16]{GLPMS}.  Using those results, we compute
 \begin{equation} \label{eq:A0}
 (y_b^{\ast}, w v_a)_{V_{-i}^{\ast}} = (w^{\dagger} y_b^{\ast}, v_a)_{V_{i}} .
 \end{equation}
 So we now need to recall the form on $V_{-i}^{\ast}$.  By \cite[Proposition 4.16]{GLPMS},
 \begin{equation} \label{eq:A1}
 (y_b^{\ast},y_b^{\ast})= \sum_k y_b^{\ast}(e_k) y_b^{\ast}(e_k')
 \end{equation}
 where $\{e_k \}$ and $\{e_k'\} $ are dual bases of $V_{-i} $.
 If $e_k=y_k$, we calculate $e_k' $ as follows.  Note that
 \begin{align} \label{defofck}
 c_k &:= (y_k, y_k)=(F^k y_{0},y_{k})=(y_0, E^k y_k) \\ \nonumber
 &=
 \prod_{s=1}^k \frac{\{ s \} \{s+i \}}{\{1\}^2}
 .
 \end{align}
 Then set $e_{k}' := \frac{1}{\overline{c}_{k}} e_{k}$.
 Now a specialization of \eqref{eq:A0} is:
  \begin{equation} \label{betadef1}
 (y_{r-1}^{\ast}, w v_0)_{V_{-i}^{\ast}} = (w^{\dagger} y_{r-1}^{\ast}, v_0)_{V_{i}} .
 \end{equation}

 Now we define $\beta(i)$ through the equation
\begin{equation} \label{defofbeta}
w^{\dagger} = \beta(i) w^{-1} .
\end{equation}
Then $ w^{\dagger} y_{r-1}^{\ast} = -q \beta(i) v_0$.
 Then the righthand side of \eqref{betadef1} is $-q^{-1} \overline{\beta(i)} $.
The lefthand side of \eqref{betadef1} is
\begin{align} \nn
& -q^{-1}(y_{r-1}^{\ast}, y_{r-1}^{\ast})
=
-q^{-1}\sum_k y_{r-1}^{\ast}(e_k) y_{r-1}^{\ast}(e_k') \\
\nn &= \quad
-q^{-1}y_{r-1}^{\ast}(y_{r-1}) y_{r-1}^{\ast}(\frac{1}{\overline{c}_{r-1}}  y_{r-1}) =
-q^{-1}\frac{1}{\overline{c}_{r-1}} .
\end{align}
This implies
\[ \beta(i) =  \frac{1}{c_{r-1}} .
\]

\begin{lemma} \label{factorialid}
There is an equality
$\{r-1\}!={\sqrt{-1}}^{r-1} r$.
\end{lemma}

\begin{proof}
We compute as follows:
\begin{align*}
& \{r-1\}! = \prod_{j=1}^{r-1} (q^j-q^{-j})= \prod_{j=1}^{r-1} (2 \sqrt{-1}) \sin(\frac{\pi j}{r})
\\
& \quad =
(2 \sqrt{-1})^{r-1} \frac{r}{2^{r-1}} =
\sqrt{-1}^{r-1} r .
\end{align*}
\end{proof}

\begin{proposition} \label{prop:beta-formula}
One has
$ \beta(i) =  \frac{ \{1 \}^{2(r-1)} \md({i})}{r \md_0}$.
\end{proposition}

\begin{proof}
Recall from \eqref{defofck} that
\begin{align*}
c_{r-1} &= \prod_{s=1}^{r-1} \frac{\{ s \} \{s+i \}}{\{1\}^2}
=
\frac{\{r-1\}!}{\{1 \}^{2(r-1)}} \prod_{s=1}^{r-1} \{s+i \}
\\ &
=
\frac{(\{r-1\}!)^2}{\{1 \}^{2(r-1)} \{r-1\}! } \prod_{s=1}^{r-1}  \{s+i \}
\\ &=
\frac{(\{r-1\}!)^2}{\{1 \}^{2(r-1)}} {r-1+i \brack i}
.
\end{align*}
By \cite[Equation 6]{CGP1}, this is equal to
\[
\frac{(\{r-1\}!)^2}{\{1 \}^{2(r-1)}}
\left(\frac{\md(i)}{(-1)^{r-1}} \right)^{-1} .
 \]
 By Lemma \ref{factorialid}, this is equal to
 \[
 \frac{(-1)^{r-1} r^2}{\{1 \}^{2(r-1)}}
 \left(\frac{\md({i})}{(-1)^{r-1}} \right)^{-1}
 =
 \frac{r^2}{\{1\}^{2(r-1)} \md({i})} .
 \]
Recall that we take $\md_0 = r$.  Thus we get
 \[
 \beta(i)= \frac{(-1)^{r-1} \{1 \}^{2(r-1)} \md({i})}{r \md_0}
 = \frac{ \{1 \}^{2(r-1)} \md({i})}{r \md_0}
 .
 \]
\end{proof}

\subsection{Coefficients $\gamma_{}$} \label{gammasec}
Consider the invariant space
\[
H(i,j,k)=\Hom(1, V_i \otimes V_j \otimes V_k).
\]
This space is $1$-dimensional and is spanned by a specific basis element $h_{ijk}$ that we depict using the graphical notation
\begin{equation}
h_{ijk}\;\; := \;\;
\hackcenter{ \begin{tikzpicture} [scale=.75, decoration={markings,
                        mark=at position 0.5 with {\arrow{>}};    }]
\draw[very thick, postaction={decorate}] (-1,1.25) .. controls ++(0,-.5) and ++(-.5,.25) ..   (0,.0);
\draw[very thick, postaction={decorate}] (1,1.25) .. controls ++(0,-.5) and ++(.5,.25) ..   (0,.0);
\draw[very thick,postaction={decorate}] (0,1.25) -- (0,0);
\node at (0,0) {$\bullet$};
\node at (-1,1.5) {$\scs i $};
\node at (0,1.5) {$\scs j $};
\node at (1,1.5) {$\scs k$};
\end{tikzpicture}}
\end{equation}
defined explicitly in \eqref{eq:h-ijk}.
Then using the pivotal structure, a choice of $h_{ijk}$ determines uniquely basis elements of the 1-dimensional spaces
\begin{align*}
&\Hom(V_{i}^{\ast}, V_j \otimes V_k)
\cong \Hom(V_j ^{\ast} \otimes V_{i}^{\ast},   V_k)
\\
& \quad \cong \Hom(V_k^{\ast} \otimes V_j ^{\ast} \otimes V_{i}^{\ast},  1 ).
\end{align*}
Combining with the isomorphisms \eqref{eq:w}, a choice of $h_{ijk}$ determines unique basis elements in
\begin{align*}
\Hom(V_{-i^{}}, V_j \otimes V_k)
&\cong \Hom(V_{-j^{}} \otimes V_{-i^{}},   V_k) \\
& \cong \Hom(V_{-k} \otimes V_{-j^{}} \otimes V_{-i^{}},  1 ) .
\end{align*}
Then we have $h^{ijk}$ as a distinguished basis element of the 1-dimensional space $
\Hom(V_i \otimes V_j \otimes V_k, 1) \cong
\Hom(1, V_k^{\ast} \otimes V_j^{\ast} \otimes V_i^{\ast} ) \cong
\Hom(1, V_{-k} \otimes V_{-j} \otimes V_{-i} )$.
%
\begin{equation}
h^{ijk}  \; = \;
\hackcenter{ \begin{tikzpicture} [scale=.7, decoration={markings,
                        mark=at position 0.5 with {\arrow{>}};    }]
\draw[very thick, postaction={decorate}] (0,.0).. controls ++(-.5,-.25)and ++(0,.5) .. (-1,-1.25);
\draw[very thick, postaction={decorate}] (0,.0).. controls ++(.5,-.25)and ++(0,.5) ..  (1,-1.25);
\draw[very thick,postaction={decorate}] (0,0)-- (0,-1.25);
\node at (0,0) {$\bullet$};
\node at (-1,-1.5) {$\scs i $};
\node at (0,-1.5) {$\scs j $};
\node at (1,-1.5) {$\scs k$};
\end{tikzpicture}}
\; := \;
\hackcenter{ \begin{tikzpicture} [scale=.75]
\draw[very thick,  directed=.45] (-1,1.0) .. controls ++(0,-.5) and ++(-.5,.25) ..   (0,.0);
\draw[very thick, directed=.45] (1,1.0) .. controls ++(0,-.5) and ++(.5,.25) ..   (0,.0);
\draw[very thick, directed=.55] (0,1.0) -- (0,0);
\draw[very thick, directed=.65] (-1,1.0) -- (-1,1.75);
\draw[very thick, directed=.65] (0,1.0) -- (0,1.75);
\draw[very thick, directed=.65] (1,1.0) -- (1,1.75);
\draw[very thick, ->] (-1,1.75).. controls ++(0,.5)and ++(0,.5) .. (-2,1.75) -- (-2,0);
\draw[very thick, ->] (0,1.75).. controls ++(0,1)and ++(0,1) .. (-3,1.75)-- (-3,0);
\draw[very thick, ->] (1,1.75).. controls ++(0,1.5)and ++(0,1.5) .. (-4,1.75) -- (-4,0);
\node at (0,0) {$\bullet$};
\node at (0,1.0) {$\bullet$};
\node at (1,1.0) {$\bullet$};
\node at (-1,1.0) {$\bullet$};
\node at (-1.25,.65) {$\scs -k $};
\node at (-.42,.65) {$\scs -j $};
\node at (1.25,.65) {$\scs -i$};
\node at (-2.25,1.25) {$\scs k $};
\node at (-3.3,1.25) {$\scs j$};
\node at (-4.25,1.25) {$\scs i$};
\end{tikzpicture}}
\end{equation}
Since $(h_{ijk})^{\dagger} \in \Hom(V_i \otimes V_j \otimes V_k, 1)$  we define nonzero scalars $\gamma(i,j,k)$ by
\begin{equation}
  h_{ijk}^{\dagger} = \gamma(i,j,k)h^{ijk}.
\end{equation}

\begin{proposition}
If $i+j-k=-2r'$, then there are maps of representations
\begin{equation} \label{defofYbasecase}
\hackcenter{ \begin{tikzpicture} [scale=.7, decoration={markings,
                        mark=at position 0.45 with {\arrow{>}};    }]
\draw[very thick, ->]   (0,-.25) --(0,-1.25) ;
\draw[very thick, postaction={decorate} ] (.5,1.25) .. controls ++(0,-.5) and ++(.25,.35) .. (.0,0);
\draw[very thick, postaction={decorate}]  (-.5,1.25).. controls ++(0,-.5) and ++(-.25,.35)   .. (0,0) ;
\node at (-.5,1.5) { $\scs i$};
\node at (0,-1.5) {$\scs k $};
\node at (.5,1.5) {$\scs j$};
\node[draw, fill=white!20 ,rounded corners ] at (0,0) {$ -2r'$};
\end{tikzpicture}}
  \colon V_k \rightarrow V_i \otimes V_j
 \qquad \quad
  \hackcenter{ \begin{tikzpicture} [scale=.7, decoration={markings,
                        mark=at position 0.6 with {\arrow{>}};    }]
\draw[very thick, postaction={decorate}] (0,1.25) -- (0,.25);
\draw[very thick, -> ] (0,0) .. controls ++(0,.35) and ++(0,.5) .. (.5,-1.25);
\draw[very thick, ->] (0,0) .. controls ++(0,.35) and ++(0,.5) ..  (-.5,-1.25);
\node at (-.5,-1.5) { $\scs i$};
\node at (0,1.5) {$\scs k $};
\node at (.5,-1.5) {$\scs j$};
\node[draw, fill=white!20 ,rounded corners ] at (0,0) {$ 2r'$};
\end{tikzpicture}}
 \colon V_i \otimes V_j \rightarrow V_k
\end{equation}
determined by $v_0 \mapsto v_0 \otimes v_0$ and  $v_0 \otimes v_0 \mapsto \frac{1}{\md(k)} v_0$, respectively.
\end{proposition}


\begin{proof}
It is trivial to see that the first map defines a morphism of modules.
The second morphism is defined using the first one and the pivotal structure of the category.  Both spaces of morphisms are $1$-dimensional.  The scalar for the second morphism is determined by
the identity (\cite[Figure 5]{GP1})
\begin{equation*}
\hackcenter{ \begin{tikzpicture} [scale=.7, decoration={markings,
                        mark=at position 0.6 with {\arrow{>}};    }]
\draw[very thick, postaction={decorate}] (0,1.25) -- (0,.25);
\draw[very thick, postaction={decorate}](0,0) .. controls ++(-.65,-.35) and ++(-.65,.35) .. (0,-2);
\draw[very thick,  postaction={decorate}](0,0) .. controls ++(.65,-.35) and ++(.65,.35) .. (0,-2);
\draw[very thick, ->] (0,-2) -- (0,-3.25);
\node at (-.95,-.8) { $\scs i$};
\node at (0,1.5) {$\scs k $};
\node at (0,-3.5) {$\scs k$};
\node at (.95,-.8) {$\scs j$};
\node[draw, fill=white!20 ,rounded corners ] at (0,0) {$ 2r'$};
\node[draw, fill=white!20 ,rounded corners ] at (0,-2) {$ -2r'$};
\end{tikzpicture}}
=
\frac{1}{\md(k)} .
\end{equation*}
\end{proof}

\begin{proposition} \label{Ydagger}
Assume $i+j-k=-2r'$.  Then
there are equalities of maps
\begin{align} \label{Yl}
\left(  \hackcenter{ \begin{tikzpicture} [scale=.7, decoration={markings,
                        mark=at position 0.6 with {\arrow{>}};    }]
\draw[very thick, postaction={decorate}] (0,1.25) -- (0,.25);
\draw[very thick, -> ] (0,0) .. controls ++(0,.35) and ++(0,.5) .. (.5,-1.25);
\draw[very thick, ->] (0,0) .. controls ++(0,.35) and ++(0,.5) ..  (-.5,-1.25);
\node at (-.5,-1.5) { $\scs i$};
\node at (0,1.5) {$\scs k $};
\node at (.5,-1.5) {$\scs j$};
\node[draw, fill=white!20 ,rounded corners ] at (0,0) {$ 2r'$};
\end{tikzpicture}}
\right)^{\dagger}
&\;\;  = \;\; \frac{1}{\md(k)}\;
\hackcenter{ \begin{tikzpicture} [scale=.7, decoration={markings,
                        mark=at position 0.45 with {\arrow{>}};    }]
\draw[very thick, ->]   (0,-.25) --(0,-1.25) ;
\draw[very thick, postaction={decorate} ] (.5,1.25) .. controls ++(0,-.5) and ++(.25,.35) .. (.0,0);
\draw[very thick, postaction={decorate}]  (-.5,1.25).. controls ++(0,-.5) and ++(-.25,.35)   .. (0,0) ;
\node at (-.5,1.5) { $\scs i$};
\node at (0,-1.5) {$\scs k $};
\node at (.5,1.5) {$\scs j$};
\node[draw, fill=white!20 ,rounded corners ] at (0,0) {$ -2r'$};
\end{tikzpicture}}
\; \colon   V_{k} \rightarrow V_{i} \otimes V_{j} ,
\\
\label{Yu}
\left(
\hackcenter{ \begin{tikzpicture} [scale=.7, decoration={markings,
                        mark=at position 0.45 with {\arrow{>}};    }]
\draw[very thick, ->]   (0,-.25) --(0,-1.25) ;
\draw[very thick, postaction={decorate} ] (.5,1.25) .. controls ++(0,-.5) and ++(.25,.35) .. (.0,0);
\draw[very thick, postaction={decorate}]  (-.5,1.25).. controls ++(0,-.5) and ++(-.25,.35)   .. (0,0) ;
\node at (-.5,1.5) { $\scs i$};
\node at (0,-1.5) {$\scs k $};
\node at (.5,1.5) {$\scs j$};
\node[draw, fill=white!20 ,rounded corners ] at (0,0) {$ -2r'$};
\end{tikzpicture}}
\right)^{\dagger}
 &\;\; = \;\;
\hackcenter{ \begin{tikzpicture} [scale=.7, decoration={markings,
                        mark=at position 0.6 with {\arrow{>}};    }]
\draw[very thick, postaction={decorate}] (0,1.25) -- (0,.25);
\draw[very thick, -> ] (0,0) .. controls ++(0,.35) and ++(0,.5) .. (.5,-1.25);
\draw[very thick, ->] (0,0) .. controls ++(0,.35) and ++(0,.5) ..  (-.5,-1.25);
\node at (-.5,-1.5) { $\scs i$};
\node at (0,1.5) {$\scs k $};
\node at (.5,-1.5) {$\scs j$};
\node[draw, fill=white!20 ,rounded corners ] at (0,0) {$ 2r'$};
\end{tikzpicture}}
\;\colon   V_{i} \otimes V_{j}  \rightarrow  V_k.
\end{align}
\end{proposition}

\begin{proof}
Denote the maps in \eqref{defofYbasecase} by $Y_k^{ij}$ and $Y_{ij}^k$ respectively.  In order to compute $(Y^k_{ij})^{\dagger}$, we use the pairing from \cite[Theorem 4.14]{GLPMS}.
Let us assume that $(Y^k_{ij})^{\dagger}=M Y^{ij}_k $, where $M$ is a complex scalar.     Then we have the chain of equalities
\begin{align*}
&(v_0, Y_{ij}^k(v_0 \otimes v_0))_{V_k}
= ((Y_{ij}^k)^{\dagger} v_0, v_0 \otimes v_0)_{V_i \otimes V_j} \\
&=(M Y_k^{ij} v_0, v_0 \otimes v_0)_{V_i \otimes V_j} \\
&=\overline{M}(v_0 \otimes v_0, v_0 \otimes v_0)_{V_i \otimes V_j} \\
&=\overline{M}(v_0 \otimes v_0, \tau X(v_0 \otimes v_0))_p \\
&= \overline{M}(v_0 \otimes v_0, \tau \sqrt{\theta_{V_j \otimes V_i}}^{-1} \!c_{V_i,V_j} \sqrt{\theta_{V_i}} \otimes \sqrt{\theta_{V_j}} (v_0 \otimes v_0))_p \\
& =\overline{M}
\end{align*}
where $(,)_{V_i \otimes V_j} $ denotes the form defined in \cite[Theorem 4.14]{GLPMS} and
$ (,)_p$ denotes the product of the pairings on each tensor factor (see \cite[Equation 31]{GLPMS}).
The fourth equality follows from \cite[Theorem 4.14]{GLPMS}.
The fifth equality is a consequence of \cite[Equation 29]{GLPMS}.
The sixth equality is a consequence of \cite[Equation 28]{GLPMS}.

On the other hand,
\begin{equation*}
(v_0, Y_{ij}^k(v_0 \otimes v_0)) = \frac{1}{\md(k)} .
\end{equation*}
Thus
\begin{equation*}
M= \frac{1}{\md(k)}
\end{equation*}
establishing \eqref{Yl}.
Equation \eqref{Yu} follows from \eqref{Yl}, \cite[Figure 5]{GP1}, and \cite[Lemma 4.19]{GLPMS}.
\end{proof}

There are maps (see \cite[Lemma 12]{GP1})
\begin{equation} \label{defofHlower}
H_{jk} \colon
\begin{array}{l}
V_{j} \otimes V_{k} \rightarrow V_{j+1} \otimes V_{k+1}, \\
 q^{k+c-d-1} \{j-c\} v_c \otimes v_{d+1} + \frac{1}{q} \{k -d \} v_{c+1} \otimes v_d,
\end{array}
\end{equation}
defined as a composite of maps introduced earlier, and graphically denoted by
\begin{equation}
\hackcenter{ \begin{tikzpicture} [scale=.7, decoration={markings,
                        mark=at position 0.45 with {\arrow{>}};    }]
\draw[very thick, ->]   (0,-.25) --(0,-1.25) ;
\draw[very thick, ->]   (1.5,-.25) --(1.5,-1.25) ;
\draw[very thick] (.2,.25) .. controls ++(0,.75) and ++(0,.75) .. (1.3,.25);
\draw[very thick, postaction={decorate}]  (-.5,1.25).. controls ++(0,-.5) and ++(-.25,.25)   .. (0,0) ;
\draw[very thick, postaction={decorate}]  (2,1.25).. controls ++(0,-.5) and ++(.25,.25)   .. (1.5,0) ;
\node at (-.5,1.5) { $\scs j+1$};
\node at (0,-1.5) {$\scs j $};
\node at (1.5,-1.5) {$\scs k $};
\node at (2,1.5) {$\scs k+1$};
\node[draw, fill=white!20 ,rounded corners ] at (0,0) {$ +$};
\node[draw, fill=white!20 ,rounded corners ] at (1.5,0) {$ +$};
\node at (.75,.8) {$\bullet$};
\end{tikzpicture}} .
\end{equation}
\begin{remark}
  A priori, there is an ambiguity in the diagram above since it is not clear which object the morphism $w$ (denoted by a dot) is acting on.  However, by the definition of
the basic data in
  a relative $\Gr$-spherical category, the dot may move anywhere along a cup or cap morphism (see \cite[Lemma 7]{GP1}).
\end{remark}

There are maps defined in \cite[Section 3.6]{GP1} (called $\overline{X}$ there)
\begin{equation} \label{defofHupper}
{H}^{jk} =
\hackcenter{ \begin{tikzpicture} [scale=.7, decoration={markings,
                        mark=at position 0.6 with {\arrow{>}};    }]
\draw[very thick, postaction={decorate}] (0,1.25) -- (0,.25);
\draw[very thick, postaction={decorate}] (1.5,1.25) -- (1.5,.25);
\draw[very thick] (.2,-.25) .. controls ++(0,-.75) and ++(0,-.75) .. (1.3,-.25);
\draw[very thick, ->] (0,0) .. controls ++(0,.25) and ++(0,.5) ..  (-.5,-1.25);
\draw[very thick, ->] (1.5,0) .. controls ++(0,.25) and ++(0,.5) ..  (2,-1.25);
\node at (-.5,-1.5) { $\scs j+1$};
\node at (0,1.5) {$\scs j $};
\node at (1.5,1.5) {$\scs k $};
\node at (2,-1.5) {$\scs k+1$};
\node at (.75,-.8) {$\bullet$};
\node[draw, fill=white!20 ,rounded corners ] at (0,0) {$ -$};
\node[draw, fill=white!20 ,rounded corners ] at (1.5,0) {$ -$};
\end{tikzpicture}}
\colon V_{j+1} \otimes V_{k+1} \rightarrow V_{j} \otimes V_{k} .
\end{equation}

\begin{proposition} \label{Hdaggerbasecase}
We have an equality of maps $$ H_{jk}^{\dagger}= \{j\} \{k\} H^{jk} .$$
\end{proposition}

\begin{proof}
This follows in a similar fashion to the proof of Proposition \ref{Ydagger}.
\end{proof}

Letting $i+j-k=2F$, we now define the map of representations
\begin{align} \label{defofYgen}
&\hackcenter{ \begin{tikzpicture} [scale=.7, decoration={markings,
                        mark=at position 0.45 with {\arrow{>}};    }]
\draw[very thick, ->]   (0,-.25) --(0,-1.25) ;
\draw[very thick, postaction={decorate} ] (.5,1.25) .. controls ++(0,-.5) and ++(.25,.35) .. (.0,0);
\draw[very thick, postaction={decorate}]  (-.5,1.25).. controls ++(0,-.5) and ++(-.25,.35)   .. (0,0) ;
\node at (-.5,1.5) { $\scs i$};
\node at (0,-1.5) {$\scs k$};
\node at (.5,1.5) {$\scs j$};
\node[draw, fill=white!20 ,rounded corners ] at (0,0) {$ 2F$};
\end{tikzpicture}} \colon V_{k} \rightarrow V_{i} \otimes V_{j} \\
&
:=
\hackcenter{ \begin{tikzpicture} [scale=.7, decoration={markings,
                        mark=at position 0.45 with {\arrow{>}};    }]
\draw[very thick, ->]   (0,-.25) --(0,-1.25) ;
\draw[very thick, ->]   (1.5,-.25) --(1.5,-1.25) ;
\draw[very thick] (.2,.25) .. controls ++(0,.75) and ++(0,.75) .. (1.3,.25);
\draw[very thick, postaction={decorate}]  (-.5,1.25).. controls ++(0,-.5) and ++(-.25,.25)   .. (0,0) ;
\draw[very thick, postaction={decorate}]  (2,1.25).. controls ++(0,-.5) and ++(.25,.25)   .. (1.5,0) ;
\node at (-.3,1.5) { $\scs i$};
\node at (0,-1.5) {$\scs i-1 $};
\node at (1.5,-1.5) {$\scs j-1 $};
\node at (2,1.5) {$\scs j$};
\node[draw, fill=white!20 ,rounded corners ] at (0,0) {$ +$};
\node[draw, fill=white!20 ,rounded corners ] at (1.5,0) {$ +$};
\node at (.75,.8) {$\bullet$};
\end{tikzpicture}}
\!\!\circ \cdots \circ \!\!\!\!\!\!\!\!\!
\hackcenter{ \begin{tikzpicture} [scale=.7, decoration={markings,
                        mark=at position 0.45 with {\arrow{>}};    }]
\draw[very thick, ->]   (0,-.25) --(0,-1.25) ;
\draw[very thick, ->]   (1.5,-.25) --(1.5,-1.25) ;
\draw[very thick] (.2,.25) .. controls ++(0,.75) and ++(0,.75) .. (1.3,.25);
\draw[very thick, postaction={decorate}]  (-.5,1.25).. controls ++(0,-.5) and ++(-.25,.25)   .. (0,0) ;
\draw[very thick, postaction={decorate}]  (2,1.25).. controls ++(0,-.5) and ++(.25,.25)   .. (1.5,0) ;
\node at (-.4,1.5) { $\scs i-r'-F+1$};
\node at (-.2,-1.5) {$\scs i-r'-F $};
\node at (1.6,-1.5) {$\scs j-r'-F $};
\node at (1.8,1.5) {$\scs j-r'-F+1$};
\node[draw, fill=white!20 ,rounded corners ] at (0,0) {$ +$};
\node[draw, fill=white!20 ,rounded corners ] at (1.5,0) {$ +$};
\node at (.75,.8) {$\bullet$};
\end{tikzpicture}}
\!\!\!\circ \!\!\!
 \hackcenter{ \begin{tikzpicture} [scale=.7, decoration={markings,
                        mark=at position 0.45 with {\arrow{>}};    }]
\draw[very thick, ->]   (0,-.25) --(0,-1.25) ;
\draw[very thick, postaction={decorate} ] (.5,1.25) .. controls ++(0,-.5) and ++(.25,.35) .. (.0,0);
\draw[very thick, postaction={decorate}]  (-.5,1.25).. controls ++(0,-.5) and ++(-.25,.35)   .. (0,0) ;
\node at (-.8,1.5) { $\scs i-r'-F$};
\node at (0,-1.5) {$\scs k $};
\node at (.8,1.5) {$\scs j-r'-F$};
\node[draw, fill=white!20 ,rounded corners ] at (0,0) {$ -2r'$};
\end{tikzpicture}} .
\end{align}
For $i+j+k=2F$, this allows us to fix an explicit choice of basis element $h_{ijk} \in \Hom(1, V_i \otimes V_j \otimes V_k)$ by
\begin{equation} \label{eq:h-ijk}
\hackcenter{ \begin{tikzpicture} [scale=.7, decoration={markings,
                        mark=at position 0.5 with {\arrow{>}};    }]
\draw[very thick, postaction={decorate}] (-1,1.25) .. controls ++(0,-.5) and ++(-.5,.25) ..   (0,.0);
\draw[very thick, postaction={decorate}] (1,1.25) .. controls ++(0,-.5) and ++(.5,.25) ..   (0,.0);
\draw[very thick,postaction={decorate}] (0,1.25) -- (0,0);
\node at (0,0) {$\bullet$};
\node at (-1,1.5) {$\scs i $};
\node at (0,1.5) {$\scs j $};
\node at (1,1.5) {$\scs k$};
\end{tikzpicture}}
\;\; =\;\;
\hackcenter{ \begin{tikzpicture} [scale=.7, decoration={markings,
                        mark=at position 0.45 with {\arrow{>}};    }]
\draw[very thick, postaction={decorate} ] (.5,1.25) .. controls ++(0,-.5) and ++(.25,.35) .. (.0,0);
\draw[very thick, postaction={decorate}]  (-.5,1.25).. controls ++(0,-.5) and ++(-.25,.35)   .. (0,0) ;
\node at (-.5,1.5) { $\scs i$};
\node at (.12,-.85) {$\scs -k$};
\node at (.5,1.5) {$\scs j$};
\node at (1.3,1.5) {$\scs k$};
\draw[very thick] (.2,-.25) .. controls ++(0,-.75) and ++(0,-.75) .. (1.3,-.25);
\node at (.75,-.8) {$\bullet$};
\node[draw, fill=white!20 ,rounded corners ] at (0,0) {$ 2F$};
\draw[very thick,<-]   (1.3,-.25) --(1.3,1.25) ;
\end{tikzpicture}} .
\end{equation}

\begin{proposition}
Let $k-i-j=2F$.  Then there is an equality of maps $V_{i+1} \otimes V_{j+1} \rightarrow V_k$
\begin{equation} \label{absorbinganH}
\hackcenter{ \begin{tikzpicture} [scale=.65, decoration={markings,
                        mark=at position 0.6 with {\arrow{>}};    }]
\draw[very thick, postaction={decorate}] (.75,1.25) .. controls ++(0,.25) and ++(0,.5) ..   (0,.25);
\draw[very thick, postaction={decorate}] (.75,1.25).. controls ++(0,.25) and ++(0,.5) ..  (1.5,.25);
\draw[very thick] (.2,-.25) .. controls ++(0,-.75) and ++(0,-.75) .. (1.3,-.25);
\draw[very thick, ->] (0,0) .. controls ++(0,.25) and ++(0,.5) ..  (-.5,-1.25);
\draw[very thick, ->] (1.5,0) .. controls ++(0,.25) and ++(0,.5) ..  (2,-1.25);
\draw[very thick, postaction={decorate}] (.75,2.75) -- (.75,1.75);
\node at (-.5,-1.5) { $\scs i+1$};
\node at (-.3,1.05) {$\scs i $};
\node at (1.8,1.05) {$\scs j $};
\node at (2,-1.5) {$\scs j+1$};
\node at (.75,3) {$\scs k$};
\node at (.75,-.8) {$\bullet$};
\node[draw, fill=white!20 ,rounded corners ] at (0,0) {$ -$};
\node[draw, fill=white!20 ,rounded corners ] at (1.5,0) {$ -$};
\node[draw, fill=white!20 ,rounded corners ] at (.75,1.5) {$ 2F$};
\end{tikzpicture}}
\!= \;
\frac{\{r'+F\}\{k-F+r'+1 \}}{\{1\}}
 \hackcenter{ \begin{tikzpicture} [scale=.7, decoration={markings,
                        mark=at position 0.6 with {\arrow{>}};    }]
\draw[very thick, postaction={decorate}] (0,1.25) -- (0,.25);
\draw[very thick, -> ] (0,0) .. controls ++(0,.35) and ++(0,.5) .. (.5,-1.25);
\draw[very thick, ->] (0,0) .. controls ++(0,.35) and ++(0,.5) ..  (-.5,-1.25);
\node at (-.5,-1.5) { $\scs i+1$};
\node at (0,1.5) {$\scs k $};
\node at (.5,-1.5) {$\scs j+1$};
\node[draw, fill=white!20 ,rounded corners ] at (0,0) {$ 2F-2$};
\end{tikzpicture}}
\end{equation}
\end{proposition}

\begin{proof}
Since these morphism spaces are $1$-dimensional,   we must have
\begin{equation} \label{eq:--}
\hackcenter{ \begin{tikzpicture} [scale=.65, decoration={markings,
                        mark=at position 0.6 with {\arrow{>}};    }]
\draw[very thick, postaction={decorate}] (.75,1.25) .. controls ++(0,.25) and ++(0,.5) ..   (0,.25);
\draw[very thick, postaction={decorate}] (.75,1.25).. controls ++(0,.25) and ++(0,.5) ..  (1.5,.25);
\draw[very thick] (.2,-.25) .. controls ++(0,-.75) and ++(0,-.75) .. (1.3,-.25);
\draw[very thick, ->] (0,0) .. controls ++(0,.25) and ++(0,.5) ..  (-.5,-1.25);
\draw[very thick, ->] (1.5,0) .. controls ++(0,.25) and ++(0,.5) ..  (2,-1.25);
\draw[very thick, postaction={decorate}] (.75,2.75) -- (.75,1.75);
\node at (-.5,-1.5) { $\scs i+1$};
\node at (-.3,1.05) {$\scs i $};
\node at (1.8,1.05) {$\scs j $};
\node at (2,-1.5) {$\scs j+1$};
\node at (.75,3) {$\scs k$};
\node at (.75,-.8) {$\bullet$};
\node[draw, fill=white!20 ,rounded corners ] at (0,0) {$ -$};
\node[draw, fill=white!20 ,rounded corners ] at (1.5,0) {$ -$};
\node[draw, fill=white!20 ,rounded corners ] at (.75,1.5) {$ 2F$};
\end{tikzpicture}}
\;\; = \;\;
p_{2F}\;
 \hackcenter{ \begin{tikzpicture} [scale=.7, decoration={markings,
                        mark=at position 0.6 with {\arrow{>}};    }]
\draw[very thick, postaction={decorate}] (0,1.25) -- (0,.25);
\draw[very thick, -> ] (0,0) .. controls ++(0,.35) and ++(0,.5) .. (.5,-1.25);
\draw[very thick, ->] (0,0) .. controls ++(0,.35) and ++(0,.5) ..  (-.5,-1.25);
\node at (-.5,-1.5) { $\scs i+1$};
\node at (0,1.5) {$\scs k  $};
\node at (.5,-1.5) {$\scs j+1$};
\node[draw, fill=white!20 ,rounded corners ] at (0,0) {$ 2F-2$};
\end{tikzpicture}}
\end{equation}
for some constant $p_{2F}$.
Precomposing yields the equality
\begin{equation} \label{precomp}
 \hackcenter{ \begin{tikzpicture} [scale=.65, decoration={markings,
                        mark=at position 0.6 with {\arrow{>}};    }]
\draw[very thick, postaction={decorate}] (.75,1.25) .. controls ++(0,.25) and ++(0,.5) ..   (0,.25);
\draw[very thick, postaction={decorate}] (.75,1.25).. controls ++(0,.25) and ++(0,.5) ..  (1.5,.25);
\draw[very thick] (.2,-.25) .. controls ++(0,-.75) and ++(0,-.75) .. (1.3,-.25);
\draw[very thick,postaction={decorate}]  (0,0) .. controls ++(-.5,-.65) and ++(-.5,.5) ..  (.6,-1.5);
\draw[very thick,postaction={decorate}]  (1.5,0) .. controls ++(.5,-.65) and ++(.5,.5) ..  (.8,-1.5);
\draw[very thick, postaction={decorate}] (.75,2.75) -- (.75,1.75);
\draw[very thick, ->]  (.75,-2) -- (.75,-3);
\node at (-.45,-1.25) { $\scs i+1$};
\node at (-.3,1.05) {$\scs i $};
\node at (1.8,1.05) {$\scs j $};
\node at (1.75,-1.25) {$\scs j+1$};
\node at (.75,3) {$\scs k$};
\node at (.75,-3.25) {$\scs k$};
\node at (.75,-.8) {$\bullet$};
\node[draw, fill=white!20 ,rounded corners ] at (0,0) {$ -$};
\node[draw, fill=white!20 ,rounded corners ] at (1.5,0) {$ -$};
\node[draw, fill=white!20 ,rounded corners ] at (.75,1.5) {$ 2F$};
\node[draw, fill=white!20 ,rounded corners ] at (.75,-1.75) {$ 2-2F$};
\end{tikzpicture} }
\;\; \refequal{\eqref{eq:--}} \;\;
p_{2F} \;
\hackcenter{ \begin{tikzpicture} [scale=.7, decoration={markings,
                        mark=at position 0.6 with {\arrow{>}};    }]
\draw[very thick, postaction={decorate}] (0,1.25) -- (0,.25);
\draw[very thick, postaction={decorate}](0,0) .. controls ++(-.65,-.35) and ++(-.65,.35) .. (0,-2);
\draw[very thick,  postaction={decorate}](0,0) .. controls ++(.65,-.35) and ++(.65,.35) .. (0,-2);
\draw[very thick, ->] (0,-2) -- (0,-3.25);
\node at (-.95,-.8) { $\scs i+1$};
\node at (0,1.5) {$\scs k $};
\node at (0,-3.5) {$\scs k$};
\node at (.95,-.8) {$\scs j+1$};
\node[draw, fill=white!20 ,rounded corners ] at (0,0) {$ 2F-2$};
\node[draw, fill=white!20 ,rounded corners ] at (0,-2) {$ 2-2F$};
\end{tikzpicture}} .
\end{equation}
By \cite[Proposition 24]{GP1}, \eqref{precomp} becomes
\begin{equation} \label{precomp2}
\frac{\{r'+F\}\{k-F+r'+1 \}}{\{1\}}
\hackcenter{ \begin{tikzpicture} [scale=.7, decoration={markings,
                        mark=at position 0.6 with {\arrow{>}};    }]
\draw[very thick, postaction={decorate}] (0,1.25) -- (0,.25);
\draw[very thick, postaction={decorate}](0,0) .. controls ++(-.65,-.35) and ++(-.65,.35) .. (0,-2);
\draw[very thick,  postaction={decorate}](0,0) .. controls ++(.65,-.35) and ++(.65,.35) .. (0,-2);
\draw[very thick, ->] (0,-2) -- (0,-3.25);
\node at (-.8,-.8) { $\scs i$};
\node at (0,1.5) {$\scs k $};
\node at (0,-3.5) {$\scs k$};
\node at (.8,-.8) {$\scs j$};
\node[draw, fill=white!20 ,rounded corners ] at (0,0) {$ 2F$};
\node[draw, fill=white!20 ,rounded corners ] at (0,-2) {$ -2F$};
\end{tikzpicture}}
  =
p_{2F}
\hackcenter{ \begin{tikzpicture} [scale=.7, decoration={markings,
                        mark=at position 0.6 with {\arrow{>}};    }]
\draw[very thick, postaction={decorate}] (0,1.25) -- (0,.25);
\draw[very thick, postaction={decorate}](0,0) .. controls ++(-.65,-.35) and ++(-.65,.35) .. (0,-2);
\draw[very thick,  postaction={decorate}](0,0) .. controls ++(.65,-.35) and ++(.65,.35) .. (0,-2);
\draw[very thick, ->] (0,-2) -- (0,-3.25);
\node at (-.95,-.8) { $\scs i+1$};
\node at (0,1.5) {$\scs k $};
\node at (0,-3.5) {$\scs k$};
\node at (.95,-.8) {$\scs j+1$};
\node[draw, fill=white!20 ,rounded corners ] at (0,0) {$ 2F-2$};
\node[draw, fill=white!20 ,rounded corners ] at (0,-2) {$ 2-2F$};
\end{tikzpicture}}
\end{equation}
By \cite[Figure 5]{GP1}, both diagrams in \eqref{precomp2} evaluate to $\frac{1}{\md(k)}$.  This implies that
\[
p_{2F}=\frac{\{r'+F\}\{k-F+r'+1 \}}{\{1\}} .
\]
\end{proof}

\begin{proposition} \label{daggerYgen}
There is an equality of maps
\begin{align}
\left( \hackcenter{ \begin{tikzpicture} [scale=.65, decoration={markings,
                        mark=at position 0.45 with {\arrow{>}};    }]
\draw[very thick, ->]   (0,-.25) --(0,-1.25) ;
\draw[very thick, postaction={decorate} ] (.5,1.25) .. controls ++(0,-.5) and ++(.25,.35) .. (.0,0);
\draw[very thick, postaction={decorate}]  (-.5,1.25).. controls ++(0,-.5) and ++(-.25,.35)   .. (0,0) ;
\node at (-.5,1.5) { $\scs i$};
\node at (0,-1.5) {$\scs k $};
\node at (.5,1.5) {$\scs j$};
\node[draw, fill=white!20 ,rounded corners ] at (0,0) {$ 2F$};
\end{tikzpicture}}
\right)^{\dagger} &=  \prod_{g=1}^{r'+F} \{i-g\} \{j -g\}
\nn \\ \nn
&    \quad \times
 \prod_{t=1}^{r'+F} \frac{\{2r'-t+1\}\{k+t\}}{\{1\}}
 \hackcenter{ \begin{tikzpicture} [scale=.65, decoration={markings,
                        mark=at position 0.6 with {\arrow{>}};    }]
\draw[very thick, postaction={decorate}] (0,1.25) -- (0,.25);
\draw[very thick, -> ] (0,0) .. controls ++(0,.35) and ++(0,.5) .. (.5,-1.25);
\draw[very thick, ->] (0,0) .. controls ++(0,.35) and ++(0,.5) ..  (-.5,-1.25);
\node at (-.5,-1.5) { $\scs i$};
\node at (0,1.5) {$\scs k $};
\node at (.5,-1.5) {$\scs j$};
\node[draw, fill=white!20 ,rounded corners ] at (0,0) {$ -2F$};
\end{tikzpicture}}
\end{align}
\end{proposition}

\begin{proof}
Recall the definition of the map \eqref{defofYgen}.  Using Proposition \ref{Hdaggerbasecase} $r'+F$ times and \eqref{Yu},  we get that the left-hand side of the proposition is shown in Figure~\ref{fig:too-long}
\begin{figure*}
$\prod_{g=1}^{r'+F} \{i-g\} \{j -g \}
 \hackcenter{ \begin{tikzpicture} [scale=.7, decoration={markings,
                        mark=at position 0.6 with {\arrow{>}};    }]
\draw[very thick, postaction={decorate}] (0,1.25) -- (0,.25);
\draw[very thick, -> ] (0,0) .. controls ++(0,.35) and ++(0,.5) .. (.5,-1.25);
\draw[very thick, ->] (0,0) .. controls ++(0,.35) and ++(0,.5) ..  (-.5,-1.25);
\node at (-.8,-1.5) { $\scs i-r'-F$};
\node at (0,1.5) {$\scs k $};
\node at (.8,-1.5) {$\scs j-r'-F$};
\node[draw, fill=white!20 ,rounded corners ] at (0,0) {$ 2r'$};
\end{tikzpicture}}
\circ
\hackcenter{ \begin{tikzpicture} [scale=.7, decoration={markings,
                        mark=at position 0.6 with {\arrow{>}};    }]
\draw[very thick, postaction={decorate}] (0,1.25) -- (0,.25);
\draw[very thick, postaction={decorate}] (1.5,1.25) -- (1.5,.25);
\draw[very thick] (.2,-.25) .. controls ++(0,-.75) and ++(0,-.75) .. (1.3,-.25);
\draw[very thick, ->] (0,0) .. controls ++(0,.25) and ++(0,.5) ..  (-.5,-1.25);
\draw[very thick, ->] (1.5,0) .. controls ++(0,.25) and ++(0,.5) ..  (2,-1.25);
\node at (-.5,-1.5) { $\scs i-r'-F+1$};
\node at (-.2,1.5) {$\scs i-r'-F $};
\node at (1.7,1.5) {$\scs j-r'-F $};
\node at (2,-1.5) {$\scs j-r'-F+1$};
\node at (.75,-.8) {$\bullet$};
\node[draw, fill=white!20 ,rounded corners ] at (0,0) {$ -$};
\node[draw, fill=white!20 ,rounded corners ] at (1.5,0) {$ -$};
\end{tikzpicture}}
\circ \cdots \circ
\hackcenter{ \begin{tikzpicture} [scale=.7, decoration={markings,
                        mark=at position 0.6 with {\arrow{>}};    }]
\draw[very thick, postaction={decorate}] (0,1.25) -- (0,.25);
\draw[very thick, postaction={decorate}] (1.5,1.25) -- (1.5,.25);
\draw[very thick] (.2,-.25) .. controls ++(0,-.75) and ++(0,-.75) .. (1.3,-.25);
\draw[very thick, ->] (0,0) .. controls ++(0,.25) and ++(0,.5) ..  (-.5,-1.25);
\draw[very thick, ->] (1.5,0) .. controls ++(0,.25) and ++(0,.5) ..  (2,-1.25);
\node at (-.5,-1.5) { $\scs i$};
\node at (-.2,1.5) {$\scs i-1 $};
\node at (1.7,1.5) {$\scs j-1 $};
\node at (2,-1.5) {$\scs j$};
\node at (.75,-.8) {$\bullet$};
\node[draw, fill=white!20 ,rounded corners ] at (0,0) {$ -$};
\node[draw, fill=white!20 ,rounded corners ] at (1.5,0) {$ -$};
\end{tikzpicture}}$  \caption{A decomposition of the left-hand-side of Proposition~\ref{daggerYgen}.} \label{fig:too-long}
\end{figure*}
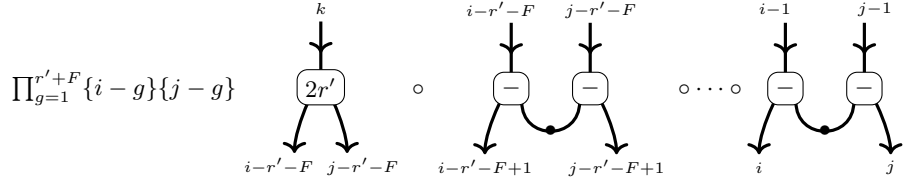
Applying \eqref{absorbinganH} $r'+F$ times yields the result.
\end{proof}

The next proposition yields the value of $\gamma(i,j,k)$.  Note that an explicit formula for $\beta(-k)$ is given in Proposition \ref{prop:beta-formula}.
\begin{proposition}
Assume $i+j+k=2F$.  Then
\begin{align}
& \left(
\hackcenter{ \begin{tikzpicture} [scale=.7, decoration={markings,
                        mark=at position 0.5 with {\arrow{>}};    }]
\draw[very thick, postaction={decorate}] (-1,1.25) .. controls ++(0,-.5) and ++(-.5,.25) ..   (0,.0);
\draw[very thick, postaction={decorate}] (1,1.25) .. controls ++(0,-.5) and ++(.5,.25) ..   (0,.0);
\draw[very thick,postaction={decorate}] (0,1.25) -- (0,0);
\node at (0,0) {$\bullet$};
\node at (-1,1.5) {$\scs i $};
\node at (0,1.5) {$\scs j $};
\node at (1,1.5) {$\scs k$};
\end{tikzpicture}}
\right)^{\dagger}
\;\; = \;\;
\beta(-k)
\prod_{g=1}^{r'+F} \{i-g\} \{j -g\}
\\
& \quad  \times
 \prod_{t=1}^{r'+F} \frac{\{2r'-t+1\}\{-k+t\}}{\{1\}}
\hackcenter{ \begin{tikzpicture} [scale=.7, decoration={markings,
                        mark=at position 0.5 with {\arrow{>}};    }]
\draw[very thick, postaction={decorate}] (0,.0).. controls ++(-.5,-.25)and ++(0,.5) .. (-1,-1.25);
\draw[very thick, postaction={decorate}] (0,.0).. controls ++(.5,-.25)and ++(0,.5) ..  (1,-1.25);
\draw[very thick,postaction={decorate}] (0,0)-- (0,-1.25);
\node at (0,0) {$\bullet$};
\node at (-1,-1.5) {$\scs i $};
\node at (0,-1.5) {$\scs j $};
\node at (1,-1.5) {$\scs k$};
\end{tikzpicture}} .
\end{align}
\end{proposition}

\begin{proof}
Taking the dagger of the right-hand side of \eqref{eq:h-ijk} using Proposition \ref{daggerYgen} and the definition of the dagger of the bullet (the map $w_{-k}$) from Section \ref{subsecbeta} yield the result.
\end{proof}

\subsection{Modified $6j$ symbols}
In order to reduce the clutter on the diagrams, we will abbreviate the morphism defined earlier as follows:
\begin{equation}
\hackcenter{ \begin{tikzpicture} [scale=.7, decoration={markings,
                        mark=at position 0.45 with {\arrow{>}};    }]
\draw[very thick, ->]   (0,-.25) --(0,-1.25) ;
\draw[very thick, postaction={decorate} ] (.5,1.25) .. controls ++(0,-.5) and ++(.25,.35) .. (.0,0);
\draw[very thick, postaction={decorate}]  (-.5,1.25).. controls ++(0,-.5) and ++(-.25,.35)   .. (0,0) ;
\node at (-.5,1.5) { $\scs i$};
\node at (0,-1.5) {$\scs k$};
\node at (.5,1.5) {$\scs j$};
\node[draw, fill=white!20 ,rounded corners ] at (0,0) {$ 2F$};
\end{tikzpicture}}
:=
\hackcenter{ \begin{tikzpicture} [scale=.7, decoration={markings,
                        mark=at position 0.45 with {\arrow{>}};    }]
\draw[very thick, ->]   (0,0) --(0,-1.25) ;
\draw[very thick, postaction={decorate} ] (.5,1.25) .. controls ++(0,-.5) and ++(.25,.35) .. (.0,0);
\draw[very thick, postaction={decorate}]  (-.5,1.25).. controls ++(0,-.5) and ++(-.25,.35)   .. (0,0) ;
\node at (-.5,1.5) { $\scs i$};
\node at (0,-1.5) {$\scs k$};
\node at (.5,1.5) {$\scs j$};
\end{tikzpicture}} .
\end{equation}

The category of representations being monoidal implies that there is an equality of morphisms given below.


\[
 \hackcenter{
\begin{tikzpicture}[yscale=-1, scale=0.7,  decoration={markings, mark=at position 0.6 with {\arrow{>}};},]
\draw[very thick, postaction={decorate}] (0,0) to [out=90, in=220] (.75,1);
\draw[very thick, postaction={decorate}] (1.5,-1) to [out=90, in=-30] (.75,1);
\draw[very thick,  postaction={decorate}] (.8,1) to  (.8,2);
\draw[very thick,  postaction={decorate}] (-.5,-1)to [out=90, in=210] (0,0);
\draw[very thick, postaction={decorate}] (.5,-1) to [out=90, in=-30] (0,0);
\node at (.75,-.8) {$\scs j$};
\node at (-.75,-.8) {$\scs i$};
\node at (1.75,-.8) {$\scs l$};
\node at (.5,1.6) {$\scs m$};
\node at (-.1,.6) {$\scs k$};
\end{tikzpicture} }
\quad = \quad
\sum_n
\md(n)
N^{i  j k}_{ l  m n} \;
\hackcenter{
\begin{tikzpicture}[yscale=-1.0, scale=0.7, decoration={markings, mark=at position 0.6 with {\arrow{>}};}]
\draw[very thick,  postaction={decorate}] (0,0) to [out=90, in=-30] (-.75,1);
\draw[very thick,  postaction={decorate}] (-1.5,-1) to [out=90, in=210] (-.75,1);
\draw[very thick,  postaction={decorate}] (-.8,1) to  (-.8,2);
\draw[very thick,  postaction={decorate}] (.5,-1)to [out=90, in=-40] (0,0);
\draw[very thick,  postaction={decorate}] (-.5,-1) to [out=90, in=210] (0,0);
\node at (-.75,-.8) {$\scs j$};
\node at (.75,-.8) {$\scs l$};
\node at (-1.75,-.8) {$\scs i$};
\node at (-.5,1.6) {$\scs m$};
\node at (.1,.6) {$\scs n$};
\end{tikzpicture} }
\]

\begin{remark}
We note that
\begin{equation}
N^{i  j k}_{ l  m n} \;\; := \;\;
 \left|
  \begin{array}{ccc}
    i & j & k \\
    l & m & n \\
  \end{array}
\right|
\end{equation}
in the notation from \cite{GPT2}.
Since all the vertices in the graphs represent specified maps in a $1$-dimensional morphism space, the $N^{i  j k}_{ l  m n}$ are some complex numbers whose formulas could be found in \cite{GP1}.
\end{remark}

%


Note that $N^{i  j k}_{ l  m n}$ could be defined as the value of the following network:
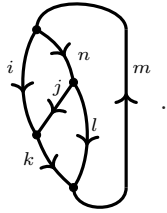
\[
\hackcenter{\begin{tikzpicture}[   decoration={markings, mark=at position 0.6 with {\arrow{>}};}, scale =0.7]
\draw[very thick,  postaction={decorate}, out=-30, in=100] (-.7,3) to (0,2);
\draw[very thick,  postaction={decorate}] (-.7,3) .. controls ++(-.35,-.5) and ++(-.35,.5) .. (-.7,1);
\draw[very thick,  postaction={decorate}, out=-70, in=140] (-.7,1) to (0,0);
\draw[very thick,  postaction={decorate}] (0,2) -- (-.7,1);
\draw[very thick,  postaction={decorate}] (0,2) .. controls ++(.35,-.5) and ++(.35,.5) ..  (0,0);
\draw[very thick, ] (0,0) .. controls ++(0,-.5) and ++(0,-.5) ..  (1,0);
\draw[very thick, , postaction={decorate}] (1,0) to (1,3);
\draw[very thick, ] (-.7,3) .. controls ++(0,.75) and ++(0,.55) ..  (1,3);
\draw[black,fill=black] (-.7,1) circle (.5ex);
\draw[black,fill=black](0,2) circle (.5ex);
\draw[black,fill=black] (-.7,3) circle (.5ex);
\draw[black,fill=black] (0,0) circle (.5ex);
    \node at (-1.2,2.3) {$\scs i$};
    \node at (-.3,1.9) {$\scs j$};
    \node at (-.85,.55) {$\scs k$};
    \node at (.4,1.2) {$\scs l$};
    \node at (.2,2.5) {$\scs n$};
    \node at (1.3,2.25) {$\scs m$};
\end{tikzpicture}} .
\]






In Sections \ref{subsecbeta} and \ref{gammasec} we defined and computed quantities $\beta$ and $\gamma$:
$$\beta:\R\setminus\Z\to\R\qquad\gamma:(\R\setminus\Z)^3\to\R .$$
This leads to a non involutive semilinear conjugation operator
$\Gamma\mapsto \Gamma^\ddagger$ on the space of $\R$-colored planar
uni-trivalent networks.  The image of a graph $\Gamma$ with state $\sigma$, with
set of trivalent vertices $\Gamma_0$, with set of univalent vertices
$\Gamma'_0$ and with set of edges $\Gamma_1$ is given by
\begin{equation}
  \label{eq:plan-dag}
  \Gamma^\ddagger=c(\Gamma)\mirror \Gamma
\end{equation}
where $\mirror \Gamma$ is the mirror image of $\Gamma$  with the opposite state and the complex coefficient $c(\Gamma)$ is given by
$$c(\Gamma)=\prod_{v\in \Gamma_0}\gamma(\sigma(v))\prod_{u\in \Gamma'_0}\beta(\sigma(u))^{-1}\prod_{e\in \Gamma_1}\beta(\sigma(e))$$
where $\sigma(u)$ is $\sigma$ of the unique edge adjacent to $u$ and
$\sigma(v)=(j_1,j_2,j_3)\in\R^3$ where $j_k=\sigma(e_k)$ and
$e_1,e_2,e_3$ are the $3$ ordered edges adjacent to $v$ oriented
toward $v$.
A network with no univalent vertices is closed.
\begin{remark}
Note that for an oriented edge $e$, $\beta(\sigma(-e))=\beta(\sigma(e))$.  Thus $\beta(\sigma(e))$ make sense even if we forget the orientation on $e$.
\end{remark}

The Poincar\'{e} dual of $\col$ is a $\Gr$-valued 1-cocycle on the 1-skeleton of $\T$ given on an oriented edge $a$ of $\T$ by $\col(a)=\col(e)$ where $e$ is the oriented edge of $\Gamma$ such that the algebraic intersection $a\cap e$ is $+1$.  We denote by $[\col]\in H^1(\Sigma,\Gr)$ the corresponding cohomology class.

\begin{proposition}
  If $\Gamma$ is a closed planar (or spherical) network, the
  evaluation of the network $\Gamma^\ddagger$ is the conjugate of the
  value of $\Gamma$
  \begin{equation}
    \label{eq:ddagger}
    \brk{\Gamma^\ddagger}=\overline{\brk{\Gamma}}.
  \end{equation}
\end{proposition}
\begin{proof}
 In \cite{CGP1} it is shown that the graph evaluation arises from the representation theory of the unrolled quantum group for $\slt$, where the evaluation of a closed graph is determined from a modified trace applied to a cutting of the graph along a projective object.  The claim then follows from \cite{GLPMS} proving that this category is Hermitian.
\end{proof}

The next result is the goal of the section and is a key ingredient in proving the Hermiticity of the Hamiltonian for the example coming from the semi-restricted quantum group for $\slt$.
\begin{proposition} \label{prop:herm-6jgraph}
The following identities hold: \hfill
\begin{enumerate}
 \item \label{prop:gamma-theta}$\gamma(i,j,k)\gamma(k^{\ast},j^{\ast},i^{\ast})\beta(i)\beta(j)\beta(k)=1$;

 \item 
 $\wb{(N^{j_1 j_2 j_3 }_{j_4j_5j_6})}{=}N^{j_2^{\ast} j_1^{\ast} j_3^{\ast}}_{j_5 j_4j_6}\gamma(j_1,j_2,j_3^{\ast})\gamma(j_1^{\ast},j_5,j_6^{\ast})$
\hspace{30 mm} $\times \gamma(j_2^{\ast},j_6,j_4^{\ast}) \gamma(j_3,j_4,j_5^{\ast})\prod_i\beta(j_i)$. \hfill
\end{enumerate}
\end{proposition}
\begin{proof}
  The first identity comes from
  $\brk{\Theta^\ddagger}=\overline{\brk{\Theta}}=1$ where $\Theta$ is
  the network
  \[
    \Theta=\hackcenter
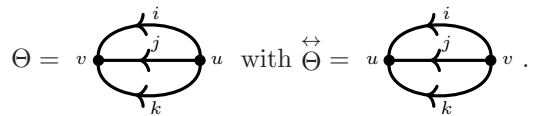
{
      \begin{tikzpicture}[   decoration={markings, mark=at position 0.6 with {\arrow{>}};}, scale =0.9]
        \draw[very thick,  postaction={decorate}] (1.5,0).. controls ++(.1,.7) and ++(-.1,.7) .. (0,0);
        \draw[very thick,  postaction={decorate}] (1.5,0).. controls ++(.1,-.7) and ++(-.1,-.7) .. (0,0);
        \draw[very thick,  postaction={decorate}] (1.5,0) to  (0,0);
        \draw[black,fill=black] (0,0) circle (.5ex);
        \draw[black,fill=black](1.5,0) circle (.5ex);
        \node at (.85,.7) {$\scs i$};
        \node at (.85,.25) {$\scs j$};
        \node at (.85,-.7) {$\scs k$};
        \node at (-.25,0) {$\scs v$};
        \node at (1.75,0) {$\scs u$};
      \end{tikzpicture}}\text{ with }\mirror\Theta=\hackcenter{
      \begin{tikzpicture}[   decoration={markings, mark=at position 0.6 with {\arrow{>}};}, scale =0.9]
        \draw[very thick,  postaction={decorate}] (1.5,0).. controls ++(.1,.7) and ++(-.1,.7) .. (0,0);
        \draw[very thick,  postaction={decorate}] (1.5,0).. controls ++(.1,-.7) and ++(-.1,-.7) .. (0,0);
        \draw[very thick,  postaction={decorate}] (1.5,0) to  (0,0);
        \draw[black,fill=black] (0,0) circle (.5ex);
        \draw[black,fill=black](1.5,0) circle (.5ex);
        \node at (.85,.7) {$\scs i$};
        \node at (.85,.25) {$\scs j$};
        \node at (.85,-.7) {$\scs k$};
        \node at (-.25,0) {$\scs u$};
        \node at (1.75,0) {$\scs v$};
      \end{tikzpicture}}.
  \]
  For the second identity we consider the network $\Gamma$ defining
  the 6j-symbol $N^{j_1 j_2 j_3 }_{j_4j_5j_6}$.  Then we have
  \[
    \mirror\Gamma=
    \hackcenter
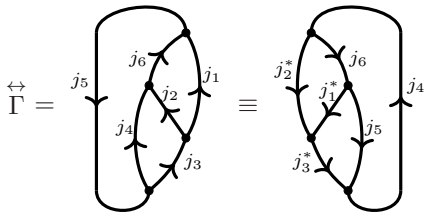
{\begin{tikzpicture}[   decoration={markings, mark=at position 0.6 with {\arrow{<}};}, scale =0.7]
        \draw[very thick,  postaction={decorate}, out=-150, in=80] (.7,3) to (0,2);
        \draw[very thick,  postaction={decorate}] (.7,3) .. controls ++(.35,-.5) and ++(.35,.5) .. (.7,1);
        \draw[very thick,  postaction={decorate}, out=-110, in=40] (.7,1) to (0,0);
        \draw[very thick,  postaction={decorate}] (0,2) -- (.7,1);
        \draw[very thick,  postaction={decorate}] (0,2) .. controls ++(-.35,-.5) and ++(-.35,.5) ..  (0,0);
        \draw[very thick, ] (0,0) .. controls ++(0,-.5) and ++(0,-.5) ..  (-1,0);
        \draw[very thick,  postaction={decorate}] (-1,0) to (-1,3);
        \draw[very thick, ] (.7,3) .. controls ++(0,.75) and ++(0,.55) ..  (-1,3);
        \draw[black,fill=black] (.7,1) circle (.5ex);
        \draw[black,fill=black](0,2) circle (.5ex);
        \draw[black,fill=black] (.7,3) circle (.5ex);
        \draw[black,fill=black] (0,0) circle (.5ex);
        \node at (1.2,2.2) {$\scs j_1$};
        \node at (.4,1.85) {$\scs j_2$};
        \node at (.85,.55) {$\scs j_3$};
        \node at (-.45,1.2) {$\scs j_4$};
        \node at (-.2,2.5) {$\scs j_6$};
        \node at (-1.3,2.1) {$\scs j_5$};
      \end{tikzpicture}}
    \equiv
    \hackcenter{\begin{tikzpicture}[   decoration={markings, mark=at position 0.6 with {\arrow{>}};}, scale =0.7]
        \draw[very thick,  postaction={decorate}, out=-30, in=100] (-.7,3) to (0,2);
        \draw[very thick,  postaction={decorate}] (-.7,3) .. controls ++(-.35,-.5) and ++(-.35,.5) .. (-.7,1);
        \draw[very thick,  postaction={decorate}, out=-70, in=140] (-.7,1) to (0,0);
        \draw[very thick,  postaction={decorate}] (0,2) -- (-.7,1);
        \draw[very thick,  postaction={decorate}] (0,2) .. controls ++(.35,-.5) and ++(.35,.5) ..  (0,0);
        \draw[very thick, ] (0,0) .. controls ++(0,-.5) and ++(0,-.5) ..  (1,0);
        \draw[very thick,  postaction={decorate}] (1,0) to (1,3);
        \draw[very thick, ] (-.7,3) .. controls ++(0,.75) and ++(0,.55) ..  (1,3);
        \draw[black,fill=black] (-.7,1) circle (.5ex);
        \draw[black,fill=black](0,2) circle (.5ex);
        \draw[black,fill=black] (-.7,3) circle (.5ex);
        \draw[black,fill=black] (0,0) circle (.5ex);
        \node at (-1.2,2.3) {$\scs j_2^{\ast}$};
        \node at (-.35,1.9) {$\scs j_1^{\ast}$};
        \node at (-.85,.55) {$\scs j_3^{\ast}$};
        \node at (.5,1.2) {$\scs j_5$};
        \node at (.2,2.5) {$\scs j_6$};
        \node at (1.3,2) {$\scs j_4$};
      \end{tikzpicture}}
\]
where $\equiv$ is an isotopy of the graph in the sphere.  This last
graph evaluates to $N^{j_2^{\ast} j_1^{\ast} j_3^{\ast}}_{j_5 j_4j_6}$. So
equation \eqref{eq:ddagger} applied to $\Gamma^\ddagger$ implies the
proposition.
\end{proof}

\vspace{.8in}

%

\end{document}